\documentclass[conference]{IEEEtran}
 \pdfoutput=1 

 \usepackage{cite}
 \usepackage[standard]{ntheorem}
\renewenvironment{proof}{\begin{IEEEproof}}{\end{IEEEproof}}
\newif\ifconf
\conftrue
\conffalse

\ifconf
  \newcommand{\short}[1]{#1}
  \newcommand{\full}[1]{}
  \usepackage{environ}
  \NewEnviron{killcontents}{}

\else
  \newcommand{\short}[1]{}
  \newcommand{\full}[1]{{#1}}
\fi

\usepackage{etex}

\usepackage[bookmarks=false]{hyperref}

\usepackage{proof}

\newcommand{\wi}{{interchangeable}}

\newlength{\myheight}

\usepackage{amssymb,graphicx,epsfig,color}
\usepackage{subcaption}
\usepackage[arrow,matrix,frame,curve,cmtip,color]{xy}
\newdir{  >}{*{}!/-.9em/\dir{>}}
\newdir{ >}{*{}!/-.4em/\dir{>}}
\usepackage{wrapfig}

\usepackage{pgf}
\usepackage{tikz}
\usetikzlibrary{arrows,backgrounds,patterns,positioning}
\tikzstyle{node}=[circle, draw=black, minimum size=1mm]
\tikzstyle{trans}=[font=\scriptsize]
\tikzstyle{lab}=[font=\small]

\usepackage{tikz-cd}

\newcommand{\pgfBox}{
  \begin{pgfonlayer}{background} 
    \fill[blue!2,thick,draw=black!50,rounded corners,inner sep=3mm] ([xshift=-1.5pt]current bounding box.south west) rectangle ([xshift=1.5pt]current bounding box.north east);
  \end{pgfonlayer}
}

\newcommand{\bx}[1]{\phantom{\big(}#1{\phantom{\big)}}} %

\newcommand{\Rrel}[1]   {\stackrel{{#1}}{\Longrightarrow}}

\usepackage{stmaryrd}

\usepackage[latin1]{inputenc}
\usepackage{amsmath}
\usepackage{enumerate}
\usepackage{xspace}
\usepackage{amsfonts}
\usepackage{mathrsfs}
\usepackage{cite}
\usepackage{float}
\usepackage{fancybox}

\newcommand{\compact}[1]{\ensuremath{\mathop{\mathsf{K}({#1})}}}

\newcommand{\principal}[1]{\ensuremath{\mathop{\downarrow\!{#1}}}}

\newcommand{\ideal}[1]{\ensuremath{\mathsf{Idl}({#1})}}

\newcommand{\pr}[1]{\ensuremath{\mathop{\mathit{pr}({#1})}}}
\newcommand{\wpr}[1]{\ensuremath{\mathop{\mathit{wpr}({#1})}}}

\newcommand{\ir}[1]{\ensuremath{\mathop{\mathit{ir}({#1})}}}

\newcommand{\diff}[2]{\ensuremath{\delta({#1},{#2})}}

\newcommand{\esabbr}{\textsc{es}}
\newcommand{\esnabbr}{\textsc{esnb}}
\newcommand{\esnmabbr}{\textsc{esn}}
\newcommand{\eseqabbr}{\textsc{epes}}

\newcommand{\pred}[1]{\ensuremath{\mathit{p}({#1})}}

\newcommand{\esir}[2]{\ensuremath{\langle{#1}, {#2}\rangle}}

\newcommand{\eqclass}[2][]{\ensuremath{[{#2}]_{\scriptscriptstyle {#1}}}}

\newcommand{\eqclassir}[1]{\ensuremath{\eqclass[\leftrightarrow^*]{#1}}}

\newcommand{\quotient}[2]{\ensuremath{{#1}_{\scriptscriptstyle {#2}}}}

\newcommand{\es}{\ensuremath{\mathsf{ES}}}
\newcommand{\ses}{\ensuremath{\mathsf{sES}}}
\newcommand{\pes}{\ensuremath{\mathsf{pES}}}
\newcommand{\epes}{\ensuremath{\mathsf{epES}}}

\newcommand{\ces}{\ensuremath{\mathsf{cES}}}

\newcommand{\WDom}{\ensuremath{\mathsf{wDom}}}
\newcommand{\PDom}{\ensuremath{\mathsf{pDom}}}

\newcommand{\esn}{\ensuremath{\mathsf{ES_{nb}}}}
\newcommand{\cesn}{\ensuremath{\mathsf{cES_{nb}}}}

\newcommand{\WDomb}{\ensuremath{\mathsf{wDom_b}}}
\newcommand{\Domb}{\ensuremath{\mathsf{Dom_b}}}
\newcommand{\slice}[2]{\ensuremath{({#1} \downarrow {#2})}}

\newcommand{\zev}[0]{\ensuremath{\mathcal{E}}}
\newcommand{\ev}[1]{\ensuremath{\zev({#1})}}

\newcommand{\zconnes}[0]{\ensuremath{\mathcal{C}}}

\newcommand{\zinces}[0]{\ensuremath{\mathcal{I}}}

\newcommand{\zsev}[0]{\ensuremath{\mathcal{E}_S}}
\newcommand{\sev}[1]{\ensuremath{\zsev({#1})}}

\newcommand{\zeveq}[0]{\ensuremath{\mathcal{E}_{eq}}}
\newcommand{\eveq}[1]{\ensuremath{\zeveq({#1})}}

\newcommand{\zfuse}[0]{\ensuremath{\mathcal{M}}}
\newcommand{\fuse}[1]{\ensuremath{\zfuse({#1})}}
\newcommand{\zunf}[0]{\ensuremath{\zunf}}
\newcommand{\unf}[1]{\ensuremath{\mathcal{U}({#1})}}

\newcommand{\zevwd}[0]{\ensuremath{\mathcal{E}_{wd}}}
\newcommand{\evwd}[1]{\ensuremath{\zevwd({#1})}}

\newcommand{\conf}[1]{\ensuremath{\mathit{Conf}({#1})}}
\newcommand{\conff}[1]{\ensuremath{\mathit{Conf_F}({#1})}}

\newcommand{\pmin}[1]{\ensuremath{U_{#1}}}

\newcommand{\conn}[1]{\ensuremath{\stackrel{#1}{\frown}}}

\newcommand{\zdom}[0]{\ensuremath{\mathcal{D}}}
\newcommand{\dom}[1]{\ensuremath{\zdom({#1})}}

\newcommand{\zdomeq}[0]{\ensuremath{\mathcal{D}_{eq}}}
\newcommand{\domeq}[1]{\ensuremath{\zdomeq({#1})}}

\newcommand{\poset}[1]{\ensuremath{\mathcal{P}({#1})}}

\newcommand{\pdom}[1]{\ensuremath{\mathcal{D}_S({#1})}}
\newcommand{\ppdom}[0]{\ensuremath{\mathcal{D}_S}}

\newcommand{\Pow}[1]{\ensuremath{\mathbf{2}^{#1}}}

\newcommand{\Powfin}[1]{\ensuremath{\mathbf{2}_\mathit{fin}^{#1}}}

\newcommand{\interval}[2][1]{\ensuremath{[{#1},{#2}]}}

\newcommand{\dint}[2]{\ensuremath{[{#1},{#2}]}}

\newcommand{\IntSet}[1]{\ensuremath{\mathop{\mathit{Int}({#1})}}}

\newcommand{\inir}{\ensuremath{\mathop{\mathit{\zeta}}}}
\newcommand{\irin}{\ensuremath{\mathop{\mathit{\iota}}}}

\newcommand{\perm}{\sigma}

\newcommand\twoheaddownarrow{\mathrel{\raisebox{0.8\depth}{\rotatebox{270}{$\twoheadrightarrow$}}}}
\newcommand{\scauses}[1]{\ensuremath{\twoheaddownarrow\!\!{#1}\,}}
\newcommand{\causes}[1]{\ensuremath{\,\downarrow\!\!{#1}}}

\newcommand{\sat}[1]{\ensuremath{\tilde{#1}}}

\newcommand{\tr}[1]{\ensuremath{\mathsf{Tr}({#1})}}
\newcommand{\trs}[1]{\ensuremath{\mathsf{Tr}_s({#1})}}
\newcommand{\graph}{\ensuremath{\mathsf{Graph}}}
\newcommand{\tgraph}[1]{\ensuremath{\mathsf{Graph}_{#1}}}
\newcommand{\can}[1]{\ensuremath{\mathsf{C}({#1})}}
\newcommand{\source}[1]{\ensuremath{\mathsf{s}({#1})}}
\newcommand{\target}[1]{\ensuremath{\mathsf{t}({#1})}}
\newcommand{\col}[1]{\ensuremath{\mathsf{col}({#1})}}

\newcommand{\ltrace}[1]{\ensuremath{\langle {#1}\rangle_c}}

\title{Domains and Event Structures for Fusions}
\author{\IEEEauthorblockN{Paolo Baldan}\IEEEauthorblockA{University of Padova}\IEEEauthorblockN{Andrea Corradini, Fabio Gadducci}\IEEEauthorblockA{University of Pisa}}

\begin{document}

\maketitle

\begin{abstract}
  Stable event structures, and their duality with prime algebraic domains
  arising as partial orders of configurations, are a landmark
  of concurrency theory, providing a clear characterisation of
  causality in computations.
  They have been used for defining a concurrent semantics of several
  formalisms, from Petri nets to (linear) graph rewriting systems,
  which in turn lay at the basis of many visual modelling frameworks.
  Stability however is restrictive when dealing with formalisms
  with ``fusion'', i.e., where a computational step can not only consume
  and produce but also merge parts of the state. This happens, e.g.,
  for graph rewriting systems with non-linear rules, which are needed
  to cover some relevant applications (such as the graphical encoding
  of calculi with name passing).
  Guided by the need of capturing the semantics of formalisms with
  fusion we leave aside stability and we characterise, as a natural
  generalisation of prime algebraic domains, a class of domains,
  referred to as weak prime domains.
  We then identify a corresponding class of event structures, that we call 
  connected event structures, via a duality result formalised as an equivalence
  of categories.
  We show that connected event structures are exactly the class of
  event structures that arise as the semantics of non-linear graph
  rewriting systems.
  Interestingly, the category of general unstable event structures
  coreflects into our category of {weak prime} domains, so that our
  result provides a characterisation of the partial orders of
  configurations of such event structures.

  \begin{IEEEkeywords}
    Event structures, fusions, graph rewriting, process calculi.
  \end{IEEEkeywords}
\end{abstract}

\section{Introduction}

For a long time stable/prime event structures and their duality with
prime algebraic domains have been considered one of the landmarks of
concurrency theory, providing a clear characterisation of causality in
software systems. They have been used to provide a concurrent
semantics to a wide range of foundational formalisms, from Petri
nets~\cite{NPW:PNES} to linear graph rewriting
systems~\cite{Handbook,Bal:PhD,Sch:RRSG} and process
calculi~\cite{Win:ESSCCS,VY:TESLP,BMM:ESSNC}. They are one of the
standard tools for the formal treatment of (true, i.e.,
non-interleaving) concurrency. See, e.g.,~\cite{Winskel11} for a
reasoned survey on the use of such causal models. Recently, they have
been used in the study of concurrency in weak memory
models~\cite{PS:CSRA,JR:OTRES,CV:GTARES} and for process mining and
differencing~\cite{DG:PMRES}.

In order to endow a chosen formalism with an event structure
semantics, a standard construction consists in viewing the class of
computations as a partial order. An element of the order is some sort
of configuration, i.e., an execution trace up to an equivalence that
identifies traces differing only for the order of independent steps
(e.g., interchange law~\cite{Mes92} in term rewriting, shift
equivalence~\cite{CMREHL:AAGT} in graph rewriting or permutation
equivalence~\cite{JJL80} in $\lambda$-calculus),
and the order relates two computations when
the latter is an extension of the former.
Events are then identified with configurations consisting
of a maximal computation step (e.g., a transition of a CCS process or
a firing for a Petri net) with all its causes.
As a simple example, consider the CCS process $a.c \mid b$. 
The corresponding transition system is depicted in Fig.~\ref{fi:ccs-ts}. We can identify the states of the computation with the sets of actions executed and obtain the partial order depicted in Fig.~\ref{fi:ccs-domain}.
The fact that each {computation step}  in a configuration has a uniquely determined
set of causes, a property that for event structures is called
\emph{stability}, allows one to characterise such elements, order
theoretically, as the prime elements: if they are included in a join
they must be included in one of the elements that are joined.
{For example}, in Fig.~\ref{fi:ccs-domain}, the events correspond to configurations $\{a\}$ (transition $a$ with
empty set of causes), $\{a,c\}$ (transition $c$ caused by $a$) and
$\{b\}$ (transition $b$ with empty set of causes).
Each element of the partial order of configurations can be
reconstructed uniquely as the join of the primes under it, so that the partial
order is prime algebraic.
This duality between event structures and domains of configurations 
can be nicely formalised in terms of an equivalence between the
categories of prime event structures and prime algebraic 
domains~\cite{NPW:PNES,Win:ES}.

\begin{figure}
  \subcaptionbox{
    \label{fi:ccs-ts}
  }{
    \begin{tikzpicture}[node distance=5mm, >=stealth',x=17mm,y=5mm]
      \node at (1,0) (s)  {$a.c \mid b$};
      \node at (0,1) (a)  {$c \mid b$};
      \node at (2,1) (c)  {$a.c \mid 0$};
      \node at (-1,2) (ab) {$0 \mid b$};
      \node at (1,2) (ac) {$c \mid 0$};
      \node at (0,3) (abc) {$0 \mid 0$};
      \draw [->] (s) -- (a);
      \draw [->] (s) -- (c);
      \draw [->] (a) -- (ab);
      \draw [->] (a) -- (ac);
      \draw [->] (c) -- (ac);
      \draw [->] (ab) -- (abc);
      \draw [->] (ac) -- (abc);
    \end{tikzpicture}
  }
  \subcaptionbox{
    \label{fi:ccs-domain}
  }{
    \begin{tikzpicture}[node distance=5mm, >=stealth',x=17mm,y=5mm]
      \node at (1,0) (s)  {$\emptyset$};
      \node at (0,1) (a)  {$\{a\}$};
      \node at (2,1) (c)  {$\{b\}$};
      \node at (-1,2) (ab) {$\{a,c\}$};
      \node at (1,2) (ac) {$\{a,b\}$};
      \node at (0,3) (abc) {$\{a,b,c\}$};
      \draw [->] (s) -- (a);
      \draw [->] (s) -- (c);
      \draw [->] (a) -- (ab);
      \draw [->] (a) -- (ac);
      \draw [->] (c) -- (ac);
      \draw [->] (ab) -- (abc);
      \draw [->] (ac) -- (abc);
    \end{tikzpicture}
  }
  \caption{The (a) transition system and (b) domain of configurations
    of the process $a.c \mid b$.}
  \label{fi:ccs}
\end{figure}

The set up described so far fails when moving to formalisms where a
computational step can merge parts of the state. This happens, e.g.,
in nominal calculi where, as a result of name passing,
the received name is identified with a local one at the
receiver~\cite{CVY:ESSPE,Gad07} or in the modelling of bonding in biological/chemical processes~\cite{PUY:MBPE}.
Whenever we think of the state of the system as some kind of graph
with the dynamics described by graph rewriting, this means that rewriting rules
are non-linear (more precisely, in the jargon of the double pushout
approach~\cite{Ehr:TIAA}, left-linear but possibly not right-linear).
In general terms, the point is that, in the presence of fusions, the
same event can be enabled by different minimal sets of events, thus
preventing the identification of a proper notion of causality.

\begin{figure}[h!]
  \begin{center}
  \subcaptionbox{The start graph $G_s$ and the rules $p_y$ ($y \in \{ a, b\}$) and $p_c$.
    \label{fi:running-gg}
  }{
    \begin{tikzpicture}[node distance=6mm, >=stealth',baseline=(current bounding box.center)]
      \node at (0,0) [node, label=below:$c$] (nc) {} 
      edge [in=160, out=130, loop]  node [lab,above] {$\bar{b}$} ()
      edge [in=105, out=75, loop]  node [lab,above] {$\bar{a}$} ()
      edge [in=50, out=20, loop]  node [lab,above] {$\mathit{in}$} ();
      \node at (1,0) [node, label=below:$\nu$] (nu) {} 
      edge [in=105, out=75, loop]  node [lab,above] {$\bar{\nu}$} ();
      \pgfBox
    \end{tikzpicture}
    \
    \hfill
    \
    \begin{tikzpicture}[node distance=3mm, baseline=(current bounding box.center)]
      \node (l) {
      \begin{tikzpicture}[node distance=8mm, >=stealth']
      \node at (0,0) [node, label=below:$c$] (nc) {} 
      edge [in=105, out=75, loop]  node [lab,above] {$\bar{y}$} ();
      \node at (.5,0) [node, label=below:$\nu$] (nu) {} 
      edge [in=105, out=75, loop]  node [lab,above] {$\bar{\nu}$} ();
      \pgfBox
      \end{tikzpicture} 
    };
    \node  [right=of l] (r) {
      \begin{tikzpicture}[node distance=8mm, >=stealth']
        \node at (0,0) [node, label=below:${c,\nu}$] (nc) {} 
        edge [in=105, out=75, loop]  node [lab,above] {$\bar{\nu}$} ();
        \pgfBox
      \end{tikzpicture}
    };
    \path (l) edge[->] node[trans, above] {$p_y$} (r);
    \end{tikzpicture}
    \
    \hfill
    \
    \begin{tikzpicture}[node distance=3mm, baseline=(current bounding box.center)]
      \node (l) {
      \begin{tikzpicture}[node distance=8mm, >=stealth']
        \node at (0,0) [node, label=below:${c,\nu}$] (nc) {} 
        edge [in=160, out=130, loop]  node [lab,above] {$\mathit{in}$} ()
        edge [in=105, out=75, loop]  node [lab,above] {$\bar{\nu}$} ();
        \pgfBox
      \end{tikzpicture} 
    };
    \node  [right=of l] (r) {
      \begin{tikzpicture}[node distance=8mm, >=stealth']
        \node at (0,0) [node, label=below:${c,\nu}$] (nc) {} 
        edge [in=105, out=75, loop]  node [lab,above] {$\bar{\nu}$} ();
        \pgfBox
      \end{tikzpicture}
    };
    \path (l) edge[->] node[trans, above] {$p_c$} (r);
  \end{tikzpicture}  
  }
  \end{center}
  \begin{center}
  \subcaptionbox{The possible rewrites.
    \label{fi:running-rewriting}
  }
  {
    \begin{tikzpicture}[node distance=5mm, >=stealth']
      \node (Gs) {$G_s$};
      \node [yshift=1.2cm, right=of Gs] (Ga) {
        \begin{tikzpicture}[node distance=8mm, >=stealth']
          \node  [node, label=below:${c,v}$] (nc) {} 
          edge [in=160, out=130, loop]  node [lab,above] {$\bar{a}$} ()
          edge [in=105, out=75, loop]  node [lab,above] {$\bar{\nu}$} ()
          edge [in=50, out=20, loop]  node [lab,above] {$\mathit{in}$} ();
          \pgfBox
        \end{tikzpicture}
      };
      \node [lab, xshift=-2mm, yshift=4mm] at (Ga.west) {$G_a$};
      \node [yshift=-1.2cm, right=of Gs] (Gb) {
        \begin{tikzpicture}[node distance=8mm, >=stealth']
          \node at (0,0) [node, label=below:${c,v}$] (nc) {} 
          edge [in=160, out=130, loop]  node [lab,above] {$\bar{b}$} ()
          edge [in=105, out=75, loop]  node [lab,above] {$\bar{\nu}$} ()
          edge [in=50, out=20, loop]  node [lab,above] {$\mathit{in}$} ();
          \pgfBox
        \end{tikzpicture}
      };
     \node [lab, xshift=-2mm, yshift=-4mm] at (Gb.west) {$G_b$};
      \node [yshift=-1.2cm, right=of Ga] (Gab) {
        \begin{tikzpicture}[node distance=8mm, >=stealth']
          \node at (0,0) [node, label=below:${c,\nu}$] (nc) {} 
          edge [in=105, out=75, loop]  node [lab,above] {$\bar{\nu}$} ()
          edge [in=50, out=20, loop]  node [lab,above] {$\mathit{in}$} ();
          \pgfBox
        \end{tikzpicture}
      };
      \node [lab, xshift=-2mm, yshift=0mm] at (Gab.west) {{$G_{ab}$}};
      \node [yshift=1.2cm, right=of Gab] (Gac) {
        \begin{tikzpicture}[node distance=8mm, >=stealth']
          \node at (0,0) [node, label=below:${c,\nu}$] (nc) {} 
          edge [in=160, out=130, loop]  node [lab,above] {$\bar{a}$} ()
          edge [in=105, out=75, loop]  node [lab,above] {$\bar{\nu}$} ();
          \pgfBox
        \end{tikzpicture}
      };
      \node [lab, xshift=-2mm, yshift=4mm] at (Gac.west) {$G_{ac}$};
      \node [yshift=-1.2cm, right=of Gab] (Gbc) {
        \begin{tikzpicture}[node distance=8mm, >=stealth']
          \node at (0,0) [node, label=below:${c,\nu}$] (nc) {} 
          edge [in=160, out=130, loop]  node [lab,above] {$\bar{b}$} ()
          edge [in=105, out=75, loop]  node [lab,above] {$\bar{\nu}$} ();
          \pgfBox
        \end{tikzpicture}
      };
      \node [lab, xshift=-2mm, yshift=-4mm] at (Gbc.west) {$G_{bc}$};
      \node [yshift=-1.2cm, right=of Gac] (Gc) {
        \begin{tikzpicture}[node distance=8mm, >=stealth']
          \node at (0,0) [node, label=below:${c,\nu}$] (nc) {} 
          edge [in=105, out=75, loop]  node [lab,above] {$\bar{\nu}$} ();
          \pgfBox
        \end{tikzpicture}
      };
      \node [lab, xshift=0mm, yshift=2mm] at (Gc.north) {$G_{c}$};
            
      \path (Gs) edge [->] node[trans,pos=0.3,below] {$p_a$} (Gb)
                 edge [->] node[trans,pos=0.3,above] {$p_b$} (Ga);
      \path (Ga) edge [->] node[trans,pos=0.6,above] {$p_a$} (Gab) 
                 edge [->] node[trans,above] {$p_c$} (Gac);
      \path (Gb) edge [->] node[trans,below] {$p_b$} (Gab) 
                 edge [->] node[trans,below] {$p_c$} (Gbc);
      \path (Gab) edge[->] node[trans,pos=0.15,above] {$p_c$} (Gc);
      \path (Gac) edge[->] node[trans,above] {$p_a$} (Gc);
      \path (Gbc) edge[->] node[trans,below] {$p_b$} (Gc);
    \end{tikzpicture}
  }
  \end{center}
  \begin{center}
  \subcaptionbox{The domain of configurations.
    \label{fi:running-configurations}
  }
  {  
    \begin{tikzpicture}[node distance=5mm, >=stealth',x=23mm,y=6mm]
      \node at (1,0) (Gs)  {$\emptyset$};
      \node at (0,1) (Ga)  {$\{a\}$};
      \node at (2,1) (Gb)  {$\{b\}$};
      \node at (1,2.1) (Gab) {$\{a,b\}$};
      \node at (0,2.7) (Gac)  {$\{a,c\}$};
      \node at (2,2.7) (Gbc)  {$\{b,c\}$};
      \node at (1,4) (Gc)  {$\{a,b,c\}$};
      \draw [->] (Gs) -- (Ga);
      \draw [->] (Gs) -- (Gb);
      \draw [->] (Ga) -- (Gab);
      \draw [->] (Gb) -- (Gab);
      \draw [->] (Ga) -- (Gac);
      \draw [->] (Gb) -- (Gbc);
      \draw [->] (Gac) -- (Gc);
      \draw [->] (Gbc) -- (Gc);
      \draw [->] (Gab) -- (Gc);
    \end{tikzpicture}
    }
  \end{center}
  \caption{A graph rewriting system with fusions.}
  \label{fi:running}
\end{figure}

\begin{figure}
{\begin{center}
{
    \begin{tikzpicture}[node distance=2mm, >=stealth']
       \node (Gs) {$(\nu c)(\bar{a}(c) \mid \bar{b}(c) \mid
c())$};
        \node [yshift=1.2cm, right=-9mm of Gs] (Ga) {$\bar{a}(c) \mid
c()$};
        \node [yshift=-1.2cm, right=-9mmof Gs] (Gb) {$\bar{b}(c) \mid
c()$};
 	\node [yshift=-1.2cm, right=of Ga] (Gab) {$c()$}; 
        \node [yshift=1.2cm, right=of Gab] (Gac) {$\bar{a}(c)$};      %
      \node [yshift=-1.2cm, right=of Gab] (Gbc) {$\bar{b}(c)$};      %
     \node [yshift=-1.2cm, right=of Gac] (Gc) {$0$};    
      \path (Gs) edge [->] node[trans,pos=0.1,below] {$\bar{a}(c)$} (Gb)
                 edge [->] node[trans,pos=0.1,above] {$\bar{b}(c)$} (Ga);
      \path (Ga) edge [->] node[trans,pos=0.7,xshift=0.2cm, above] {$\bar{a}(c)$} (Gab) 
                 edge [->] node[trans,above] {$c()$} (Gac);
      \path (Gb) edge [->] node[trans,pos=0.7, xshift=0.2cm, below] {$\bar{b}(c)$} (Gab) 
                 edge [->] node[trans,below] {$c()$} (Gbc);
      \path (Gab) edge[->] node[trans,pos=0.5,above] {$c()$} (Gc);
      \path (Gac) edge[->] node[trans,pos=0.7, xshift=0.2cm, above] {$\bar{a}(c)$} (Gc);
      \path (Gbc) edge[->] node[trans,pos=0.7, xshift=0.2cm, below] {$\bar{b}(c)$} (Gc);
    \end{tikzpicture}
  }
  \end{center}
  }
\caption{The possible transitions of the $\pi$-calculus process $(\nu c)(\bar{a}(c) \mid \bar{b}(c) \mid c())$.}
  \label{fi:pi}
\end{figure}

As an example, consider the graph rewriting system in
Fig.~\ref{fi:running}.
The start graph $G_s$ and the rewriting rules $p_a$, $p_b$, and $p_c$
are reported in Fig.~\ref{fi:running-gg}.  Observe that rules $p_y$,
where $y$ can be either $a$ or $b$, delete edge $\bar{y}$ and merge
nodes $c$ and $\nu$. The possible rewrites are depicted in
Fig.~\ref{fi:running-rewriting}. For instance, applying $p_a$ to $G_s$
we get the graph $G_b$. Now, $p_b$ can still be applied to $G_b$
matching its left-hand side non-injectively, thus getting graph
$G_{ab}$. Similarly, we can apply first $p_b$ and then $p_a$,
obtaining again $G_{ab}$. Observe that at least one between $p_a$ 
and $p_b$ must be applied to enable $p_c$, since the latter rule requires 
nodes $c$ and $\nu$ to be merged.
{Note also} that in a situation where all the three rules $p_a$, $p_b$, and
$p_c$ are applied, since $p_a$ and $p_b$ are independent,
it is not possible to define a proper notion of causality. We only
know that at least one between $p_a$ and $p_b$ must be applied  
before $p_c$.
The corresponding domain of configurations, reported in
Fig.~\ref{fi:running-configurations}, is naturally derived from the
possible rewrites in Fig.~\ref{fi:running-rewriting}. 

The graph rewriting system of Fig.~\ref{fi:running-gg} is a (simplified) representation of the
$\pi$-calculus process $(\nu c)(\bar{a}(c) \mid \bar{b}(c) \mid
c())$. Rules $p_y$, for $y \in \{a,b\}$, represent the execution of
$\bar{y}(c)$ that outputs on channel $y$ the restricted name $c$. The
first rule that is executed extrudes name $c$, while the second is just a standard
output. The name $c$ is available outside the scope only after the extrusion, and 
after that the input prefix $c()$ can be consumed. {Figure~\ref{fi:pi} shows the possible transitions of the process, which correspond one-to-one to the possible rewrites of  Fig.~\ref{fi:running-rewriting}.}

The impossibility of modelling these situations with stable
event structures 
is well-known (see,
e.g.,~\cite{Win:ES} for a general discussion,~\cite{Handbook} for
graph rewriting systems or~\cite{CVY:ESSPE} for the $\pi$-calculus). One has
to drop the stability requirement and replace causality
by an enabling relation $\vdash$. More precisely, in the specific case 
we would have
$\emptyset \vdash a$, $\emptyset \vdash b$, $\{ a \} \vdash c$,
$\{ b \} \vdash c$.

\medskip

The questions that we try to answer is: what can be retained of the
duality between events structures and domains, when 
dealing with formalisms with fusions? Which are the properties of 
the domain of computations that arise in this setting? What are the 
event structure counterparts?

The domain of configurations of the example suggests that in this
context an event is still a computation that cannot be decomposed
as the join of other computations. Hence, in order theoretical terms,
it is an irreducible.
However, %
due to unstability, irreducibles are not {necessarily} primes: two different
irreducibles can {represent the same computation step with different minimal enablings}, 
in
a way that an irreducible can be included in a computation that is the
join of two computations without being included in any of the two. For
instance, in the example above, $\{ a, c \}$ is an irreducible,
corresponding to the execution of $c$ enabled by $a$, and it is
included in $\{ a \} \sqcup \{ b, c \} = \{ a, b , c \}$,
although neither
$\{ a, c \} \subseteq \{ a \}$ nor
$\{ a, c \} \subseteq \{ b , c \}$.
Uniqueness of decomposition of an element in terms of (downward closed sets of) irreducibles
also fails, e.g.,
$\{a,b,c\} = \{a\} \sqcup \{b\} \sqcup \{a,c\} = \{a\} \sqcup \{b\} \sqcup
\{b,c\}$: the irreducibles $\{a,c\}$ and $\{b,c\}$ can be used
interchangeably in the decomposition of $\{a,b,c\}$.

Building on the previous observation, we introduce an equivalence on irreducibles
identifying those that can be used interchangeably in the
decompositions of an element (intuitively, different minimal enablings
of the same {computation step}). This is used to define a weaker notion of primality
(up to interchangeability) that allows us to characterise the class of domains
suited for modelling the semantics of formalisms with fusions {as}
the class of weak prime algebraic domains.

Given a weak prime algebraic domain, a corresponding event structure
can be obtained by taking as events the set of irreducibles,
quotiented under the (transitive closure of the) interchangeability
relation. The resulting class of event structures is a (mild)
restriction of the general event structures in~\cite{Win:ES}
that we call connected event structures. Categorically, we get an
equivalence between the category of weak prime algebraic domains and
the one of connected event structures, generalising the equivalence
between prime algebraic domains and prime event structures.

We also show that, in the same way as prime algebraic domains/prime
event structures are exactly what is needed for Petri nets/linear graph 
rewriting systems, weak prime
algebraic domains/connected event structures are exactly what is needed
for non-linear graph rewriting systems: each rewriting system maps to
a connected event structure and conversely each connected event structure
arises as the semantics of some rewriting system. This supports the
adequateness of weak prime algebraic domains and connected event
structures as semantics structures for formalisms with fusions.

Interestingly, we can also show that the category of general
event structures~\cite{Win:ES} coreflects into our category of weak
prime algebraic domains. Therefore our notion of weak prime algebraic
domain can be seen as a novel characterisation of the partial order of
configurations of such event structures that is alternative to those
based on intervals in~\cite{Winskel:phd,Dro:ESD}.  It represents a
natural generalisation of the one for prime event structures, with
irreducibles (instead of primes) having a tight connection with
events. The correspondence is established, at a categorical level, as a
coreflection of categories:
to the best of our knowledge, this {has} not been done before in the literature.

As mentioned above, weak prime domains, corresponding to possibly
unstable event structures, satisfy the same conditions as prime
domains, corresponding to stable event structures, up to an
equivalence on irreducibles. This suggests the possibility of viewing
unstable event structures as stable ones up to an equivalence on
events.
We show how this can be formalised  with a set up closely
related to the framework of prime event structures with equivalence 
recently devised in~\cite{win2017,VismeW19}.

Event structures and their domains have been also studied in relation
with automata with concurrency~\cite{Dro:CAD,DK:ACRS}, 
a form of automata endowed with a concurrency relation on transitions 
(local to each state).
On a similar line, the transition graphs of prime event structures
have been given a characterisation in terms of local axioms
in~\cite{PU:RMC}, answering a question posed in~\cite{SNW:MFCTC}.
Recently, in connection with the abstract theory of rewriting and
concurrent games, a slightly different but equivalent characterisation
has been rediscovered in~\cite{Mel:hab}, where prime event structures
are shown to correspond exactly to a suitable class of
asynchronous graphs.
Roughly, an asynchronous graph is a transition system where some
squares are declared to commute, meaning that the coinitial edges of
the square are concurrent and each one can follow the
other. Asynchronous graphs correspond to prime event structures that
satisfy the cube axiom, consisting of two parts: the forward
and the backward cube axioms, the latter often referred to as the
stability axiom. We show that asynchronous graphs that verify only the
forward part of the cube axiom are exactly the transition systems of
weak prime domains.

The rest of the paper is structured as follows. In
Section~\ref{se:background} we recall the basics of (prime) event
structures and their correspondence with prime algebraic domains.  In
Section~\ref{se:fes} we introduce weak prime algebraic domains and
connected event structures, and we characterise their relation
categorically.  
In Section~\ref{se:characterisations}
we present a characterisation of our proposal 
in terms of a formalism reminiscent of prime event structures with equivalence of~\cite{win2017,VismeW19}.
We also discuss and formalise the relation of our work with alternative  characterisations of the domains of event structures  based on intervals and on asynchronous graphs. 
In Section~\ref{se:graphs} we show the intimate connection between
weak prime algebraic domains (or equivalently, connected event
structures) and non-linear graph rewriting systems.
Finally, in Section~\ref{se:conc} we wrap up the main contributions of
the paper and we sketch further advances and some connections with
related works.

The paper is rounded up with an appendix extending our characterisation results
to event structures with non-binary conflict~\cite{Dro:ESD}. We also discuss the relation with a  proposal based on labelled event structures for
modelling the concurrent computations of name passing process
calculi~\cite{CVY:ESSPE}.

This is a revised and extended  version of the conference paper~\cite{BCG:DESF}.

\section{Background: Domains and Event Structures}
\label{se:background}

In this section we review the basics of event structures, as introduced
in~\cite{Win:ES}, and their duality with partial orders.

\subsection{Event Structures}
\label{ss:es}

For the sake of presentation, 
we focus on event structures with binary conflict.
Most results can be easily rephrased for event structures with
non-binary conflict expressed by means of a consistency predicate
(This is explicitly discussed in Appendix~\ref{app:consistency}).
Given a set $X$ we denote by
$\Pow{X}$ and $\Powfin{X}$ the powerset and the set of finite subsets
of $X$, respectively. For $m,n \in \mathbb{N}$, we denote by
$\interval[m]{n}$ the set $\{ m, m+1, \ldots, n\}$.

\begin{definition}[event structure]
  \label{de:es}
  An \emph{event structure} ({\esabbr} for short) is a tuple
  $\langle E, \vdash, \# \rangle$ such that
  \begin{itemize}
  \item $E$ is a set of events;
  \item ${\vdash} \subseteq \Powfin{E} \times E$ is the \emph{enabling}
    relation, satisfying $X \vdash e$ and $X \subseteq Y$
    implies $Y \vdash e$;
  \item $\# \subseteq E \times E$ is the conflict relation.
  \end{itemize}
  A subset $X \subseteq E$ is \emph{consistent} if $\neg (e \# e')$
  for all $e, e' \in X$.
\end{definition}

An {\esabbr} $\langle E, \vdash, \# \rangle$ is often
denoted simply by $E$.  Computations are captured by the notion of
configuration.

\begin{definition}[configuration, live {event structure}]
  Let $\langle E, \vdash, \# \rangle$ be an {\esabbr}. A
  \emph{configuration} of $E$ is a consistent subset $C \subseteq E$
  that is \emph{secured}, i.e., such that for all $e \in C$ there are
  $e_1, \ldots, e_n \in C$ with $e_n =e$ such that
  $\{ e_1, \ldots, e_{k-1} \} \vdash e_k$ for all $k \in \interval{n}$
  (in particular, $\emptyset \vdash e_1$). The set of configurations
  of an {\esabbr} $E$ is denoted by $\conf{E}$ and the subset of
  \emph{finite} configurations by $\conff{E}$.
  An {\esabbr} is \emph{live} if conflict is 
 \emph{saturated}, i.e., for all
  $e, e' \in E$, if there is no $C \in \conf{E}$ such that
  $\{e,e'\} \subseteq C$ then $e \# e'$, and moreover
  for all $e \in E$ it holds $\neg (e \# e)$.
\end{definition}

\begin{remark}
  In live {\esabbr}, the fact that conflict is saturated corresponds
  to inheritance of conflict in prime event structures. Moreover, the
  absence of self-conflicts implies that each event appears in some
  configuration (intuitively, it is executable). In the rest of the
  paper, we restrict to live {\esabbr}, hence the qualification
  ``live'' is omitted.
\end{remark}

In this setting, two events are \emph{concurrent} when they are consistent
and enabled by the same configuration.

Since the enabling predicate is over finite sets of events, we can
consider minimal sets of events enabling a given one.

\begin{definition}[minimal enabling]
  \label{de:minimimal-enabling}
  Let $\langle E, \vdash, \# \rangle$ be an {\esabbr}. Given a
  configuration $C \in \conf{{E}}$ and an event $e \in E$ we write
  $C \vdash_0 e$ and call it a \emph{minimal enabling} for $e$, when
  $C \cup \{ e \} \in \conf{E}$ (hence $C \cup \{ e \}$ consistent and
  $C \vdash e$), and for any other configuration $C' \subseteq C$, if
  $C' \vdash e$ then $C' = C$.
\end{definition}

The classes of stable and prime {\esabbr} represent our starting point and 
play an important role in the paper.

\begin{definition}[stable and prime {event structures}]
  \label{de:stable-prime-es}
  An {\esabbr} $\langle E, \vdash, \# \rangle$ is \emph{stable}
  if $X\vdash e$, $Y \vdash e$, and $X \cup Y \cup \{e\}$ consistent
  imply $X \cap Y \vdash e$.  It is \emph{prime} if $X\vdash e$ and
  $Y \vdash e$ imply $X \cap Y \vdash e$.
\end{definition}

For stable {\esabbr}, given a configuration $C$ and an event $e \in C$,
there is a unique minimal configuration $C' \subseteq C$ such that
$C' \vdash_0 e$.  The set $C'$ can be seen as the set of causes of
the event $e$ in the configuration $C$. This gives a well-defined notion of
causality that is local to each configuration. In a prime {\esabbr}, for any event $e$
there is a unique minimal enabling $C \vdash_0 e$, thus providing a
global notion of causality.
In general, in possibly unstable {\esabbr}, due to the presence of consistent
\emph{or-enablings}, there might be distinct minimal enablings in the
same configuration.

\begin{example}
  \label{ex:event-structure}
  A simple example of unstable {\esabbr} is the one associated with
  the running example discussed in the introduction (see
  Fig.~\ref{fi:running}). The set of events is $\{ a, b, c \}$, the
  conflict relation $\#$ is the empty one and the minimal enablings
  are $\emptyset \vdash_0 a$, $\emptyset \vdash_0 b$,
  $\{ a\} \vdash_0 c$, and $\{ b\} \vdash_0 c$. Thus, event $c$ has
  two minimal enablings and these are consistent, hence
  $ \{ a, b \} \vdash c$.  The corresponding configurations are
  reported in Fig.~\ref{fi:running-configurations}.
\end{example}

The class of {\esabbr} can be turned into a category.

\begin{definition}[category of {event structures}]
  \label{de:es-morphism}
  A morphism of {\esabbr} $f : {E}_1 \to {E}_2$
  is a partial function $f : E_1 \to E_2$ such that for all
  $C_1 \in \conf{E_1}$ and $e_1, e_1' \in E_1$ with $f(e_1)$, $f(e_1')$ defined
  \begin{itemize}
  \item if $f(e_1) \# f(e_1')$ then $e_1 \# e_1'$;
  \item if $f(e_1) = f(e_1')$ then  $e_1 = e_1'$ or $e_1 \# e_1'$;
  \item if $C_1 \vdash_1 e_1$ then $f(C_1) \vdash_2 f(e_1)$.
  \end{itemize}
  We denote by $\es$ the category of {\esabbr} and their
  morphisms and by $\ses$ and $\pes$, respectively, the full subcategories of stable and prime {\esabbr}.
\end{definition}

\subsection{Domains}
\label{ss:domains}

A preordered or partially ordered set $\langle D, \sqsubseteq \rangle$
is often denoted simply as $D$, omitting the (pre)order
relation.  We denote by
$\preceq$ the immediate predecessor relation, i.e., for $x, y \in D$, 
we write $x \preceq y$ whenever
$x \sqsubseteq y$ and for all $z \in D$ if
$x \sqsubseteq z \sqsubseteq y$ then $z \in \{x,y\}$.
A subset $X \subseteq D$ is \emph{consistent}
if it has an upper bound $d \in D$  (i.e., $x \sqsubseteq d$ for all $x \in X$),
while it is \emph{pairwise consistent} if
every two elements subset 
of $X$ is consistent.
A subset $X \subseteq D$ is \emph{directed} if $X \neq \emptyset$ and
every pair of elements in $X$ has an upper bound in $X$. We say that
$D$ is \emph{complete} if every directed subset has a least upper
bound in $D$.

A subset $X \subseteq D$ is an \emph{ideal} if it is directed and
downward closed.
Given an element $x \in D$, we write $\principal{x}$ to
denote the \emph{principal ideal} $\{ y \in D \mid y \sqsubseteq x \}$ generated by $x$.
Given a partial order $D$, its \emph{ideal completion}, denoted by $\ideal{D}$, is the
set of ideals of $D$,  ordered by subset inclusion.
The \emph{least upper bound} and the \emph{greatest lower bound} of a
subset $X \subseteq D$ (if they exist) are denoted by $\bigsqcup X$
and $\bigsqcap X$, respectively.

\begin{definition}[domains]
  \label{c-de:domain}
  A partial order $D$ is \emph{coherent} if for all
  pairwise  consistent $X \subseteq D$ the least
  upper bound $\bigsqcup X$ exists. 
  An element $d \in D$ is \emph{compact} if for all directed
  $X \subseteq D$, $d \sqsubseteq \bigsqcup X$ implies
  $d \sqsubseteq x$ for some $x \in X$. The set of compact
  elements of $D$ is denoted by $\compact{D}$.
  A coherent partial order $D$ is \emph{algebraic}
  if for every $x \in D$ we have
  $x = \bigsqcup (\principal{x} \cap \compact{D})$.  We say that $D$
  is \emph{finitary} if for every element $a \in \compact{D}$ the set
  $\principal{a}$ is finite.
  We refer to 
  algebraic finitary coherent partially ordered sets as
  \emph{domains}.
\end{definition}

Note that every domain has a bottom element (indeed $\bot = \bigsqcup \emptyset$), and that
in a domain all non-empty subsets have a meet. In fact, if
$\emptyset \neq X \subseteq D$, then $\bigsqcap X = \bigsqcup L(X)$
where $L(X) = \{ y \mid \forall x \in X.\, y \sqsubseteq x\}$ is the
set of lowerbounds of $X$, which is pairwise consistent since it is
dominated by any $x \in X$. And it is easy to see that finite joins of 
compact elements are compact.

For a domain $D$ we can think of its elements as ``pieces of
information'' expressing the states of evolution of a process. Compact
elements represent states that are reached after a finite number of
steps.  Thus algebraicity essentially says that infinite
computations can be approximated with arbitrary precision by finite
ones. More formally, when $D$ is algebraic, it is determined by
$\compact{D}$, i.e., $D \simeq \ideal{\compact{D}}$.

For an {\esabbr}, the configurations ordered by subset inclusion form
a domain. When the {\esabbr} is stable, if a minimal
enabling is included in the join of different configurations, then it is necessarily included in one of the configurations. In
order-theoretic terms, minimal enablings are prime elements, and thus
they represent the building blocks of computations.

\begin{definition}[primes and prime algebraicity] Let $D$ be a domain.
  A \emph{complete prime} is an element $p \in D$
  such that for any pairwise consistent $X \subseteq D$, if
  $p \sqsubseteq \bigsqcup X$ then $p \sqsubseteq x$ for some
  $x \in X$.
  The set of complete prime elements of $D$ is denoted by $\pr{D}$. 
  The domain $D$ is \emph{prime algebraic} (or simply \emph{prime}) 
  if for all $x \in D$ we have
  $x = \bigsqcup (\principal{x} \cap \pr{D})$.
\end{definition}

Prime domains can be also characterised as coherent finitary distributive complete partial orders~\cite{Win:ES}.
Note that complete primes are compact (since each directed set is pairwise
consistent). Since we will only use complete primes, the qualification ``complete'' will be omitted.

Prime domains are the domain
theoretical counterpart of stable and prime {\esabbr}.
For a stable {\esabbr} $\langle E, \#, \vdash \rangle$, the
partial order $\langle \conf{{E}}, \subseteq \rangle$ is
a prime  domain, denoted $\pdom{{E}}$. Conversely,
given a prime  domain $D$, the triple
$\langle \pr{D}, \#, \vdash \rangle$, where $p \# p'$ if $\{ p, p' \}$
is not consistent and $X \vdash p$ when
$(\principal{p} \cap \pr{D}) \setminus \{ p \} \subseteq X$,
is a prime {\esabbr}, denoted $\sev{D}$.

This correspondence can be elegantly formulated at the categorical
level~\cite{Win:ES}.
We recall the notion of domain morphism.

\begin{definition}[category of prime domains]
  \label{de:pdomain-category}
  Let $D_1$, $D_2$ be prime domains. A morphism
  $f : D_1 \to D_2$ is a total function such that
  for all consistent $X_1 \subseteq D_1$ and
  $d_1, d_1' \in D_1$
  \begin{enumerate}
  \item if $d_1 \preceq d_1'$ then $f(d_1) \preceq f(d_1')$;
  \item $f(\bigsqcup X_1) = \bigsqcup f(X_1)$;
  \item if $X_1 \neq \emptyset$ then
    $f(\bigsqcap X_1) = \bigsqcap f(X_1)$;
  \end{enumerate}
  We denote by $\PDom$ the category of prime domains and
  their morphisms.
\end{definition}

The correspondence is then captured by the result below.

\begin{theorem}[duality]
  \label{th:duality}
  There are functors $\ppdom : \ses \to \PDom$ and
  $\zsev : \PDom \to \ses$ establishing a coreflection. It restricts
  to an equivalence of categories between $\PDom$ and $\pes$.
\end{theorem}

\section{Weak Prime Domains and Connected Event Structures}
\label{se:fes}

In this section we show that, relaxing the stability assumption, we
can generalise the duality result described in the previous
section, linking suitably defined classes of domains and 
{\esabbr}. These can be proven to properly capture the
semantics of computational formalisms with fusions.

\subsection{Weak Prime Algebraic Domains}
\label{ss:weakdomains}

We show that domains arising in absence of stability can be
characterised by resorting to a weakened notion of prime element. 
We start recalling the notion of irreducible element.

\begin{definition}[irreducibles]
  Let $D$ be a domain. A \emph{complete irreducible} of $D$ is an
  element $x \in D$ such that, for any pairwise consistent
  $X \subseteq D$, if $x = \bigsqcup X$ then $x \in X$.  The set of
  complete irreducibles of $D$ is denoted by $\ir{D}$ and, for
  $d \in D$, we define $\ir{d} = \principal{d} \cap \ir{D}$.
\end{definition}

Observe that complete irreducibles in a domain are compact. In fact,
if $i$ is a complete irreducible, by algebraicity,
$i = \bigsqcup \principal{i} \cap \compact{D}$ whence
$i \in \principal{i} \cap \compact{D}$. Conversely, we have the following.

\begin{lemma}[irreducibility and compactness]
  \label{le:comp-irr}
  Let $D$ be a domain. If $d \in \compact{D}$ then $d$ is a complete
  irreducible iff for all $x, y \in D$, consistent,
  $d = x \sqcup y$ implies $d = x$ or $d=y$.
\end{lemma}

\begin{proof}
  Let $d \in \compact{D}$. Assume that for all $x, y \in D$, consistent,
  $d = x \sqcup y$ implies $d = x$ or $d=y$.
  Assume that $d = \bigsqcup X$ for some pairwise consistent $X$. It
  is easy to see that $X' = \{ \bigsqcup Y \mid Y \in \Powfin{X} \}$
  is directed and {moreover} $\bigsqcup X' = \bigsqcup X = d$. Since $d$
  is compact, there is $x' \in X'$ such that $d \sqsubseteq x'$, hence
  $d = x'$. By definition of $X'$, this means that there exists
  $Y \in \Powfin{X}$ such that $d = \bigsqcup Y$. {Now, using the hypothesis,
  an inductive reasoning} allows us to conclude that
  $d \in Y \subseteq X$, as desired.

  The converse implication is trivial: if $d \in \ir{D}$ and
  $d \sqsubseteq x \sqcup y$ then, by definition of complete
  irreducible, $d \in \{x,y\}$, i.e., either $d=x$ or $d=y$, as
  desired.
\end{proof}

Since in this paper we will refer only to complete irreducibles, the qualification ``complete'' will be omitted.

Irreducibles in domains have a simple characterisation.

\begin{lemma}[unique predecessor for irreducibles]
  \label{le:unique-pred}
  Let $D$ be a domain and $i \in D$. Then $i \in \ir{D}$ 
  iff it has a unique immediate predecessor.
\end{lemma}

\begin{proof}
  Assume that $i \in D$ has a unique immediate predecessor
  $d \prec i$, and let $X \subseteq D$ be pairwise consistent and
  such that $ i = \bigsqcup X$. Hence for any $x \in X$ we have
  $x \sqsubseteq i$. Assume by contradiction that $i \not\in X$. This
  implies that for all elements $x \in X$ we have $x \sqsubseteq d$, and therefore
  $i= \bigsqcup X \sqsubseteq d \prec i$, which is a contradiction.
  Hence it must be $i \in X$, which means that $i$ is irreducible.

  Vice versa, let $i$ be irreducible and let $d_1, d_2 \prec i$ be
  immediate predecessors. Since $D$ is a domain and $\{ d_1, d_2 \}$
  is consistent, we can take $d = d_1 \sqcup d_2$ and we know
  $d_1 \sqsubseteq  d \sqsubseteq i$. Since $i$ is irreducible it cannot be
  $d = i$, therefore $d = d_1$ and thus $d_1 = d_2$. Thus we conclude that
  $i$ has a unique immediate predecessor.
\end{proof}

The unique predecessor of an irreducible will play an important role, hence we introduce a notation.

\begin{definition}[immediate predecessor]
  Let $D$ be a domain and $i \in \ir{D}$. We denote by $\pred{i}$ the
  (unique) immediate predecessor of $i$.
\end{definition}

We next observe that any domain is actually irreducible algebraic, namely it can be generated by the irreducibles.

\begin{proposition}[domains are irreducible algebraic]
  \label{pr:domains-irr-alg}
  Let $D$ be a domain. Then for any $d \in D$ it holds
  $d = \bigsqcup \ir{d}$.
\end{proposition}

\begin{proof}
  We first prove that for any compact element $d \in \compact{D}$ it
  holds that $d = \bigsqcup (\principal{d} \cap \ir{D})$. 
  The thesis then immediately follows from algebraicity of $D$. Since
  $D$ is a domain, $\principal{d}$ is finite, hence we can proceed by
  induction on $|\principal{d}|$. When $|\principal{d}| = 1$, we have
  that $d = \bot$, hence $\principal{d} \cap \ir{D} = \emptyset$ and
  indeed $\bot = \bigsqcup \emptyset$. When $|\principal{d}| = k > 1$
  consider the immediate predecessors of $d$ and denote them
  $d_1, \ldots, d_n \prec d$. Since $D$ is a domain and
  $\{ d_1, \ldots, d_n \}$ is consistent, there exists
  $\bigsqcup \{ d_1, \ldots, d_n \} = d'$ and
  $d_i \sqsubseteq d' \sqsubseteq d$. There are two cases
  \begin{itemize}
  \item $d' = d_i$, for all $i \in \interval{n}$,
    i.e., $d$ has a unique immediate predecessor, hence it is an
    irreducible and thus clearly $d = \bigsqcup (\principal{d} \cap \ir{D})$ or
  \item $d = d' = \bigsqcup \{ d_1, \ldots, d_n \}$. Since, in turn,
    by inductive hypothesis $d_i = \bigsqcup (\principal{d_i} \cap \ir{D})$ 
    and $\principal{d} \cap \ir{D}= \bigcup_{i=1}^n (\principal{d_i} \cap \ir{D})$,
    we immediately get the thesis.
  \end{itemize}
\end{proof}

We next observe that every prime is an irreducible and, if $D$ is a
prime domain, then also the converse holds, i.e., irreducibles
coincide with primes.

\begin{proposition}[irreducibles vs. primes]
  \label{pr:irr-prime-alg}
  Let $D$ be a domain. Then $\pr{D} \subseteq \ir{D}$. Moreover, $D$ is a
  prime domain iff $\pr{D} = \ir{D}$.
\end{proposition}

\begin{proof}
  Let $D$ be a domain. We show that $\pr{D} \subseteq \ir{D}$. Let
  $d \in \pr{D}$. Assume that $d = \bigsqcup X$ for some pairwise
  consistent set $X$. By primality, since $d \sqsubseteq \bigsqcup X$
  there must be $x \in X$ such that $d \sqsubseteq x$. We have also
  $x \sqsubseteq \bigsqcup X = d$ and thus $d=x \in X$.

    \medskip
    
    For the second part, let us assume that $D$ is a prime domain. We have to prove that $\pr{D} = \ir{D}$. We already know that $\pr{D} \subseteq \ir{D}$.
  For the converse inclusion, let $i \in \ir{D}$. By
  prime algebraicity $i =\bigsqcup (\principal{i} \cap
  \pr{D})$. Since $i$ is irreducible, there exists
  $p \in \principal{i} \cap \pr{D}$ such that $i = p$, hence $i$ is
  a prime.

  Vice versa, if $D$ is a domain, by Proposition~\ref{pr:domains-irr-alg} we
  know that $D$ is irreducible algebraic. Hence, if $\pr{D} = \ir{D}$,
  we immediately conclude that $D$ is prime.
\end{proof}

Quite intuitively, in the domain of configurations of an {\esabbr} the
irreducibles are minimal enablings of events. For instance, in the
domain depicted in Fig.~\ref{fi:running-configurations} the
irreducibles are $\{a\}$, $\{b\}$, $\{ a, c \}$, and $\{ b, c \}$.
For stable {\esabbr}, the domain is prime and thus, as observed
above, irreducibles coincide with primes. This fails in unstable {\esabbr},
as we can see in our running example: while $\{a\}$ and
$\{b\}$ are primes, the two minimal enablings of $c$, namely
$\{ a, c \}$ and $\{ b, c \}$, are not. In fact,
$\{ a, c \} \subseteq \{ a \} \sqcup \{ b, c\}$, but neither
$\{ a, c \} \subseteq \{ a \}$ nor $\{ a, c\} \subseteq \{ b, c\}$.

The key observation is that in general an event corresponds to a class of
irreducibles, like $\{ a,c\}$ and $\{ b, c\}$ in our
example. Additionally, two irreducibles corresponding to the same
event can be used, to a certain extent, interchangeably for building
the same configuration. For instance,
$\{ a, b, c \} = \{a, b \} \sqcup \{ a, c \} = \{a, b \} \sqcup \{ b, c
\}$.
We next formalise this intuition, i.e., we interpret
irreducibles in a domain as minimal enablings of some event and we
identify classes of irreducibles corresponding to the same event.

We start by observing that, in a prime domain, any element admits a unique decomposition in terms of downward closed sets of irreducibles (or, equivalently, of primes).

\begin{lemma}[unique decomposition in prime domains]
  \label{le:unique-decomposition}
  Let $D$ be a prime domain and let $X, X' \subseteq \ir{D}$ be
  downward closed sets of irreducibles.
  If $\bigsqcup X = \bigsqcup X'$ then $X=X'$.
\end{lemma}

\begin{proof}
  Let $X, X' \subseteq \ir{D}$ be downward closed sets of irreducibles
  such that $\bigsqcup X = \bigsqcup X'$.  Take any $i' \in X'$. Then
  $i' \sqsubseteq \bigsqcup X$. Since the domain is prime
  algebraic, and thus $i'$ is prime, there must exist $i \in X$ such
  that $i' \sqsubseteq i$ and thus $i' \in X$.  Therefore
  $X' \subseteq X$. By symmetry also the converse inclusion holds,
  whence equality.
\end{proof}

The result above no longer holds in domains arising in the presence
of fusions.  For instance, in the domain in
Fig.~\ref{fi:running-configurations},
$X = \{ \{a\}, \{b\}, \{a,c\} \}$, $X'= \{ \{a\}, \{b\}, \{b,c\} \}$
and $X''= \{ \{a\}, \{b\}, \{b,c\}, \{a,c\} \}$ are all decompositions
for $\{a,b,c\}$.
The idea is to identify irreducibles
that can be used interchangeably in a decomposition.

\begin{definition}[interchangeability]
  \label{de:interchangeable}
  Let $D$ be a domain and $i, i' \in \ir{D}$. We write
  $i \leftrightarrow i'$ if $\{i, i'\}$ consistent and for all
  $X \subseteq \ir{D}$ such that $X \cup \{ i \}$ and
  $ X \cup \{ i' \}$ are downward closed and consistent we have
  $\bigsqcup (X \cup \{ i \}) = \bigsqcup (X \cup \{ i' \})$.
\end{definition}

In words, $i \leftrightarrow i'$ means that $i$ and $i'$ produce the
same effect when added to a decomposition that already includes their
predecessors.
Hence, intuitively, $i$ and $i'$ correspond to the execution of the
same event with different and consistent enablings.

We first observe that distinct irreducibles related by the
interchangeability relation are necessarily incomparable.

\begin{lemma}[$\leftrightarrow$ vs $\sqsubseteq$]
  \label{le:inter-ord}
  Let $D$ be a domain and let $i, i' \in \ir{D}$.
  If $i \leftrightarrow i'$ and $i \sqsubseteq i'$ then $i=i'$.
\end{lemma}

\begin{proof}
  Let $i \leftrightarrow i'$ and $i \sqsubseteq i'$. 
  If $i \neq i'$ and we let
  $X = \ir{\pred{i'}}$, it turns out that $X \cup \{i\} = X$ and
  $X \cup \{i'\}$ are consistent and downward closed. Moreover
  $\bigsqcup X \cup \{i\} = \bigsqcup X = \pred{i'} \neq \bigsqcup X
  \cup \{i'\} = i'$, contradicting $i \leftrightarrow i'$.
\end{proof}

We next give some characterisations of interchangeability.

\begin{lemma}[characterising $\leftrightarrow$]
  \label{le:eq-char}
  Let $D$ be a domain and $i, i' \in \ir{D}$.  Then the
  following are equivalent
  \begin{enumerate}
  \item 
    \label{le:eq-char:1}
    $i \leftrightarrow i'$;

  \item 
    \label{le:eq-char:3}
    $\{i, i'\}$ consistent and for all $d \in \compact{D}$
    if $\pred{i}, \pred{i'} \sqsubseteq d$ then
    $d \sqcup i = d \sqcup i'$;

  \item 
    \label{le:eq-char:4}
    $\{i, i'\}$ consistent and $i \sqcup \pred{i'} = \pred{i} \sqcup i'$.
  \end{enumerate}
\end{lemma}

\begin{proof}

  (\ref{le:eq-char:1} $\to$ \ref{le:eq-char:3})
  Assume that $i \leftrightarrow i'$. By definition, $\{i, i'\}$ is
  consistent. Let $d \in \compact{D}$ be such that
  $\pred{i}, \pred{i'} \sqsubseteq d$. If we let $X = \ir{d}$ we have
  that $\ir{i} \setminus \{i\} \subseteq X$ and similarly
  $\ir{i'} \setminus \{i'\} \subseteq X$.  This implies that
  $X \cup \{ i \}$ and $X \cup \{ i' \}$ are downward closed and
  consistent. Hence
  $d \sqcup i = \bigsqcup X \sqcup i = \bigsqcup (X \cup \{ i \}) =
  \bigsqcup (X \cup \{ i' \}) = \bigsqcup X \sqcup i' = d \sqcup i'$.

  \bigskip

  (\ref{le:eq-char:3} $\to$ \ref{le:eq-char:4})
  Assume (\ref{le:eq-char:3}). Let $p = \pred{i} \sqcup \pred{i'}$,
  which is in $\compact{D}$ since
  $\pred{i}, \pred{i'} \in \compact{D}$. Clearly,
  $\pred{i}, \pred{i'} \sqsubseteq p$.
  Therefore
  $i \sqcup \pred{i'} = i \sqcup \pred{i} \sqcup \pred{i'} = i \sqcup
  p = p \sqcup i' = \pred{i} \sqcup \pred{i'} \sqcup i' = \pred{i}
  \sqcup i'$.

  \bigskip

  (\ref{le:eq-char:4} $\to$ \ref{le:eq-char:1}) Assume
  (\ref{le:eq-char:4}).  Let $X \subseteq \ir{D}$
  be such that $X \cup \{ i\}$ and $X \cup \{ i'\}$ are downward
  closed and consistent sets of irreducibles. This implies that
  $\ir{\pred{i}} \subseteq X$ and similarly
  $\ir{\pred{i'}} \subseteq X$.  Hence, if we let
  $P = \ir{\pred{i}} \cup \ir{\pred{i'}}$, we have
  \begin{center}
    $P  \subseteq X$ \quad and \quad $\bigsqcup P = \pred{i} \sqcup \pred{i'}$.
  \end{center}
  Therefore
  \begin{center}
    $
    \begin{array}{ll}
      \bigsqcup (X \cup \{ i \}) =\\
      \quad = (\bigsqcup X \setminus P) \sqcup \bigsqcup P \sqcup i =\\
      \quad = (\bigsqcup X \setminus P) \sqcup \pred{i} \sqcup \pred{i'} \sqcup i =\\
      \quad = (\bigsqcup X \setminus P) \sqcup i \sqcup \pred{i'} =\\
      \quad = (\bigsqcup X \setminus P) \sqcup \pred{i} \sqcup i' =\\
      \quad = (\bigsqcup X \setminus P) \sqcup \pred{i} \sqcup \pred{i'}  \sqcup i' =\\
      \quad = (\bigsqcup X \setminus P) \sqcup \bigsqcup P \sqcup i' =\\
      \quad = \bigsqcup (X \cup \{ i' \})
    \end{array}
    $
  \end{center}
\end{proof}

The interchangeability relation is clearly reflexive and symmetric, but not
transitive in general: in the domain of Fig.~\ref{fi:not-trans}, using the characterisation in Lemma~\ref{le:eq-char}(\ref{le:eq-char:4}) one can easily see that
$i \leftrightarrow i_1$ and $i_1 \leftrightarrow i'$ but not
$i \leftrightarrow i'$, simply because
$\{i, i'\}$ is not consistent.
More interestingly, in the domain of Fig.~\ref{fi:not-trans-cons}, we have $i \leftrightarrow i_1$, $i_1 \leftrightarrow i_2$, $i_2 \leftrightarrow i'$, hence $i \leftrightarrow^* i'$. However, despite the fact that $\{i, i'\}$ is consistent, it does not hold $i \leftrightarrow i'$, since $\pred{i} \sqcup i' \neq i \sqcup \pred{i'}$. 
This shows that the intuition that
interchangeable irreducibles correspond to the execution of the
same event with different and consistent enablings is still not properly captured. Since $i$ and $i'$ represent the execution of the same event and they are consistent, one would expect that they are interchangeable.

\begin{figure}
  \begin{center}
    \begin{tikzpicture}[node distance=5mm, >=stealth',x=20mm,y=8mm, scale=.8]
      \node[lab] at (2,-0.4) (bot)  {$\bot$};

      \node[lab] at (0,0.8) (pi1)  {$\pred{i}$};
      \node[lab] at (2,0.8) (pi2)  {$\pred{i_1}$};
      \node[lab] at (4,0.8) (pi3)  {$\pred{i'}$};

      \node[lab] at (0,1.9) (i1)   {$i$};
      \node[lab] at (1,1.7) (pi12) {$\bullet$};
      \node[lab] at (2,1.9) (i2)   {$i_1$};
      \node[lab] at (3,1.7) (pi23) {$\bullet$};
      \node[lab] at (4,1.9) (i3)   {$i'$};

      \node[lab] at (1,2.7) (i12) {$\bullet$};
      \node[lab] at (3,2.7) (i23) {$\bullet$};

      \draw [->]  (bot)  -- (pi1);
      \draw [->]  (bot)  -- (pi2);
      \draw [->]  (bot)  -- (pi3);

      \draw [->]  (pi1)  -- (i1);
      \draw [->]  (pi1)  -- (pi12);
      \draw [->]  (pi2)  -- (i2);
      \draw [->]  (pi2)  -- (pi12);
      \draw [->]  (pi2)  -- (pi23);
      \draw [->]  (pi3)  -- (i3);
      \draw [->]  (pi3)  -- (pi23);

      \draw [->]  (i1)  -- (i12);
      \draw [->]  (i2)  -- (i12);
      \draw [->]  (pi12)  -- (i12);
      \draw [->]  (i2)  -- (i23);
      \draw [->]  (i3)  -- (i23);

      \draw [->]  (pi12)  -- (i12);
      \draw [->]  (pi23)  -- (i23);
    \end{tikzpicture}
  \end{center}
  
\caption{Interchangeability need not be transitive.}
\label{fi:not-trans}
\end{figure}

\begin{figure}
  \begin{center}
    \begin{tikzpicture}[node distance=5mm, >=stealth',x=20mm,y=12mm, scale=.7]     
      \node[lab] at (3,0) (bot)  {$\bot$};

      \node[lab] at (0,0.8) (pi1)  {$\pred{i}$};
      \node[lab] at (2,0.8) (pi2)  {$\pred{i_1}$};
      \node[lab] at (4,0.8) (pi3)  {$\pred{i_2}$};
      \node[lab] at (6,0.8) (pi4)  {$\pred{i'}$};

      \node[lab] at (0,1.9) (i1)   {$i$};
      \node[lab] at (1,1.4) (pi12) {$\pred{i} \sqcup \pred{i_1}$};
      \node[lab] at (2,1.9) (i2)   {$i_1$};
      \node[lab] at (3,1.4) (pi23) {$\pred{i_1} \sqcup \pred{i_2}$};
      \node[lab] at (4,1.9) (i3)   {$i_2$};
      \node[lab] at (5,1.4) (pi34) {$\pred{i_2} \sqcup \pred{i'}$};
      \node[lab] at (6,1.9) (i4)   {$i'$};

      \node[lab] at (1,2.9) (i12) {$i \sqcup i_1$};
      \node[lab] at (3,2.7) (i23) {$i_1 \sqcup i_2$};
      \node[lab] at (5,2.9) (i34) {$i_2 \sqcup i'$};

      \node[lab] at (2,4)   (pi141) {$i \sqcup \pred{i'}$};
      \node[lab] at (4,4)   (pi144) {$\pred{i}  \sqcup i'$};

      \node[lab] at (3,3.2) (pi14)  {$\pred{i} \sqcup \pred{i'}$};

      \node[lab] at (3,4.5) (i14)  {$i \sqcup i'$};

      \draw [->]  (bot)  -- (pi1);
      \draw [->]  (bot)  -- (pi2);
      \draw [->]  (bot)  -- (pi3);
      \draw [->]  (bot)  -- (pi4);

      \draw [->]  (pi1)  -- (i1);
      \draw [->]  (pi1)  -- (pi12);
      \draw [->]  (pi2)  -- (i2);
      \draw [->]  (pi2)  -- (pi12);
      \draw [->]  (pi2)  -- (pi23);
      \draw [->]  (pi3)  -- (i3);
      \draw [->]  (pi3)  -- (pi23);
      \draw [->]  (pi3)  -- (pi34);
      \draw [->]  (pi4)  -- (i4);
      \draw [->]  (pi4)  -- (pi34);

      \draw [->]  (i1)  -- (i12);
      \draw [->]  (i2)  -- (i12);
      \draw [->]  (pi12)  -- (i12);
      \draw [->]  (i2)  -- (i23);
      \draw [->]  (i3)  -- (i23);
      \draw [->]  (i3)  -- (i34);
      \draw [->]  (i4)  -- (i34);

      \draw [->]  (pi12)  -- (i12);
      \draw [->]  (pi23)  -- (i23);
      \draw [->]  (pi34)  -- (i34);

      \draw [-,color=white,line width=1.5mm]  (pi1)  to [bend left=17] (pi14);	
      \draw [->]  (pi1)  to [bend left=17] (pi14);
      \draw [-,color=white,line width=1.5mm]   (pi4)  to [bend right=17] (pi14);
      \draw [->]   (pi4)  to [bend right=17] (pi14);
      \draw [->]  (pi14)   -- (pi141);
      \draw [->]  (pi14)   -- (pi144);
      \draw [->]  (pi144)   -- (i14);
      \draw [->]  (pi141)   -- (i14);
      \draw [->] (i1)   to [bend left=17]  (pi141);
      \draw [->]  (i4)   to [bend right=17]  (pi144);
    \end{tikzpicture}
  \end{center}
  
\caption{A domain which is not {\wi} , since Definition~\ref{de:well-interchange}(\ref{de:well-interchange:1}) is violated.}
\label{fi:not-trans-cons}
\end{figure}

\begin{figure}
  \begin{center}
    \begin{tikzpicture}[node distance=5mm, >=stealth',x=20mm,y=12mm, scale=.7]     
      \node[lab] at (3,0) (bot)  {$\bot$};

      \node[lab] at (0,0.8) (pi1)  {$\pred{i'}$};
      \node[lab] at (2,0.8) (pi2)  {$\pred{i}$};
      \node[lab] at (4,0.8) (pi3)  {$\pred{j}$};
      \node[lab] at (6,0.8) (pi4)  {$\pred{j'}$};

      \node[lab] at (0,1.9) (i1)   {$i'$};
      \node[lab] at (1,1.4) (pi12) {$\pred{i'} \sqcup \pred{i}$};
      \node[lab] at (2,1.9) (i2)   {$i$};
      \node[lab] at (3,1.4) (pi23) {$\pred{i} \sqcup \pred{j}$};
      \node[lab] at (4,1.9) (i3)   {$j$};
      \node[lab] at (5,1.4) (pi34) {$\pred{j} \sqcup \pred{j'}$};
      \node[lab] at (6,1.9) (i4)   {$j'$};

      \node[lab] at (1,2.9) (i12) {$i' \sqcup i$};
      \node[lab] at (2.5,2.6) (i2p3) {$i \sqcup \pred{j}$};
      \node[lab] at (3.5,2.6) (i3p2) {$j \sqcup \pred{i}$};
      \node[lab] at (5,2.9) (i34) {$j \sqcup j'$};

      \node[lab] at (2,4)   (pi141) {$i' \sqcup \pred{j'}$};
      \node[lab] at (4,4)   (pi144) {$\pred{i'}  \sqcup j'$};

      \node[lab] at (3,3.2) (pi14)  {$\pred{i'} \sqcup \pred{j'}$};

      \node[lab] at (3,4.5) (i14)  {$i' \sqcup j'$};

      \draw [->]  (bot)  -- (pi1);
      \draw [->]  (bot)  -- (pi2);
      \draw [->]  (bot)  -- (pi3);
      \draw [->]  (bot)  -- (pi4);

      \draw [->]  (pi1)  -- (i1);
      \draw [->]  (pi1)  -- (pi12);
      \draw [->]  (pi2)  -- (i2);
      \draw [->]  (pi2)  -- (pi12);
      \draw [->]  (pi2)  -- (pi23);
      \draw [->]  (pi3)  -- (i3);
      \draw [->]  (pi3)  -- (pi23);
      \draw [->]  (pi3)  -- (pi34);
      \draw [->]  (pi4)  -- (i4);
      \draw [->]  (pi4)  -- (pi34);

      \draw [->]  (i1)  -- (i12);
      \draw [->]  (i2)  -- (i12);
      \draw [->]  (pi12)  -- (i12);
      \draw [->]  (i2)  -- (i2p3); 
      \draw [->]  (i3)  -- (i3p2);
      \draw [->]  (i3)  -- (i34);
      \draw [->]  (i4)  -- (i34);
      
      \draw [->]  (pi12)  -- (i12);
      \draw [->]  (pi23)  -- (i2p3);
      \draw [->]  (pi23)  -- (i3p2);
      \draw [->]  (pi34)  -- (i34);

      \draw [-,color=white,line width=1.5mm]  (pi1)  to [bend left=17] (pi14);	
      \draw [->]  (pi1)  to [bend left=17] (pi14);
      \draw [-,color=white,line width=1.5mm]   (pi4)  to [bend right=17] (pi14);
      \draw [->]   (pi4)  to [bend right=17] (pi14);
      \draw [->]  (pi14)   -- (pi141);
      \draw [->]  (pi14)   -- (pi144);
      \draw [->]  (pi144)   -- (i14);
      \draw [->]  (pi141)   -- (i14);
      \draw [->] (i1)   to [bend left=17]  (pi141);
      \draw [->]  (i4)   to [bend right=17]  (pi144);
    \end{tikzpicture}
  \end{center}
  
\caption{A domain which is not {\wi}, since Definition~\ref{de:well-interchange}(\ref{de:well-interchange:2}) is violated.}
\label{fi:not-trans-cons-bis}
\end{figure}

The next definition formalises two additional properties that a domain must enjoy to provide $\leftrightarrow$ the intended meaning.

\begin{definition}[{\wi} domain]
  \label{de:well-interchange}
  Let $D$ be a domain.
  We say that $D$ is \emph{\wi} when
  \begin{enumerate}
  \item
    \label{de:well-interchange:1}
    for all $i,i' \in \ir{D}$, if $\{\pred{i}, \pred{i'}\}$ consistent
    and $i \leftrightarrow^* i'$ then $i \leftrightarrow i'$;

  \item
    \label{de:well-interchange:2} for all $i,i',j, j' \in \ir{D}$, if
    $i \leftrightarrow^* i'$, $j \leftrightarrow^* j'$,  and $\{i', j'\}$,
    $\{\pred{i}, j\}$,
    $\{\pred{j}, i\}$ consistent then $\{i, j\}$ consistent.
  \end{enumerate}
\end{definition}

Property (\ref{de:well-interchange:1}) is motivated by the discussion above. It intuitively asks that whenever $i$ and $i'$ represents the execution of the same event and they are consistent, then they are interchangeable. Property (\ref{de:well-interchange:2}) can
be read as follows: if $i, i'$ and $j, j'$ represent the same events
and $i', j'$ are consistent, the only source of inconsistency
between $i$ and $j$ is in their enablings. In other words, either $i$ and $j$ are consistent or it must be that $\pred{i}$ is inconsistent with $j$, or $i$ is inconsistent with $\pred{j}$.  A situation in which this property fails is illustrated in Fig.~\ref{fi:not-trans-cons-bis}.

We now introduce weak primes: they weaken the property
of prime elements, requiring that it holds up to
interchangeability.

\begin{definition}[weak prime]
  \label{de:weak-prime}
  Let $D$ be a domain. A \emph{weak prime} of $D$ is an element
  $i \in \ir{D}$ such that for all pairwise consistent
  $X \subseteq D$, if $i \sqsubseteq \bigsqcup X$ then there exist
  $i' \in \ir{D}$ and $d \in X$ such that $i \leftrightarrow i'$ and
  $i' \sqsubseteq d$. We denote by $\wpr{D}$ the set of weak primes of
  $D$.
\end{definition}

Clearly, since interchangeability is reflexive, any prime is a weak
prime. Moreover, in prime domains interchangeability turns out to be
the identity and thus also the converse holds.

\begin{lemma}[weak primes in prime domains]
  \label{le:interchange-id-pad}
  Let $D$ be a prime domain. Then $\leftrightarrow$ is the
  identity and $\wpr{D} = \pr{D}$.
\end{lemma}

\begin{proof}
  Let $i, i' \in \ir{D}$ be such that $i \leftrightarrow i'$. 

  If $i$ and $i'$ are comparable, i.e., $i \sqsubseteq i'$ or
  $i' \sqsubseteq i$, by Lemma~\ref{le:inter-ord} we deduce
  $i = i'$ and we are done.

  Otherwise,
  let
  $X = (\ir{i} \setminus \{i\}) \cup (\ir{i'} \setminus \{i'\})$.
  Note that $X \cup \{ i \}$ and $X \cup \{ i' \}$ are consistent,
  since, by definition of $\leftrightarrow$, $i$ and $i'$ are so.
  Moreover $X \cup \{ i \}$ and $X \cup \{ i' \}$ are downward closed,
  and thus, from $i \leftrightarrow i'$, we deduce
  $\bigsqcup(X \cup \{ i \})= \bigsqcup(X \cup \{ i' \})$.  Since $D$
  is prime, by Lemma~\ref{le:unique-decomposition}, this
  implies that $X \cup \{ i \} = X \cup \{ i' \}$.
  Since $i$ and $i'$ are uncomparable,
  $i, i' \not\in X$ and we
  conclude $i=i'$.
\end{proof}

We argue that the domain of configurations arising in the presence of
fusions can be characterised domain-theoretically as {\wi} domains
  where all irreducibles are weak primes, i.e., that the domain is
algebraic with respect to weak primes.

\begin{definition}[weak prime algebraic domains]
  \label{de:fusion-domain-new}
  Let $D$ be an {\wi} domain.  It is \emph{weak prime algebraic} (or simply
  \emph{weak prime}) if for all $d \in D$ it holds
  $d = \bigsqcup (\principal{d} \cap \wpr{D})$.
\end{definition}

Observe that weak prime domains are assumed to be interchangeable. This hypothesis will actually play a role only when relating weak prime domains and event structures in Section~\ref{ss:domain-to-es}. However, in order to simplify the presentation we preferred to assume it since the beginning.

In the same way as prime domains are domains where all
irreducibles are primes (see Proposition~\ref{pr:irr-prime-alg}), we
can provide a characterisation of weak prime domains in terms of
coincidence between irreducibles and weak primes.

\begin{proposition}[weak prime domains, again]
  \label{pr:fusion-domains}
  Let $D$ be an {\wi} domain.
  It is weak prime iff  all irreducibles are weak primes.
\end{proposition}

\begin{proof}
  Let $D$ be an {\wi} domain. We know, by
  Lemma~\ref{pr:domains-irr-alg}, that $D$ is irreducible
  algebraic. If all irreducibles are weak primes, then clearly $D$ is
  also weak prime algebraic. Conversely, if it is weak prime
  algebraic, then for any irreducible $i \in \ir{D}$, we have that
  $i = \bigsqcup (\principal{i} \cap \wpr{D})$. Since $i$ is
  irreducible, this implies
  $i \in \principal{i} \cap \wpr{D} \subseteq \wpr{D}$, as desired.
\end{proof}

A domain is often built as the ideal completion of its compact
elements. We next provide a characterisation of domains and weak prime
domains based on the generators.

\begin{lemma}[weak prime domains from generators]
\label{le:generators}
  Let $(P,\sqsubseteq)$ be a finitary partial order such that 
 for all $d, d', d'' \in P$, if
  $\{ d, d', d''\}$ is pairwise consistent then $d \sqcup d'$ exists
  and is consistent with $d''$. Then $\ideal{P}$ is a domain with
  $\compact{\ideal{P}} = \{ \principal{d} \mid d \in P \} \simeq P$.

  Additionally, let $P$ be {\wi} and for all $i \in \ir{P}$,
  $d, d' \in P$ consistent, if $i \sqsubseteq d \sqcup d'$ then there
  is $i' \in \ir{P}$, $i \leftrightarrow i'$ such that
  $i' \sqsubseteq d$ or $i' \sqsubseteq d'$. Then $\ideal{P}$ is a
  weak prime domain.
\end{lemma}

\begin{proof}
  Let $(P,\sqsubseteq)$ be a 
  finitary partial order such that for all $d, d', d' \in P'$,
  if $\{ d, d', d''\}$ is pairwise consistent then
  $d \sqcup d'$ exists and is consistent with $d''$.

  The fact that $\ideal{P}$ is a
  complete algebraic finitary
  partial order with
  $\compact{\ideal{P}} = \{ \principal{d} \mid d \in P \} \simeq P$ is
  a standard result. 
  Moreover, let $X \subseteq \ideal{P}$ pairwise consistent. Consider
  $A = \bigcup \{ I \mid I \in X\}$. Observe that for any finite
  $Y \subseteq A$ there exists $\bigsqcup Y$ in $P$. In fact, let
  $Y = \{ y_1, \ldots, y_n \}$. This means that there are
  $I_1, \ldots, I_n$ such that $y_i \in I_i$ for each
  $i \in \interval{n}$. Since $X$ is pairwise consistent in
  $\ideal{P}$, we deduce that $Y$ is pairwise consistent in $P$.
  Since $y_1, y_2$ are consistent and both are consistent with
  $y_3, \ldots, y_n$, by {hypothesis} there exists $y_1 \sqcup y_2$ and it is
  consistent with $y_3, \ldots, y_n$, i.e.,
  $\{ y_1 \sqcup y_2, y_3, \ldots, y_n \}$ is again pairwise
  consistent. Iterating the reasoning we get the existence of
  $y_1 \sqcup y_2 \sqcup \ldots \sqcup y_n = \bigsqcup Y$, as desired.
  Now, if we define
  $I' = \{ \bigsqcup Y \mid Y \subseteq_{\mathit{fin}} A \}$, then
  $I'$ is an ideal and $I' = \bigsqcup X$.

  \bigskip

  Let us consider the second part. It is easy to see that
  $\ir{\ideal{P}} = \{ \principal{i} \mid i \in \ir{P} \}$.
  Moreover, for $i, i' \in \ir{P}$ we have $i \leftrightarrow i'$
  in $P$ if and only if 
  $\principal{i} \leftrightarrow \principal{i'}$ in $\ideal{P}$.
  This immediately implies that $\ideal{P}$ is {\wi}, because so is
  $P$ by assumption.

  We need to show that, under the hypotheses, if
  $I \in \ir{\ideal{P}}$ and $X \subseteq \ideal{P}$ pairwise
  consistent and $I \subseteq \bigsqcup X$ then there exists
  $I' \leftrightarrow I$ and $A \in X$ such that $I' \subseteq
  A$. Thus let $I = \principal{i}$ for some $i \in \ir{P}$. The fact
  that
  $I = \principal{i} \subseteq \bigsqcup X = \bigsqcup \{
  \principal{d} \mid d \in \bigcup X \}$, since $\principal{i}$ is
  finite, means that
  $\principal{i} \subseteq \principal{d_1} \cup \ldots \cup
  \principal{d_n}$ for some finite subset
  $\{ d_1, \ldots, d_n \} \subseteq \bigcup X$ and thus
  $i \sqsubseteq \bigsqcup \{ d_1, \ldots, d_n \}$. Since
  $i \sqsubseteq d_1 \sqcup \bigsqcup \{ d_2, \ldots, d_n\}$, by the
  hypothesis there is $i_1 \leftrightarrow i$ such that
  $i_1 \sqsubseteq d_1$ or
  $i_1 \sqsubseteq \bigsqcup \{ d_2, \ldots, d_n\} = d_2 \sqcup
  \bigsqcup \{ d_3, \ldots, d_n\}$.
  In the second case, again by the hypotheses, there are two
  possibilities. The first is that there is $i_2 \leftrightarrow i_1$,
  such that $i_2 \sqsubseteq d_2$. Note that, since $\pred{i}$ and
  $\pred{i_2}$ are clearly consistent (they are dominated by
  $\bigsqcup \{ d_1, \ldots, d_n \}$), by property
  (\ref{de:well-interchange:1}) of well-interchangeability
  (Definition~\ref{de:well-interchange}), we get
  $i_2 \leftrightarrow i$.
  The second possibility is that
  ${i_2} \sqsubseteq \bigsqcup \{ d_3, \ldots, d_n\}$, and we can iterate the reasoning. 
  In the end, we get the existence of some
  $i' \leftrightarrow i$ and $j \in \interval{n}$ such that
  $i' \sqsubseteq d_j$. Recalling that $d_j \in \bigcup X$, there is
  ${A} \in X$ such that $d_j \in {A}$, hence
  $\principal{{i'}} \subseteq \principal{d_j} \subseteq
  {A}$. Since $i \leftrightarrow {i'}$ in $P$ implies
    $\principal{i} \leftrightarrow \principal{{i'}}$ in $\ideal{P}$,
    we conclude.
\end{proof}

We finally introduce a category of weak prime domains by defining a notion of morphism.

\begin{definition}[category of weak prime domains]
  \label{de:domain-category}
  Let $D_1$, $D_2$ be weak prime domains.
  A  weak prime domain morphism $f : D_1 \to D_2$ is a total function such that
  for all consistent $X_1 \subseteq D_1$ and $d_1, d_1' \in D_1$
  \begin{enumerate}
  \item if $d_1 \preceq d_1'$ then $f(d_1) \preceq f(d_1')$;
  \item $f(\bigsqcup X_1) = \bigsqcup f(X_1)$;
  \item if $d_1, d_1'$ consistent and
    $d_1 \sqcap d_1' \preceq d_1$ then
    $f(d_1 \sqcap d_1') = f(d_1) \sqcap f(d_1')$;
  \end{enumerate}
  We denote by $\WDom$ the category of weak prime domains
  and their morphisms.
\end{definition}

Compared with the notion of morphism for prime domains in
Definition~\ref{de:pdomain-category} (from~\cite{Win:ES}), we still
require the preservation of $\preceq$ and $\sqcup$ of consistent sets
(conditions (1) and (2)). However, the third condition, i.e.,
preservation of $\sqcap$, is weakened to preservation in some
cases. General preservation of meets is indeed not expected in the
presence of fusions. Consider, e.g., the {\esabbr} in
Example~\ref{ex:event-structure}. Take another {\esabbr} $E' = \{c \}$
with $\emptyset \vdash c$ and the morphism $f : E \to E'$ that forgets
$a$ and $b$, i.e., $f(c) = c$ and $f(a)$, $f(b)$ undefined.
Then $f(\{a,c\}) \sqcap f(\{b,c\}) = \{c\} \sqcap \{c\} = \{ c \} \neq
f(\{a,c\} \sqcap \{b,c\}) = f(\emptyset) = \emptyset$. 
Intuitively, the condition $d_1 \sqcap d_1' \prec d_1$ means that
$d_1'$ includes the computation modelled by $d_1$ apart from a final
step, hence $d_1 \sqcap d_1'$ coincides with $d_1$ when such step is
removed. Since domain morphisms preserve immediate precedence (i.e.,
single steps), also $f(d_1)$ differs from $f(d_1')$ for the
execution of a final step and the meet $f(d_1) \sqcap f(d_1')$ is
$f(d_1)$ without such step, and thus it coincides with
$f(d_1 \sqcap d_1')$.

In general we only have
\begin{center}
  $f(\bigsqcap X_1) \sqsubseteq \bigsqcap f(X_1)$
\end{center}
In fact, for all $x_1 \in X_1$, we have
$\bigsqcap X_1 \sqsubseteq x_1$, hence
$f(\bigsqcap X_1) \sqsubseteq f(x_1)$ and thus
$f(\bigsqcap X_1) \sqsubseteq \bigsqcap f(X_1)$.
Still, when restricted to prime domains, also the converse inequality
holds and our notion of morphism boils down to the original one, i.e.,
the full subcategory of $\WDom$ having prime domains as objects is
$\PDom$.

\begin{theorem}[$\PDom$ as a subcategory of $\WDom$]
    \label{th:PdomFullSub}
  The category of prime domains $\PDom$ is the full
  subcategory of $\WDom$ having prime domains as objects.
\end{theorem}

\begin{proof}
  We just need to show that weak prime domain morphisms preserve meets
  on prime domains, i.e., that if $D_1$, $D_2$ are prime domains and
  $f : D_1 \to D_2$ is a weak prime domain morphism then
  $f(\bigsqcap X_1) = \bigsqcap f(X_1)$ for all $X_1 \subseteq D_1$
  pairwise consistent.
  
  We first show that for $d_1, d_1' \in \compact{D_1}$, consistent, it
  holds that $f(d_1 \sqcap d_1') = f(d_1) \sqcap f(d_1')$. We proceed
  by induction on $k = |\principal{d_1} \cap \pr{D}|$.

  When $k=0$ we have $d_1 = \bot$. Since $f$ preserves joins, we have
  that
  $f(\bot) = f(\bigsqcup \emptyset) = \bigsqcup f(\emptyset) =
  \bigsqcup \emptyset = \bot$.  Hence
  \begin{center}
    $f(d_1 \sqcap d_1') = f(\bot \sqcap d_1') = f(\bot)= \bot = \bot \sqcap
    f(d_1') = f(\bot) \sqcap f(d_1') = f(d_1) \sqcap f(d_1')$.
  \end{center}
  
  Suppose now $k > 0$. We distinguish two subcases. If $d_1$ is not
  prime then, recalling that in a prime domain, primes and
  irreducibles coincide, $d_1$ is not irreducible and thus
  $d_1 = e_1 \sqcup f_1$ with $d_1 \neq e_1, f_1 \neq \bot$. It is
  immediate to see that $|\principal{e_1} \cap \pr{D}| < k$ and
  $|\principal{f_1} \cap \pr{D}| < k$. Moreover, since any prime
  algebraic domain is distributive we have
  $d_1 \sqcap d_1' = (e_1 \sqcup f_1) \sqcap d_1' = (e_1 \sqcap d_1')
  \sqcup (f_1 \sqcap d_1')$. Summing up
  \begin{align*}
    & f(d_1 \sqcap d_1') = \\[1mm]
    & \quad =f((e_1 \sqcap d_1') \sqcup (f_1 \sqcap d_1'))\\
    & \qquad \mbox{[Preservation of $\sqcup$]}\\[1mm]
    & \quad =f(e_1 \sqcap d_1') \sqcup f(f_1 \sqcap d_1')\\
    & \qquad \mbox{[Inductive hypothesis]}\\[1mm]
    & \quad =(f(e_1) \sqcap f(d_1')) \sqcup (f(f_1) \sqcap f(d_1'))\\
    & \qquad \mbox{[Distributivity]}\\[1mm]
    & \quad =(f(e_1) \sqcup f(f_1)) \sqcap f(d_1')\\
    & \qquad \mbox{[Preservation of $\sqcup$]}\\[1mm]
    & \quad =f(e_1 \sqcup f_1) \sqcap f(d_1') =\\[1mm]
    & \quad =f(d_1) \sqcap f(d_1')
  \end{align*}

  If instead $d_1$ is prime then note that if $d_1 \sqsubseteq d_1'$
  the thesis is immediate: by monotonicity
  $f(d_1) \sqsubseteq f(d_1')$. Thus
  $f(d_1 \sqcap d_1') = f(d_1) = f(d_1)\sqcap f(d_1')$ as
  desired. Therefore, let us assume that $d_1 \not\sqsubseteq
  d_1'$.
  In this case $d_1 \sqcap d_1' = \pred{d_1} \sqcap d_1'$, since the set
  of lower bounds of $\{ d_1, d_1'\}$ and of $\{ \pred{d_1}, d_1' \}$ is
  the same. Observe that
  \begin{equation}
    \label{eq:pad1}
    \pred{d_1} = d_1 \sqcap (\pred{d_1} \sqcup d_1')
  \end{equation}
  In fact, the join exists since $d_1$, $d_1'$ are consistent. Moreover, by distributivity,
  $d_1 \sqcap (\pred{d_1} \sqcup d_1') = (d_1 \sqcap \pred{d_1})
  \sqcup (d_1 \sqcap d_1') = \pred{d_1} \sqcup (\pred{d_1} \sqcap
  d_1') = \pred{d_1}$.
  
  Therefore
  \begin{align*}
    & f(d_1 \sqcap d_1') =\\[1mm]
    & \quad = f(\pred{d_1} \sqcap d_1')\\
    & \qquad \mbox{[Inductive hypothesis]}\\[1mm]
    & \quad =f(\pred{d_1}) \sqcap f(d_1')\\
    & \qquad \mbox{[Using~(\ref{eq:pad1})]}\\[1mm]
    & \quad = f(d_1 \sqcap (\pred{d_1} \sqcup d_1')) \sqcap f(d_1')\\
    & \qquad \mbox{[By Definition~\ref{de:domain-category}(3)]}\\[1mm]
    & \quad = f(d_1) \sqcap f(\pred{d_1} \sqcup d_1')) \sqcap f(d_1')\\
    & \qquad \mbox{[Preservation of $\sqcup$]}\\[1mm]
    & \quad = f(d_1) \sqcap f(d_1')
  \end{align*}
  as desired. This extends to the meet of finite sets of compact elements, by
  associativity of $\sqcap$, and to 
  infinite sets of compacts by observing that, given an infinite set
  $X$, by finitariness we can identify a finite subset $F \subseteq X$
  such that $\bigsqcap X = \bigsqcap F$.
  The last assertion can be proved by induction on
  $k = \min \{ |\principal{d}| : d \in X \}$. In fact, let $d \in X$
  be an element such that $|\principal{d}| = k$. If {$k=1$} then
  $d = \bot$ and thus $\bigsqcap X = \bot = \bigsqcap \{ d \}$, as
  desired. If {$k >1$}, then we distinguish two possibilities. If for
  all $d' \in X$ it holds $d \sqcap d' = d$ then
  $\bigsqcap X = d = \bigsqcap \{d\}$. If instead, there is $d' \in X$
  such that $d \sqcap d' \sqsubset d$ then recall that the meet of
  compact elements is compact and consider
  $X' = X \cup \{ d \sqcap d' \}$. We have that
  $\bigsqcap X = \bigsqcap X'$. Moreover
  $|\principal{d \sqcap d'}| < k$, hence we can apply the inductive
  hypothesis to $X'$ and get a finite subset $F' \subseteq X'$ such
  that $\bigsqcap X' = \bigsqcap F'$. We conclude observing that
  $\bigsqcap X = \bigsqcap X' = \bigsqcap F' = \bigsqcap ((F'
  \setminus \{d \sqcap d'\}) \cup \{ d, d'\})$. Therefore we can take
  $F = (F' \setminus \{d \sqcap d'\}) \cup \{ d, d'\}$ and we
  conclude.
\end{proof}

\subsection{{From Event Structures to Weak Prime Domains}}
\label{ss:es-to-domain}

We show that the set of configurations of an {\esabbr}, ordered by
subset inclusion, is a weak prime domain where the compact elements
are the finite configurations.
Moreover, the
correspondence can be lifted to a functor. We also identify a subclass
of {\esabbr} that we call \emph{connected {\esabbr}} and that are the exact counterpart of
weak prime domains (in the same way as prime {\esabbr} correspond to prime
algebraic domains).

\begin{definition}[configurations of an {event structure}, ordered]
  Let ${E}$ be an {\esabbr}. We define
  $\dom{{E}} = \langle \conf{{E}}, \subseteq \rangle$.
  Given an {\esabbr} morphism $f : {E}_1 \to {E}_2$, its image
  $\dom{f} : \dom{{E}_1} \to \dom{{E}_2}$ is defined as
  $\dom{f}(C_1) =   \{ f(e_1) \mid e_1 \in C_1 \}$.
\end{definition}

We need some technical facts, collected in the following lemma.
Recall that in the setting of unstable {\esabbr} we can have distinct
consistent minimal enablings for an event. The following notation will be useful.

\begin{definition}[connected enablings]
  Let $E$ be an {\esabbr}, $C, C' \in \conf{E}$ and $e \in E$.
  When $C \vdash_0 e$, $C' \vdash_0 e$, and $C \cup C' \cup \{ e \}$
  is consistent, we write $C \conn{e} C'$. We denote by $\conn{e}^*$ the
  transitive closure of the relation $\conn{e}$.
\end{definition}

Note that, whenever $C \vdash_0 e$ and $C' \vdash_0 e$, requiring 
$C \cup C' \cup \{ e \}$ consistent amounts to require  $C \cup C'$ 
consistent, since conflict is binary.

\begin{lemma}[properties of the domain of configurations]
  \label{le:es-to-fusion-domain}
  Let $\langle E, \vdash, Con \rangle$ be an {\esabbr}.
  Then
  \begin{enumerate}

  \item 
    \label{le:es-to-fusion-domain:1}
    $\dom{{E}}$ is a domain,
    $\compact{\dom{{E}}} = \conff{{E}}$, join is union and
    $C \prec C'$ iff {$C' = C \cup \{ e \}$
    for some $e \in E \setminus C$};

  \item
    \label{le:es-to-fusion-domain:2}
    $C \in \conf{{E}}$ is irreducible iff $C = C' \cup \{ e \}$ and
    $C' \vdash_0 e$; in this case we denote $C$ as $\esir{C'}{e}$;

  \item 
    \label{le:es-to-fusion-domain:3}
    for $C \in \conf{{E}}$, we have
    $\ir{C} = \{ \esir{C'}{e'} \mid e' \in C\ \land\ C' \subseteq C\
    \land\ C' \vdash_0 e' \}$; moreover $\pred{\esir{C'}{e'}} = C'$;

  \item 
    \label{le:es-to-fusion-domain:4}
    for $\esir{C_1}{e_1}, \esir{C_2}{e_2} \in \ir{\dom{{E}}}$, we have
    $\esir{C_1}{e_1} \leftrightarrow \esir{C_2}{e_2}$ iff $e = e_1 = e_2$
    and $C_1 \conn{e} C_2$;
    
  \item
    \label{le:es-to-fusion-domain:5}
    $\dom{E}$ is {\wi}.

  \end{enumerate}
\end{lemma}

\begin{proof}
  \begin{enumerate}
  \item
    We first observe that, given a pairwise consistent set of
    configurations $X \subseteq \conf{E}$, the join is the union
    $\bigsqcup X = \bigcup X$. The fact that $\bigcup X$ is a
    configuration, i.e., that it is consistent and secured immediately
    follows from the fact that each $C \in X$ is.

    Let $C \in \conf{E}$ be a configuration.  For every event $e \in E$,
    since $C$ is secured, we can consider a set
    $C_e = \{ e_1, \ldots, e_n \} \subseteq C$ such that $e_n =e$ and
    $\{ e_1, \ldots, e_{k-1} \} \vdash e_k$ for all
    $k \in \interval{n}$. It is immediate to see that
    $C_e \in \conff{E}$ and clearly $C = \bigsqcup_{e \in C} C_e$.
    
    From the above it is almost immediate to conclude that the
    compact elements of $\dom{{E}}$ are the finite configurations
    $\compact{\dom{E}} =\conff{E}$ and that $\dom{{E}}$ is algebraic.
    Moreover, $\dom{E}$ is finitary, since the number of subsets of a
    finite configurations is clearly finite. Hence $\dom{{E}}$ is a
    domain.

    Concerning immediate precedence, let $C, C' \in \conff{E}$. If
    $C' = C \cup \{ e \}$ with $e \not\in C$ then clearly
    $C \prec C'$, since the order is subset inclusion. Conversely,
      if $C \prec C'$ by definition $C \subseteq C'$ and it must be
      $|C' \setminus C| = 1$. In fact, $C \subseteq C'$ and
      $C \neq C'$, hence $C' \setminus C \neq \emptyset$. Let
    $e, e' \in C' \setminus C$. Let us prove that $e=e'$. Since $C'$
    is secured there is a set of events
    $D = \{ e_1, \ldots, e_n \} \subseteq C'$, such that $e_n = e$ and
    $\{ e_1, \ldots, e_{k-1} \} \vdash e_k$ for all
    $k \in \interval{n}$. Now, if $e' \not\in D$, observe that
    $C \cup D$ is a configuration and $C \subset C \cup D \subset C'$,
    contradicting $C \prec C'$.
    Assume that, instead, $e' \in D$. If $e' = e_k$ for $k< n$ we would
    have that $D' = \{ e_1, \ldots, e_k \}$ is a configuration and we
    could replace $D$ by $D'$ in the contradiction above. Hence it
    must be $e =e'$, as desired.

    \medskip

  \item Let $C \in \conf{{E}}$ be a configuration and assume
    that $C = C' \cup \{ e \}$ with $C' \vdash_0 e$. Then $C$ is a finite configuration, and thus a compact element. Moreover, if
    $C = C_1 \cup C_2$ for $C_1, C_2 \in \conf{{E}}$, 
    then $e$ must occur either in $C_1$ or in
    $C_2$. If $e \in C_1$, since $C_1$ is secured,
    there exists $C_1' \subseteq C_1 \setminus \{ e \}$ such that
    $C_1' \vdash e$. Hence, by monotonicity of enabling,
    $C_1 \setminus \{ e \} \vdash e$. Since $C' \vdash_0 e$ and
    $C_1 \setminus \{ e \} \subseteq C'$ it follows that
    $C_1 \setminus \{ e \} = C'$ and thus $C_1 = C$.
    Therefore, by Lemma~\ref{le:comp-irr}, $C$ is an irreducible.
    
    Vice versa, let $C \in \conf{{E}}$ be an irreducible. It is
    compact, hence finite. Hence we can consider a secured execution
    $\langle e_1, \ldots, e_{n} \rangle$
    of configuration $C$. Note that for any $k \in \interval{n-1}$ it
    must be $\{ e_1, \ldots, e_{k-1} \} \not\vdash e_n$.
    Otherwise, if it were $\{ e_1, \ldots, e_{k-1} \} \vdash e_n$ for
    some $k \in \interval{n-1}$, we would have that $C' = \{
    e_1, \ldots, e_{k}, e_n \}$ and $C'' = \{ e_1, \ldots, e_{n-1} \}$
    are two proper subconfigurations of $C$ such that $C = C'
    \cup C''$, violating the fact that $C$ is irreducible. But this
    means exactly that $\{ e_1, \ldots, e_{n-1} \} \vdash_0 e_n$, as
    desired.

  \medskip
  
\item Immediate.

  \medskip

\item Let $I_j = \esir{C_j}{e_j} \in \ir{\dom{{E}}}$ for
  $j \in \{1,2\}$ be irreducibles. Assume $I_1 \leftrightarrow I_2$.
  By Lemma~\ref{le:eq-char}(\ref{le:eq-char:4}), observing that
  $\pred{I_j} = C_j$, we must have $I_1 \cup C_2 = C_1 \cup I_2$,
  namely $C_1 \cup \{e_1\} \cup C_2 = C_1 \cup C_2 \cup \{ e_2\}$,
  from which we conclude that it must be $e_1 = e_2$, i.e., as desired
  $I_j = \esir{C_j}{e}$, where $e = e_1 = e_2$ for $j \in \{ 1,2 \}$.
  Additionally, $I_1$ and $I_2$ are consistent, by definition of
  $\leftrightarrow$,
  meaning that $C_1 \conn{e} C_2$.

  For the converse, consider two irreducibles $I_1 = \esir{C_1}{e}$
  and $I_2 = \esir{C_2}{e}$, such that $C_1 \conn{e} C_2$. Hence
  $C_1 \vdash_0 e$, $C_2 \vdash_0 e$ and
  $C = C_1 \cup C_2 \cup \{ e \}$ is consistent. Since
  $I_1, I_2 \subseteq C$, they are consistent in
  $\dom{{E}}$. Moreover, $\pred{I_1} = C_1$, $\pred{I_2} = C_2$
  and $I_1 \cup C_2 = I_2 \cup C_1 = C$.
  Hence by
  Lemma~\ref{le:eq-char}(\ref{le:eq-char:4}) we have
  $I_1 \leftrightarrow I_2$, as desired.

\item 
  We have to show that $\dom{E}$ satisfies the conditions of Definition~\ref{de:well-interchange}. Concerning condition (\ref{de:well-interchange:1}), let
  $I_1 = \esir{C_1}{e_1}$ and $I_2 = \esir{C_2}{e_2}$ such that
  $I_1 \leftrightarrow^* I_2$ and $\pred{I_1}=C_1, \pred{I_2}=C_2$ consistent. From
  $I_1 \leftrightarrow^* I_2$, by the above result, we deduce $e_1 = e_2$.
  Since, $C_1, C_2$ consistent, we deduce $C_1 \conn{e} C_2$ and thus,
  again by the same result, $I_1 \leftrightarrow I_2$.

  As for Condition~(\ref{de:well-interchange:2}), consider the
  irreducibles $I$, $I'$, $J$ and $J'$ such that
  $I \leftrightarrow^* I'$, $J \leftrightarrow^* J'$, and
  $\{I', J'\}$, $\{\pred{I}, J\}$ and $\{\pred{J}, I\}$
  consistent. From $I \leftrightarrow^* I'$ and
  $J \leftrightarrow^* J'$ we deduce that $I = \esir{C}{e}$,
  $I' = \esir{C'}{e}$, $J = \esir{D}{f}$ and $J' =
  \esir{D'}{f}$. Moreover, we have $\pred{I}=C$ and $\pred{J}=D$,
  hence the hypotheses say $\{C, J\}$ and $\{D, I\}$ consistent. From
  the consistency of $\{I', J'\}$ we deduce that $\{e,f\}$
  consistent. Therefore we have that $I = \esir{C}{e}$ and
  $J = \esir{D}{f}$ are consistent.
\end{enumerate}
\end{proof}

Concerning point 1, observe that the meet in the domain of
configurations is
$C \sqcap C' = \bigcup \{ C'' \in \conf{{E}} \mid C'' \subseteq C\
\land\ C'' \subseteq C'\}$, which is usually smaller than the
intersection. For instance, in Fig.~\ref{fi:running},
$\{a,c\} \sqcap \{b, c\} = \emptyset \neq \{c\}$.
Point 2 says that irreducibles are configurations of the form
$C \cup \{ e \}$ that admits a secured execution in which the event
$e$ appears as the last one and cannot be switched with any other. In
other words, irreducibles are minimal enablings of events.
Point 3 characterises the irreducibles in a configuration.  According
to point 4, two irreducibles are interchangeable when they are
different minimal enablings for the same event.

\begin{proposition}[{the domain of configurations is weak prime}]
  \label{pr:es-to-dom}
  Let ${E}$ be an {\esabbr}.
  Then $\dom{{E}}$
  is a weak prime domain. Moreover, given two {\esabbr} $E_1$ and $E_2$, and  a morphism
  $f : {E}_1 \to {E}_2$, its image
  $\dom{f} : \dom{{E}_1} \to \dom{{E}_2}$ is a weak prime domain morphism.
\end{proposition}

\begin{proof}
  We know that
  $\dom{{E}}$ is a domain
  (Lemma~\ref{le:es-to-fusion-domain}(\ref{le:es-to-fusion-domain:1}))
  and that it is {\wi}
  (Lemma~\ref{le:es-to-fusion-domain}(\ref{le:es-to-fusion-domain:5})).

  In order to show that $\dom{{E}}$ is a weak prime domain, we
  exploit the characterisation in Proposition~\ref{pr:fusion-domains},
  i.e., we prove that all irreducibles are weak primes.
  Consider an irreducible $I$, which by
  Lemma~\ref{le:es-to-fusion-domain}(\ref{le:es-to-fusion-domain:2})
  is of the shape $I = \esir{C}{e}$ with $C \vdash_0 e$, and suppose that
  $I \subseteq \bigsqcup X$ for some
  $X \subseteq \dom{{E}}$.
  In particular, this means that $e \in \bigsqcup X$ and thus there is
  $C' \in X$ such that $e \in C'$.  In turn, we can consider a minimal
  enabling of $e$ in $C'$, i.e., a minimal $C'' \subseteq C'$ such
  that $C'' \vdash_0 e$, and we have that $I'' = \esir{C''}{e}$ is an
  irreducible $I'' \subseteq C'$. Since $I$ and $I''$ are consistent,
  as they are both included in $\bigsqcup X$, then $C \conn{e} C''$
  and by
  Lemma~\ref{le:es-to-fusion-domain}(\ref{le:es-to-fusion-domain:4})
  $I \leftrightarrow I''$.

  \bigskip
  
  We next prove that given an {\esabbr} morphism
  $f : {E}_1 \to {E}_2$, its image
  $\dom{f} : \dom{{E}_1} \to \dom{{E}_2}$ is a weak prime
  domain morphism.
  \begin{itemize}
  \item $C_1 \preceq C_1'$ implies $\dom{f}(C_1) \preceq \dom{f}(C_1')$\\  
    Since $\dom{f}(C_i) = \{ f(d_i) \mid d_i \in C_i \}$ and by
    Lemma~\ref{le:es-to-fusion-domain}(\ref{le:es-to-fusion-domain:1})
    $C_1 \preceq C_1'$ iff $C_1' = C_1 \cup \{ e_1 \}$ for some event
    $e_1$, the result
    follows immediately.\\

  \item for $X_1 \subseteq \dom{{E}_1}$ consistent, $\dom{f}(\bigsqcup X_1) = \bigsqcup \dom{f}(X_1)$\\
    Since $\dom{f}$ takes the image as set and $\bigsqcup$ on
    consistent sets is union, the result follows.\\

  \item for $C_1, C_1' \in \dom{{E}_1}$ consistent such that
    $C_1 \sqcap C_1' \prec C_1$ it holds
    $f(C_1 \sqcap C_1') = f(C_1) \sqcap f(C_1')$\\
    Since $C_1 \sqcap C_1' \prec C_1$, by
      Lemma~\ref{le:es-to-fusion-domain}(\ref{le:es-to-fusion-domain:1})
      we have that $C_ 1 = (C_1 \sqcap C_1') \cup \{ e_1 \}$ for some
      $e_1 \not\in C_1 \sqcap C_1'$.
    Clearly $e_1 \not\in C_1'$, otherwise we would have $C_1 \subseteq C_1'$
    and thus $C_1 \sqcap C_1' = C_1$. Therefore in this case, the meet
    coincides with intersection,
    $C_1 \sqcap C_1' = C_1 \cap C_1' = C_1 \setminus \{ e_1 \}$.
    Since for the events in $C_1 \cup C_1'$, by definition of event
    structure morphism, $f$ is injective, we have that
    $f(C_1) \cap f(C_1') = f(C_1 \cap C_1')$.
    As a general fact,
    $f(C_1) \sqcap f(C_1') \subseteq f(C_1) \cap f(C_1')$. Therefore,
    putting things together, we conclude
    \begin{center}
      $f(C_1) \sqcap f(C_1') \subseteq f(C_1) \cap f(C_1') = f(C_1
      \cap C_1') = f(C_1 \sqcap C_1')$
    \end{center}
    The converse inequality holds in any domain (as observed after
    Definition~\ref{de:domain-category}) and thus the result follows.
  \end{itemize}
\end{proof}

A special role is played by the subclass of \emph{connected} {\esabbr}
which will be shown to be exact counterpart of weak prime domains.

\begin{definition}[connected {event structure}]
  \label{de:connected-es}
  An {\esabbr} is \emph{connected} if whenever
  $C \vdash_0 e$ and $C' \vdash_0 e$ then $C \conn{e}^* C'$.  We
  denote by $\ces$ the full subcategory of $\es$ having connected {\esabbr}
  as objects.
\end{definition}

In words, different minimal enablings for the same event must be
pairwise connected by a chain of consistency. Equivalently, for each
event $e$ the set of minimal enablings, say
$M_e = \{ C \mid C \vdash_0 e \}$, endowed with the relation
$\conn{e}$ is a connected graph. Intuitively, as discussed in more
detail below, if $M_e$ were not connected, then we could split event
$e$ into different instances, one for each connected component,
without changing the associated domain.

For instance, the {\esabbr} in
Example~\ref{ex:event-structure} is a connected {\esabbr}. Only event $c$ has
two minimal enablings $\{ a \} \vdash_0 c$ and $\{ b \} \vdash_0 c$
and obviously $\{ a \} \conn{c} \{ b \}$.
Clearly, prime {\esabbr} are also connected {\esabbr}. More precisely, we have the following.

\begin{proposition}[{primality = stability + connectedness}]
  \label{le:constapri}
  Let ${E}$ be an {\esabbr}. Then $E$ is prime iff it is stable and connected.
\end{proposition}

\begin{proof}
  The fact that a prime {\esabbr} is stable and connected follows immediately
  from the definitions. Conversely, let ${E}$ be a stable and
  connected {\esabbr}. We show that $E$ is prime, i.e., each $e \in E$ has a
  unique minimal enabling. Let $C, C' \in \conf{E}$ be minimal
  enablings for $e$, i.e., $C \vdash_0 e$ and $C' \vdash_0 e$. Since
  $E$ is connected $C \conn{e}^* C'$. Let
  $C \conn{e} C_1 \conn{e} \ldots \conn{e} C_n \conn{e} C'$. Then
  by stability we get that $C=C_1=\ldots = C_n =C'$. 
\end{proof}

The defining property of connected {\esabbr} allows one to recognise that two
minimal enablings are relative to the same event by only looking at
the partially ordered structure and thus, as we will see, from the
domain of configurations of a connected {\esabbr} we can recover an {\esabbr}
isomorphic to the original one and vice versa (see Theorem~\ref{th:es-dom-equivalence}). 
In general, this is not possible.
For instance, consider the {\esabbr} ${E}'$ with events
$E' = \{ a, b, c \}$, and where $a \# b$ and the minimal enablings are
again $\emptyset \vdash_0 a$, $\emptyset \vdash_0 b$,
$\{ a\} \vdash_0 c$, and $\{ b\} \vdash_0 c$. Namely, event $c$ has two
minimal enablings, but differently from what happens in the
running example, these are not consistent, hence
$ \{ a, b \} \not\vdash c$. The resulting domain of configurations is depicted on
the left of Fig.~\ref{fi:non-fusion}. Intuitively, it is not possible
to recognise that $\{ a,c\}$ and $\{ b, c \}$ are different minimal enablings
of the same event. In fact, we would get an isomorphic
domain of configurations by considering the {\esabbr} ${E}''$ with events
$E'' = \{ a, b, c_1, c_2 \}$ such that $a \# b$ and the minimal enablings
are again $\emptyset \vdash_0 a$, $\emptyset \vdash_0 b$,
$\{ a\} \vdash_0 c_1$, and $\{ b\} \vdash_0 c_2$.

\begin{figure}
  \hfill
  \begin{tikzpicture}[node distance=5mm, >=stealth',x=10mm,y=6mm]
    \node at (1,0) (Gs)  {$\emptyset$};
    \node at (0,1) (Ga)  {$\{a\}$};
    \node at (2,1) (Gb)  {$\{b\}$};
    \node at (0,2.7) (Gac)  {$\{a,c\}$};
    \node at (2,2.7) (Gbc)  {$\{b,c\}$};
    \draw [->] (Gs) -- (Ga);
    \draw [->] (Gs) -- (Gb);
    \draw [->] (Ga) -- (Gac);
    \draw [->] (Gb) -- (Gbc);
  \end{tikzpicture}
  \hfill
  \begin{tikzpicture}[node distance=5mm, >=stealth',x=10mm,y=6mm]
    \node at (1,0) (Gs)  {$\emptyset$};
    \node at (0,1) (Ga)  {$\{a\}$};
    \node at (2,1) (Gb)  {$\{b\}$};
    \node at (0,2.7) (Gac)  {$\{a,c_1\}$};
    \node at (2,2.7) (Gbc)  {$\{b,c_2\}$};
    \draw [->] (Gs) -- (Ga);
    \draw [->] (Gs) -- (Gb);
    \draw [->] (Ga) -- (Gac);
    \draw [->] (Gb) -- (Gbc);
  \end{tikzpicture}
  \hfill$\mbox{}$
  \caption{Non-connected {\esabbr} do not uniquely determine a domain.}
  \label{fi:non-fusion}
\end{figure}

\subsection{From {Weak Prime Domains} to {Connected} Event Structures}
\label{ss:domain-to-es}
We show how to get an {\esabbr} from a weak prime domain. As {anticipated},
events are equivalence classes of irreducibles, where the equivalence
is (the transitive closure of) interchangeability.

In order to properly relate domains to the corresponding {\esabbr} we need to prove some properties of irreducibles and of the interchangeability relation in weak prime domains.

Domains are irreducible algebraic
(see Proposition~\ref{pr:domains-irr-alg}), hence any element is
determined by the irreducibles under it. The difference between two
elements is thus somehow captured by the irreducibles that are under
one element and not under the other.
This motivates the following definition.

\begin{definition}[irreducible difference]
  Let $D$ be a domain and $d, d' \in \compact{D}$ such that
  $d \sqsubseteq d'$. Then we define
  $\diff{d'}{d} = \ir{d'} \setminus \ir{d}$.
\end{definition}

The immediate precedence relation intuitively relates domain elements
corresponding to configurations that differ for the
execution of a single event.
In order to formalise this fact we first need a preliminary technical lemma.

\begin{lemma}[immediate precedence and irreducibles/1]
  \label{le:prec-irr-a}
  Let $D$ be a weak prime domain, $d \in \compact{D}$, and $i \in \ir{D}$
  such that $d$, $i$ are consistent and $\pred{i} \sqsubseteq d$. Then
  \begin{enumerate}
  \item 
    \label{le:prec-irr-a:1}
    for all $i' \in \diff{d \sqcup i}{d}$ minimal, it holds
    $i \leftrightarrow i'$;
  \item 
    \label{le:prec-irr-a:2}
    $d \preceq d \sqcup i$.
  \end{enumerate}
\end{lemma}

\begin{proof}
  \begin{enumerate}
  \item Clearly, if $d = d \sqcup i$ {then $\diff{d \sqcup i}{d} = \emptyset$ and 
   the property trivially
    holds.} Assume $d \neq d \sqcup i$ and take
    $i' \in \diff{d \sqcup i}{d}$ minimal. Note that minimality
    implies that $\pred{i'} \sqsubseteq d$. In fact, for all
    {$i_1' \in \ir{\pred{i'}}$} we have
    $i_1' \sqsubset i' \sqsubseteq d \sqcup i$. Hence
    $i_1' \sqsubseteq d$, otherwise ${i_1'} \in \diff{d \sqcup i}{d}$,
    violating minimality of $i'$. Therefore
    $\pred{i'} = \bigsqcup \ir{\pred{i'}} \sqsubseteq d$.

    Now, from $i' \sqsubseteq d \sqcup i$, since $D$ is a weak prime
    domain and thus irreducibles are weak primes, there must be
    $i'' \in \ir{D}$, $i'' \leftrightarrow i'$ such that
    $i'' \sqsubseteq d$ or $i'' \sqsubseteq i$. We first note that it
    cannot be $i'' \sqsubseteq d$, otherwise
    $d = d \sqcup i'' = d \sqcup i'$, the last equality motivated by
    Lemma~\ref{le:eq-char}(\ref{le:eq-char:3}), {which implies that 
    $i' \sqsubseteq d$,} contradicting the hypothesis. 
    {Hence it must be $i'' \sqsubseteq i$, 
    which by Lemma~\ref{le:unique-pred} means that either $i'' = i$ or $i'' \sqsubseteq \pred{i}$. Since $\pred{i} \sqsubseteq d$ by hypothesis, the latter 
    case would contradict $i'' \not \sqsubseteq d$, hence $i'' = i$ which means
    that $i' \leftrightarrow i$, as desired.}

  \item {Let us assume  that $d \not = d \sqcup i$ (otherwise the property is trivial), and} consider $d'$ such that $d \prec d' \sqsubseteq d \sqcup i$:
   we prove that $d' = d \sqcup i$. Since $d \prec d'$, hence
    $d \neq d'$, we know that $\diff{d'}{d}$ is not empty. Take a
    minimal $i' \in \diff{d'}{d}$. Thus $i'$ is minimal also in
    $\diff{d \sqcup i}{d}$, and thus, by point~(\ref{le:prec-irr-a:1}),
    $i \leftrightarrow i'$. By minimality of $i'$ we deduce also that
    $\pred{i'} \sqsubseteq d$. Since also $\pred{i} \sqsubseteq d$ by
    hypothesis, using Lemma~\ref{le:eq-char}(\ref{le:eq-char:3}), we have
    $d \sqcup i = d \sqcup i'$.  Observing that
    $d \sqcup i' \sqsubseteq d' \sqsubseteq d \sqcup i$ we conclude
    that $d' = d \sqcup i$, as desired.
  \end{enumerate}
  
\end{proof}

We can now  show that whenever $d \prec d'$ the
irreducible difference of $d'$ and $d$ consists of a set of
irreducibles which are pairwise interchangeble, hence, intuitively corresponding to the same event.

\begin{lemma}[immediate precedence and irreducibles/2]
  \label{le:prec-irr-b}
  Let $D$ be a weak prime domain and $d, d' \in D$ such that $d \preceq d'$.
  Then for all $i,i' \in \diff{d'}{d}$ 
  \begin{enumerate}
  \item  
    \label{le:prec-irr-b:1}
    $d' = d \sqcup i$;

  \item  
    \label{le:prec-irr-b:1bis}
    if $i \sqsubseteq i'$ then $i=i'$;

  \item  
    \label{le:prec-irr-b:3}
    $i \leftrightarrow i'$.
  \end{enumerate}
\end{lemma}

\begin{proof}
If $ d = d'$ all properties hold trivially.
\begin{enumerate}
\item[\ref{le:prec-irr-b:1})] Let $i \in \diff{d'}{d}$. Then
  $d \sqsubseteq d \sqcup i \sqsubseteq d'$. It follows that either
  $d \sqcup i = d$ or $d \sqcup i = d'$. The first possibility can be
  excluded for the fact that it would imply $i \sqsubseteq d$, while we
  know that $i \not\in \ir{d}$. Hence we get the thesis.

\item[\ref{le:prec-irr-b:1bis})] Let $i, i' \in \diff{d'}{d}$, with
  $i \sqsubseteq i'$. Let us first assume $i$ minimal  in
  $\diff{d'}{d}$, hence $\pred{i} \sqsubseteq d$. Then
  $i' \sqsubseteq d' = d \sqcup i$. Since $i'$ is a weak prime, there
  exists $i'' \in \ir{D}$ such that {$i' \leftrightarrow i''$} 
  and either $i'' \sqsubseteq i$ or $i'' \sqsubseteq d$. The second
  possibility is excluded. In fact, if $i'' \sqsubseteq d$, then we
  would have $\pred{i}, \pred{i''} \sqsubseteq d$ and thus, by
  Lemma~\ref{le:eq-char}(\ref{le:eq-char:3}),
  $d' = d \sqcup i = d \sqcup i'' = d$, contradicting $d \neq
  d'$. Hence it must be $i'' \sqsubseteq i$. Since
  $i \sqsubseteq i'$, by transitivity $i'' \sqsubseteq i'$ and since
  {$i' \leftrightarrow i''$}, 
  by Lemma~\ref{le:inter-ord},  $i''=i'$ and thus $i''=i=i'$.

  If instead, $i$ is not minimal in $\diff{d'}{d}$, take
  $i_1 \sqsubseteq i$ minimal. By the argument above, we have that
  $i_1 \leftrightarrow i'$, and thus, by Lemma~\ref{le:inter-ord},
  $i_1=i'$. Recalling that $i_1 \sqsubseteq i \sqsubseteq i'$ we
  conclude $i=i'$, as desired.

\item[\ref{le:prec-irr-b:3})] Let $i, i' \in \diff{d'}{d}$ be
  irreducibles. By (\ref{le:prec-irr-b:1}) we have $d' = d \sqcup i$,
  hence $i' \in \diff{d \sqcup i}{d}$. By (\ref{le:prec-irr-b:1bis})
  $i'$ is minimal in $\diff{d \sqcup i}{d}$. Therefore, by
  Lemma~\ref{le:prec-irr-a}(\ref{le:prec-irr-a:1}), we conclude
  $i \leftrightarrow i'$ .
\end{enumerate}
\end{proof}

We next show another technical result, i.e., that chains of immediate
precedence are generated in essentially a unique way by sequences of
irreducibles. Given a domain $D$ and an irreducible $i \in \ir{D}$, we
denote by $\eqclassir{i}$ the corresponding equivalence class. For
$X \subseteq \ir{D}$ we define
$\eqclassir{X} = \{ \eqclassir{i} \mid i \in X \}$.

\begin{lemma}[chains of immediate precedence]
  \label{le:chains}
  Let $D$ be a weak prime domain, $d \in \compact{D}$ and
  $\ir{d} = \{ i_1, \ldots, i_n \}$ such that the sequence
  $i_1, \ldots, i_n$ is compatible with the order (i.e., for all
  $h, k$ if $i_h \sqsubseteq i_k$ then $h \leq k$). If we let
  $d_k = \bigsqcup_{h = 1}^k i_h$ for $k \in \{ 1, \ldots, n\}$ we
  have
  \begin{center}
    $\bot = d_0 \preceq  d_1 \preceq \ldots \preceq d_n =d$
  \end{center}
  
  Vice versa, given a chain
  $\bot = d_0\prec d_1\prec \ldots\prec d_n$ and taking
  $i_h \in \diff{d_h}{d_{h-1}}$ for
  $h \in \{ 1, \ldots, n\}$ we have
  \begin{center}
    $d_n = \bigsqcup_{h=1}^n i_h$ \ \ and \ \
    $\forall i \in \ir{d_n}.\ \exists h \in \interval{n}.\ i
    \leftrightarrow i_h$.
  \end{center}
  Therefore
  $\eqclassir{\{i_1, \ldots, i_n \}} = \eqclassir{\ir{d_n}}$.
\end{lemma}

\begin{proof}
  For the first part, observe that for $k \in \{ 1, \ldots, n\}$ we have that
  \begin{center}
    $\pred{i_k} \sqsubseteq d_{k-1}$
  \end{center}
  In fact, recalling that $\ir{i_k} \subseteq \ir{d}$, we have that
  irreducibles in $\ir{\pred{i_k}} = \ir{i_k} \setminus \{ i_k \}$, which are
  smaller than $i_k$, must occur before in the list hence
  \begin{center}
    $\ir{\pred{i_k}} = \ir{i_k} \setminus \{ i_k \} \subseteq \{ i_1, \ldots,
    i_{k-1}\}$.
  \end{center}
  Therefore
  $\pred{i_k} = \bigsqcup \ir{\pred{i_k}} \sqsubseteq \bigsqcup \{ i_1,
  \ldots, i_{k-1}\} = d_{k-1}$. Thus we use
  Lemma~\ref{le:prec-irr-a}(\ref{le:prec-irr-a:2}) to infer
  $d_{k-1} \preceq d_{k-1} \sqcup i_k = d_k$.
  
  \bigskip

  For the second part, we proceed by induction on $n$. 
  
  \begin{itemize}
  \item ($n=0$) Note that $d_0 = \bigsqcup \emptyset = \bot$ and
    $\ir{\bot} = \emptyset$, hence the thesis trivially holds.

  \item ($n>0$) By induction hypothesis
    \begin{center}
      $d_{n-1} = \bigsqcup_{h=1}^{n-1} i_h$ \ \ and \ \
    $\forall i \in \ir{d_{n-1}}.\ \exists h \in \interval{n-1}.\ i
    \leftrightarrow i_h$.
    \end{center}

    Since by construction $i_n \in \diff{d_n}{d_{n-1}}$,
    by Lemma~\ref{le:prec-irr-b}(\ref{le:prec-irr-b:1}) we deduce
    \begin{center}
      $d_n = i_n \sqcup d_{n-1} = {\bigsqcup_{h=1}^{n} i_h}$
    \end{center}
    
    Moreover, for all $i \in \diff{d_n}{d_{n-1}}$, we have $i
    \sqsubseteq d_n = i_n \sqcup
    d_{n-1}$. By definition of weak prime domain, there exists $i'
    \leftrightarrow i$ such that $i' \sqsubseteq d_{n-1}$ or $i'
    \sqsubseteq i_n$. In the first case, since $i' \sqsubseteq
    d_{n-1}$, by the inductive hypothesis there is $h \in
    \interval{n-1}$ such that $i' \leftrightarrow i_h$. Since $i
    \leftrightarrow i' \leftrightarrow i_h$, and $i, i_h \sqsubseteq
    d_n$ are consistent, by using the fact that
      $D$ is {\wi} we deduce $i \leftrightarrow
  i_h$, as desired. If, instead, we are in the second case, namely $i'
  \sqsubseteq
  i_n$, by Lemma~\ref{le:prec-irr-b}(\ref{le:prec-irr-b:1bis}) it
  follows that $i_n = i' \leftrightarrow i$, as desired.
  \end{itemize}
\end{proof}

In a prime domain, an element admits a unique decomposition in
terms of primes (see Lemma~\ref{le:unique-decomposition}). Here the
same holds for irreducibles but only up to interchangeability. 

\begin{proposition}[unique decomposition up to $\leftrightarrow$]
  \label{pr:unique-dec}
  Let $D$ be a weak prime domain, let $d \in \compact{D}$, and let
  $X \subseteq D$ be a downward closed and consistent set such that
  $\eqclassir{X} \subseteq \eqclassir{\ir{d}}$. Then $d = \bigsqcup X$
  iff $\eqclassir{X} = \eqclassir{\ir{d}}$.
\end{proposition}

\begin{proof}
  ($\Rightarrow$) Let $d = \bigsqcup X$. By hypothesis
  $\eqclassir{X} \subseteq \eqclassir{\ir{d}}$. Hence we only need to
  prove that $\eqclassir{\ir{d}} \subseteq \eqclassir{X}$. Let $i \in
  \ir{d}$. Hence $i \sqsubseteq d = \bigsqcup X$. By definition of
  weak prime domain, this implies that there exists
  $i' \leftrightarrow i$ and $x \in X$ such that $i' \sqsubseteq
  x$. Since $X$ is downward closed, necessarily $i' \in X$ and thus
  $\eqclassir{i} \in \eqclassir{X}$, as desired.

  \bigskip

  ($\Leftarrow$) Let $\eqclassir{X} = \eqclassir{\ir{d}}$. We can
  prove that $\bigsqcup X = d$ by induction on
  $k(X) = |(\ir{d} \setminus X) \cup (X \setminus \ir{d}|$. If $k(X)=0$ then $X = \ir{d}$ and thus, by
  Proposition~\ref{pr:domains-irr-alg}, we conclude that
  $d = \bigsqcup X$.
  If $k(X) > 0$ we distinguish two subcases.
  \begin{itemize}
  \item First assume
  $\ir{d} \setminus X \neq \emptyset$. Then take a minimal element
  $i \in \ir{d} \setminus X$. Therefore $X' = X \cup \{ i \}$ is
  downward closed and, by minimality of $i$, we have
  $\pred{i} \sqsubseteq \bigsqcup X$. Since
  $\eqclassir{X} = \eqclassir{\ir{d}}$, there is $i' \in X$ such that
  $i \leftrightarrow^* i'$ and thus, since
  $\pred{i}, \pred{i'} \sqsubseteq \bigsqcup X$ are consistent and $D$
  is {\wi}, $i \leftrightarrow i'$.  
  Therefore
  \begin{equation}
    \label{eq:irr-b}
    \bigsqcup X' = \bigsqcup X \cup \{ i \} = \bigsqcup X \cup \{ i' \}
    = \bigsqcup X.
  \end{equation}
  
  Since $\eqclassir{X'} =  \eqclassir{X}  =  \eqclassir{\ir{d}}$ and
  $|\ir{d} \setminus X'| = |\ir{d} \setminus X| -1$, we have
  $k(X') < k(X)$, and thus by inductive hypothesis $\bigsqcup X' =
  d$. Hence, using (\ref{eq:irr-b}), we get $\bigsqcup X = d$, as
  desired.
  
\item If instead $\ir{d} \setminus X = \emptyset$, i.e.,
  $\ir{d} \subseteq X$, since $k(X)>0$, it must be
  $X \setminus \ir{d} \neq \emptyset$. Consider a maximal element
  $i \in X \setminus \ir{d}$, and let $X' = X \setminus \{i\}$. Clearly, $X'$ is downward closed because so are $X$ and $\ir{d}$.  Since
  $\eqclassir{X} = \eqclassir{\ir{d}}$, there is
  $i' \in \ir{d} \subseteq X$ such that $i \leftrightarrow^* i'$. Since $X$ is consistent and $D$ is {\wi}, $i \leftrightarrow i'$. Therefore
    \begin{equation}
    \label{eq:irr-c}
    \bigsqcup X = \bigsqcup X' \cup \{ i \} = \bigsqcup X' \cup \{ i' \}
    = \bigsqcup X'.
  \end{equation}
  Since by construction $k(X') = k(X)-1$, the inductive hypothesis
  gives us $\bigsqcup X'=d$. Hence, using (\ref{eq:irr-c}), we get
  $\bigsqcup X = d$, as desired.
  \end{itemize}
\end{proof}

We explicitly observe that, by the above result, whenever $X = \eqclassir{\ir{d}}$ for some $d \in \compact{D}$ then $d$ is uniquely determined by $X$.

We now have all the tools needed for mapping our domains to an {\esabbr}.

\begin{definition}[{event structure} for a weak prime domain]
\label{de:esfusdom}
  Let $D$ be a weak prime domain. The {\esabbr}
  $\ev{D} = \langle E, \#, \vdash \rangle$ is defined as
  follows
  \begin{itemize}

  \item $E = \eqclassir{\ir{D}}$;

  \item $e \# e'$ if there is no $d \in \compact{D}$ such that
    $e, e' \in \eqclassir{\ir{d}}$;

  \item $X \vdash e$ if there is $i \in e$ such that
    $\eqclassir{\ir{i} \setminus \{ i \}} \subseteq X$.

  \end{itemize}

  Given a morphism $f : D_1 \to D_2$, its image
  $\ev{f} : \ev{D_1} \to \ev{D_2}$ is defined  for
  $\eqclassir{i_1} \in {E_1}$ as 
  $\ev{f}(\eqclassir{i_1})= \eqclassir{i_2}$,
  where $i_2 \in \diff{f(i_1)}{f(\pred{i_1})}$, and
  $\ev{f}(\eqclassir{i_1})$ is undefined if $f(\pred{i_1}) = f(i_1)$.
\end{definition}

The events in $\ev{D}$ are equivalence classes of irreducibles. Two events $e$, $e'$ are consistent (not in conflict) when there is some compact element $d$ such that $e, e' \in \eqclassir{\ir{d}}$. Spelled out, this means that there are irreducibles $i \in e$ and $i' \in e'$ such that $i, i' \sqsubseteq d$, i.e., there are minimal enablings of the events $e$ and $e'$ in the same configuration. Finally, an event $e$ is enabled by a set $X$ when $X$ includes, up to intechangeability, all the predecessors of $e$.

Note that the definition above is well-given: in particular, there is no
ambiguity in the definition of the image of a morphism, since by
Lemma~\ref{le:prec-irr-b}(\ref{le:prec-irr-b:3}) we easily conclude
that for all $i_2, i_2' \in \diff{f(i_1)}{f(\pred{i_1})}$, it
holds $i_2 \leftrightarrow i_2'$ (this is argued in detail in the proof of Lemma~\ref{le:domain-to-es-wd}).

In the following we often use a technical lemma that holds in any domain.

\begin{lemma}
  \label{le:inf}
  Let $D$ be a domain and $a, b, c \in D$ such that
  $c \sqsubseteq a$ and $c \preceq b$. Then
  either $b \sqsubseteq a$ or $c = a \sqcap b$.
\end{lemma}

\begin{proof}
  Recall that in a domain the meet of non-empty  sets
  exists. Since $c$ is a lower bound for $a$ and $b$, necessarily
  $c \sqsubseteq a\sqcap b \sqsubseteq b$. If it were
  $c \neq a\sqcap b$ then we would have $a \sqcap b = b$, hence
  $b \sqsubseteq a$, as desired.
\end{proof}

\begin{lemma}[{from weak prime domains to event structures}]
  \label{le:domain-to-es-wd}
  Let  $D$ be a weak prime domain. Then $\ev{D}$ is an {\esabbr}. Moreover, given two weak prime domains $D_1$, $D_2$ and a
  morphism $f : D_1 \to D_2$, its image
  $\ev{f} : \ev{D_1} \to \ev{D_2}$ is an {\esabbr} morphism.
\end{lemma}

\begin{proof}

  We first show that $\ev{D}$ is a live {\esabbr}. In fact, it is an
  {\esabbr}: if $X \vdash e$ and $X \subseteq Y$ then
  $Y \vdash e$. In fact, by definition, if $X \vdash e$ then there
  exists $i \in e$ such that
  $\eqclassir{\ir{i} \setminus \{i\}} \subseteq X$. Hence if
  $X \subseteq Y$ it immediately follows that $Y \vdash e$.
  Moreover $\ev{D}$ is live. The fact that conflict is saturated
  follows immediately by the definition of conflict and the
  characterisation of configurations provided later in
  Lemma~\ref{le:comp-conf}. Conflict is irreflexive since for any
  $e \in \ev{D}$, if $e = \eqclassir{i}$ then
  $e \in \eqclassir{\ir{i}}$, which is a configuration again by
  Lemma~\ref{le:comp-conf}.\footnote{This forward reference is only
    useful to simplify the structure of the presentation and to avoid
    breaking the statement in two parts, but it
    introduces no cyclic dependency.}

  \bigskip

  Given a morphism $f : D_1 \to D_2$, its image
  $\ev{f} : \ev{D_1} \to \ev{D_2}$ is defined for
  $\eqclassir{i_1} \in {E_1}$ as
  $\ev{f}(\eqclassir{i_1})= \eqclassir{i_2}$,
  where $i_2 \in \diff{f(i_1)}{f(\pred{i_1})}$, and
  $\ev{f}(\eqclassir{i_1})$ is undefined if $f(\pred{i_1}) = f(i_1)$.
  First observe that $\ev{f}(\eqclassir{i_1})$ does
  not depend on the choice of the representative. In fact, let
  $i_2, i_2' \in \diff{f(i_1)}{f(\pred{i_1})}$.  Since
  $\pred{i_1} \prec i_1$, by definition of domain morphism,
  $f(\pred{i_1}) \prec f(i_1)$. Thus, by
  Lemma~\ref{le:prec-irr-b}(\ref{le:prec-irr-b:3}),
  $i_2 \leftrightarrow i_2'$.  

  We next show that $\ev{f}$ is an {\esabbr} morphism.
  \begin{itemize}
  \item If $\ev{f}(e_1) \# \ev{f}(e_1')$ then $e_1 \# e_1'$.\\
    We prove the contronominal, namely if $e_1,e_1'$ consistent
    then $\ev{f}(e_1),\ev{f}(e_1')$ consistent.

    The fact that $e_1, e_1'$ consistent means that there exists
    $d_1 \in \compact{D_1}$ such that
    $e_1, e_1' \in \eqclassir{\ir{d_1}}$. We show that
    $\ev{f}(e_1), \ev{f}(e_1') \in \eqclassir{\ir{f(d_1)}}$ (note
    that $f(d_1)$ is a compact, since $f$ is a domain morphism). 

    Let us show, for instance, that
    $\ev{f}(e_1) \in \eqclassir{\ir{f(d_1)}}$. The fact that
    $e_1 \in \eqclassir{\ir{d_1}}$ means that $e_1 = \eqclassir{i_1}$
    for some $i_1 \sqsubseteq d_1$. By definition
    $\ev{f}(e_1) = \eqclassir{i_2}$, where
    $i_2 \in \diff{f(i_1)}{f(\pred{i_1})}$ (since $\ev{f}(e_1)$ is
    defined the irreducible difference cannot be empty). Now, since
    $i_1 \sqsubseteq d_1$ we have that $f(i_1) \sqsubseteq f(d_1)$,
    whence $i_2 \sqsubseteq f(i_1) \sqsubseteq f(d_1)$ and
    $\ev{f}(\eqclassir{i_1}) = \eqclassir{i_2} \in
    \eqclassir{\ir{f(d_1)}}$, as desired.

  \item If $\ev{f}(e_1) = \ev{f}(e_1')$ and $e_1 \neq e_1'$ 
    then $e_1 \# e_1'$.\\
    We prove the contronominal, namely if $e_1, e_1'$ consistent and
    $\ev{f}(e_1) = \ev{f}(e_1')$ then $e_1=e_1'$.

    Assume $e_1, e_1'$ consistent and $\ev{f}(e_1) = \ev{f}(e_1')$. By
    the first condition and the definition of conflict, there must be $d_1 \in \compact{D_1}$ such
    that $e_1, e_1' \in \eqclassir{\ir{d_1}}$, namely
    $e_1=\eqclassir{i_1}$ and ${e_1'}=\eqclassir{i_1'}$ with
    $i_1, i_1' \sqsubseteq d_1$.

    Moreover, $\ev{f}(\eqclassir{i_1}) = \eqclassir{i_2}$ and
    $\ev{f}(\eqclassir{i_1'}) = \eqclassir{i_2'}$ where $i_2$ and
    $i_2'$ are in $\diff{f(i_1)}{f(\pred{i_1})}$ and
    $\diff{f(i_1')}{f(\pred{i_1'})}$, respectively, and
    $\eqclassir{i_2} = \eqclassir{i_2'}$, which means
    $i_2 \leftrightarrow^* i_2'$, and in turn, being $i_2$ and $i_2'$
    consistent, by the fact that $D$ is {\wi},
    implies
    $i_2 \leftrightarrow i_2'$. 

    We distinguish two cases.

    \begin{enumerate}[A.]
    \item If $i_1$ and $i_1'$ are comparable, e.g.,
      if $i_1 \sqsubseteq i_1'$, then $i_1 = i_1'$ and we are
      done. In fact, otherwise, if $i_1 \neq i_1'$ we have
      $\pred{i_1} \prec i_1 \sqsubseteq \pred{i_1'} \prec i_1'$.  By
      monotonicity of $f$ we have
      $f(\pred{i_1}) \prec f(i_1) \sqsubseteq f(\pred{i_1'}) \prec
      f(i_1')$
      (where strict inequalities $\prec$  are motivated by the 
      definition of $\ev{f}$, since
      both $\ev{f}(\eqclassir{i_1})$ and $\ev{f}(\eqclassir{i_1'})$
      are defined).
      Now notice that
      $\pred{i_2} \sqsubseteq i_2 \sqsubseteq f(i_1) \sqsubseteq
      f(\pred{i_1'})$. Moreover,
      $i_2' \in \diff{f(i_1')}{f(\pred{i_1'})}$, therefore
      $\pred{i_2'} \sqsubseteq i_2' \sqsubseteq f(\pred{i_1'})$.
      Hence, using the fact that $i_2 \leftrightarrow i_2'$, by
      Lemma~\ref{le:eq-char}(\ref{le:eq-char:3}) we have
      \begin{center}
        $f(\pred{i_1'}) = f(\pred{i_1'}) \sqcup i_2 = f(\pred{i_1'})
        \sqcup i_2' = f(i_1')$
      \end{center}
      contradicting  the fact that $f(\pred{i_1'}) \prec f(i_1')$.
    
      \smallskip
      
    \item Assume now that $i_1$ and $i_1'$ are uncomparable{: we show that this leads to a contradiction.}   Let 
      $p = \pred{i_1} \sqcup \pred{i_1'}$. We can also assume
      $i_1, i_1' \not\sqsubseteq p$. In fact, otherwise, e.g., if
      $i_1 \sqsubseteq p$, then, by the defining property of weak
      prime domains, we derive the existence of
      $i_1'' \leftrightarrow i_1$ such that
      $i_1'' \sqsubseteq \pred{i_1}$ or
      $i_1'' \sqsubseteq \pred{i'_1}$. The first possibility can be
      excluded because it would imply $i_1'' \sqsubseteq i_1$. Hence,
      since $i_1'' \leftrightarrow i_1$, by
      Lemma~\ref{le:inter-ord}, we would get $i_1=i_1''$,
      contradicting $i_1'' \sqsubseteq \pred{i_1}$. Then it should be
      $i_1'' \sqsubseteq \pred{i_1'} \sqsubseteq i_1'$. Therefore, if
      we take $i_1''$ as representative of the equivalence class we
      are back to case A above.

      Using the fact that $i_1, i_1' \not\sqsubseteq p$ and
      $\pred{i_1}, \pred{i_1'} \sqsubseteq p$, by
      Lemma~\ref{le:prec-irr-a}(\ref{le:prec-irr-a:2}) we deduce that
      $p \prec p \sqcup i_1$ and $p \prec p \sqcup i_1'$.
      Hence $f(p) \prec f(p \sqcup i_1)$ with strict inequality again
      motivated by the definition of $\ev{f}$, since
      $\ev{f}(\eqclassir{i_1})$ is defined.

      By Lemma~\ref{le:prec-irr-b}(\ref{le:prec-irr-b:1}), since
      $i_2 \in \diff{f(i_1)}{f(\pred{i_1})}$ and
      $i_2' \in \diff{f(i_1')}{f(\pred{i_1'})}$, we have
      \begin{equation}
        \label{eq:adding}
        f(\pred{i_1}) \sqcup i_2 = f(i_1) 
        \qquad f(\pred{i_1'}) \sqcup i_2' = f(i_1')
      \end{equation}
          
      Now, observe that
      \begin{center}
      $\begin{array}{lll}
         f(p \sqcup i_1)  =  \\
         = f(\pred{i_1} \sqcup \pred{i_1'} \sqcup i_1)\\
         =  f(\pred{i_1'} \sqcup i_1)\\
         =  f(\pred{i_1'}) \sqcup f(i_1) 
      & \mbox{[preservation of $\sqcup$]}\\
         =  f(\pred{i_1'}) \sqcup f(\pred{i_1}) \sqcup i_2 
      & \mbox{[by~(\ref{eq:adding})]}\\
         =   f(\pred{i_1'}) \sqcup f(\pred{i_1}) \sqcup i_2' 
      & \mbox{[by Lemma~\ref{le:eq-char}(\ref{le:eq-char:3}),} \\
         & \mbox{\ since $i_2 \leftrightarrow i_2'$]}\\
                          =  f(i_1') \sqcup f(\pred{i_1}) 
         & \mbox{[by~(\ref{eq:adding})]}\\
                          =  f(\pred{i_1} \sqcup i_1')
         & \mbox{[preservation of $\sqcup$]}\\
                          =  f(\pred{i_1} \sqcup \pred{i_1'} \sqcup  i_1')\\
                          =  f(p \sqcup  i_1')\\
       \end{array}
       $
     \end{center}
     Since $p \prec p \sqcup i_1$ and $p \prec p \sqcup i_1'$, by
     Lemma~\ref{le:inf}, we have
     $(p \sqcup i_1) \sqcap (p \sqcup i_1') = p$.
     Therefore, on the one hand
     $f((p \sqcup i_1) \sqcap (p \sqcup i_1')) = f(p)$. On the other
     hand, since the meet is an immediate predecessor, by definition
     of weak domain morphism (Definition~\ref{de:domain-category}), it
     is preserved:
     $f((p \sqcup i_1) \sqcap (p \sqcup i_1')) = f(p \sqcup i_1)
     \sqcap f(p \sqcup i_1') = f(p \sqcup i_1) = f(p \sqcup i_1')$.
     Putting things together,
     $f(p) = f(p \sqcup i_1) = f(p \sqcup i_1')$, contradicting  the
     fact that $f(p) \prec f(p \sqcup i_1)$.
   \end{enumerate}

   \item if $C_1 \vdash_1 \eqclassir{i_1}$ and
     $\ev{f}(\eqclassir{i_1})$ is defined then
     $\ev{f}(C_1) \vdash_2 \ev{f}(\eqclassir{i_1})$

     Recall that $C_1 \vdash_1 \eqclassir{i_1}$ means that 
     $\eqclassir{\ir{i_1'} \setminus \{ i_1'\}} =
     \eqclassir{\ir{\pred{i_1'}}} \subseteq C_1$ for some
     $i_1' \leftrightarrow i_1$.

     By definition, $\ev{f}(\eqclassir{i_1}) = \eqclassir{i_2}$ where
     $i_2 \in \diff{f(i_1')}{f(\pred{i_1'})}$. We show that
     $\ev{f}(C_1) \vdash_2 \eqclassir{i_2}$, namely that
     \begin{equation}
       \label{eq:new}
       \eqclassir{\ir{i_2}\setminus \{ i_2\} } =
       \eqclassir{\ir{\pred{i_2}}} \subseteq \ev{f}(C_1)
     \end{equation}
     Observe that since $i_2 \in \diff{f(i_1')}{f(\pred{i_1'})} $ and
     distinct elements in $\diff{f(i_1')}{f(\pred{i_1'})}$ are
       incomparable by
       Lemma~\ref{le:prec-irr-b}(\ref{le:prec-irr-b:1bis}), it holds
       $\pred{i_2} \sqsubseteq f(\pred{i_1'})$. Therefore, we have
     \begin{center}
       $\ir{\pred{i_2}} \subseteq \ir{f(\pred{i_1'})}$
     \end{center}
     Hence, in order to conclude~(\ref{eq:new}), it suffices to show that
     \begin{equation}
       \label{eq:obj}
       \eqclassir{\ir{f(\pred{i_1'})}} \subseteq \ev{f}(C_1)
     \end{equation}
   
     In order to reach this result, first note that, by
     Lemma~\ref{le:chains}, if
     $\ir{\pred{i_1'}} = \{ j_1^1, \ldots, j_1^n \}$ is a sequence of
     irreducibles compatible with the order, we can obtain a
     $\preceq$-chain
     \begin{center}
       $\bot = d_1^0 \preceq d_1^1 \preceq \ldots \preceq d_1^n =
       \pred{i_1'} \prec i_1'$
     \end{center}
     We can extract a strictly increasing subsequence
     \begin{center}
       $\bot = d_1'^0 \prec d_1'^1 \prec \ldots \prec d_1'^m =
       \pred{i_1'} \prec i_1'$
     \end{center}
     and, if we take irreducibles $j_1'^1, \ldots, j_1'^m$ in
     $\diff{d_1'^{i}}{d_1'^{i-1}}$, again by Lemma~\ref{le:chains} we
     know that
     \begin{equation}
       \label{eq:chain}
       \eqclassir{\ir{\pred{i_1'}}} = \eqclassir{\{j_1'^1, \ldots,
         j_1'^m\}}
     \end{equation}
   
     Since $f$ is a domain morphism, it preserves $\preceq$, namely
     \begin{center}
       $\bot = f(d_1'^0) \preceq f(d_1'^1) \preceq \ldots \preceq f(d_1'^m) =
       f(\pred{i_1'}) \prec f(i_1')$
     \end{center}
     where the last inequality is strict since
     $\ev{f}(\eqclassir{i_1'}) = \eqclassir{i_2}$ is
     defined. Moreover, whenever $f(d_1'^{h-1}) \prec f(d_1'^h)$, then
     $\ev{f}(\eqclassir{j_1'^h}) = \eqclassir{\ell_2^h}$ where $\ell_2^h$
     is any irreducible in $\diff{f(d_1'^h)}{f(d_1'^{h-1})}$, otherwise
     $\ev{f}(\eqclassir{j_1'^h})$ is undefined.

     Once more by Lemma~\ref{le:chains} we know that
     \begin{center}
       $\eqclassir{\ir{f(\pred{i_1'})}} = \eqclassir{\{\ell_2^1, \ldots,
         \ell_2^m\}} = \ev{f}(\eqclassir{\{j_1'^1, \ldots, j_1'^m\}})$,
     \end{center}
     thus, using (\ref{eq:chain})
     \begin{equation}
       \label{eq:chain1}
       \eqclassir{\ir{f(\pred{i_1'})}} =
       \ev{f}(\eqclassir{\ir{\pred{i_1'}}}).
     \end{equation}
     Hence, recalling that, by hypothesis,
     $\eqclassir{\ir{\pred{i_1'}}} \subseteq C_1$, we conclude the desired
     inclusion~(\ref{eq:obj}).
  \end{itemize}
\end{proof}

Since in a prime domain irreducibles coincide with
primes (Proposition~\ref{pr:irr-prime-alg}), $\leftrightarrow$ is the
identity (Lemma~\ref{le:interchange-id-pad}) and $\diff{d'}{d}$ is a
singleton when $d \prec d'$, the construction above produces the prime
{\esabbr} $\pes(D)$ as defined in Section~\ref{se:background}.

Given a weak prime domain $D$, the finite configurations of the {\esabbr}
$\ev{D}$ exactly correspond to the elements in
$\compact{D}$. Moreover, in such {\esabbr} we have a minimal enabling
$C \vdash_0 e$ when there is an irreducible in $e$ (recall that events
are equivalence classes of irreducibles) such that $C$ contains all and only
(the equivalence classes of) its predecessors.

\begin{lemma}[compacts vs. configurations]
  \label{le:comp-conf}
  Let $D$ be a weak prime domain and $C \subseteq \ev{D}$ a finite set
  of events.  Then $C$ is a configuration in the {\esabbr} $\ev{D}$
  iff there exists a unique $d \in \compact{D}$ such that $C =
  \eqclassir{\ir{d}}$. Moreover, for any $e \in \ev{D}$ we have that
  $C \vdash_0 e$ iff $C = \eqclassir{\ir{i} \setminus \{ i \}}$ for
  some $i \in e$.
\end{lemma}

\begin{proof}
  The left to right implication {of the first part} follows by proving that, given a
  configuration $C \in \conff{\ev{D}}$, there exists
  $X \subseteq \ir{D}$ downward closed and consistent such that
  $\eqclassir{X} = C$. Hence, if we let $d = \bigsqcup X$, by
  Proposition~\ref{pr:unique-dec}, we have that
  $C = \eqclassir{X} = \eqclassir{\ir{d}}$. Moreover, $d$ is uniquely
  determined, since, by the same proposition we have that for any
  other $X'$ such that $\eqclassir{X'} = C$, since
  $\eqclassir{X'} = C = \eqclassir{X} = \eqclassir{\ir{d}}$,
  necessarily $d = \bigsqcup X'$.

  Let us thus prove the existence of $X \subseteq \ir{D}$ consistent
  and downward closed such that $\eqclassir{X} = C$.  
  We proceed by induction on the cardinality of $C$.

  \begin{itemize}
  \item if $|C| = 0$, namely $C = \emptyset$ then we can take $X= \emptyset$,
    and trivially conclude.

  \item if $|C|>0$, since $C$ is secured, there is
    $\eqclassir{i} \in C$ such that
    $C' = C \setminus \{ \eqclassir{i} \} \vdash \eqclassir{i}$.\
    By inductive hypothesis there is $X' \subseteq \ir{D}$, downward
    closed and consistent such that $\eqclassir{X'} = C'$.

    The fact that
    $C' = C \setminus \{ \eqclassir{i} \} \vdash \eqclassir{i}$ means
    that for some $i' \in \ir{D}$ such that $i' \leftrightarrow^* i$, it holds
    $\eqclassir{\ir{i'} \setminus \{i'\}} = \eqclassir{\ir{\pred{i'}}}
    \subseteq C'$.
    Therefore, there is $X'' \subseteq X'$ such that
    $\eqclassir{X''} = \eqclassir{\ir{\pred{i}}}$ and thus, by
    Proposition~\ref{pr:unique-dec},
    $\pred{i'} \sqsubseteq \bigsqcup X'$.
    We can assume, without loss of generality that
    $\ir{\pred{i'}} \subseteq X'$. If not, we can replace $X'$ by
    $X' \cup \ir{\pred{i'}}$. By the consideration above, it is
    consistent and it has the same join of $X'$.

    Now, an induction on the cardinality $k$ of
    $X' \setminus \ir{\pred{i'}}$ allows us to show that $\{i', j\}$
    consistent for all $j \in X'$. If $k=0$ then
    $X' \setminus \ir{\pred{i'}} =\emptyset$ and the thesis is
    trivial. Otherwise, consider $j' \in X' \setminus \ir{\pred{i'}}$
    maximal and $X'' = X' \setminus \{j'\}$. Since
    $|X' \setminus \ir{\pred{i'}}|=k-1$, by inductive hypothesis, for
    all $j \in X''$, we have $\{j,i'\}$ consistent. Now, since
    $j, \pred{i} \sqsubseteq \bigsqcup X'$, we have that
    $\{j, \pred{i}\}$ is consistent. Moreover, since
    $\ir{j'} \setminus\{j'\} = \ir{\pred{j'}} \subseteq X''$, we have
    that $\{i, \pred{j'}\}$ is consistent. 

    Finally, recalling that, since $C$ is
    consistent, we have that
    $\neg (\eqclassir{j'} \# \eqclassir{i'})$, i.e., there is
    $d \in \compact{D}$ such that
    $\{ \eqclassir{j'}, \eqclassir{i'} \} \subseteq
    \eqclassir{\ir{d}}$. More explicitly, this means that there are
    $j'', i'' \in \ir{D}$ such that $j'' \leftrightarrow^* j'$,
    $i'' \leftrightarrow^* i'$ and $j', i''$ consistent.
    Since $D$ is {\wi}, by condition (\ref{de:well-interchange:2}) of
    Definition~\ref{de:well-interchange}, we conlcude $j', i'$
    consistent.

    We can thus conclude that $X = X' \cup \{ i' \}$ is consistent,
    and downward closed since
    $\ir{\pred{i'}} \subseteq X'$. Hence we conclude.
  \end{itemize}

  \bigskip

  For the converse, let $C = \eqclassir{\ir{d}}$. Let
  $\bot = d_0 \prec d_1 \prec \ldots d_{n-1} \prec d_n = d$ be a chain
  of immediate precedence and for each $h \in \{1, \ldots, n\}$ take
  $i_h \in \diff{d_h}{d_{h-1}}$. By Lemma~\ref{le:chains},
  $d = \bigsqcup \{ i_1, \ldots, i_n\}$ and
  $\eqclassir{\ir{d}} = \eqclassir{\{ i_1, \ldots, i_n\}}$. Moreover,
  for all $h \in \{1, \ldots, n\}$, we have
  $\eqclassir{\ir{i_h} \setminus \{ i_h \}} \subseteq
  \eqclassir{\ir{d_{h-1}}}$, hence
  $\eqclassir{\ir{d_{h-1}}} \vdash \eqclassir{i_h}$. Therefore $C$ is
  secured. Moreover, it is clearly consistent and thus
  $C \in \conf{\ev{D}}$.

  \bigskip

  The second part follows immediately by Definition~\ref{de:esfusdom}.
\end{proof}

Given the lemma above, it is now possible to state how weak prime 
domains relate to connected {\esabbr}.

\begin{proposition}[{from weak prime domains to connected ES}]
  \label{pr:domain-to-fes}
  Let  $D$ be a weak prime domain. Then $\ev{D}$ is a connected {\esabbr}.
\end{proposition}

\begin{proof}
  We have to show that
  if $X \vdash_0 e$ and $X' \vdash_0 e$, then $X \conn{e}^* X'$.  Note
  that, by Lemma~\ref{le:comp-conf}, 
  from $X \vdash_0 e$ and
  $X' \vdash_0 e$, we deduce that there exists $i, i' \in e$ such that
  $\eqclassir{\ir{i} \setminus \{i\}} = X$ and
  $\eqclassir{\ir{i'} \setminus \{i'\}} = X'$. Since $i, i' \in e$ we
  deduce that $i \leftrightarrow^* i'$, namely
  $i = i_0 \leftrightarrow i_1 \leftrightarrow \ldots \leftrightarrow
  i_n = i'$. We proceed by induction on $n$. The base case $n=0$ is
  trivial. If $n >0$ then from
  $i \leftrightarrow i_1 \leftrightarrow^* i'$ we have that
  $i_1 \in e$ and, if we let
  $X_1 = \eqclassir{\ir{i_1} \setminus \{i_1\}}$, then
  $X_1 \vdash_0 e$. By inductive hypothesis, we know that
  $X_1 \conn{e}^* X'$. Moreover, since $i \leftrightarrow i_1$,
  the irreducibles $i$ and $i_1$ are consistent. Hence, by definition
  of conflict in $\ev{D}$, also $X \cup X_1 \cup \{e\}$ is consistent
  and hence $X \conn{e} X_1$. Therefore $X \conn{e}^* X'$, as desired.
\end{proof}

\subsection{Relating Categories of Models}
\label{ss:relating-cats}
We show that, at a categorical level, the constructions taking a weak prime domain to an
{\esabbr} and an {\esabbr} to a domain (the domain of its configurations) establish
a coreflection between the corresponding categories.
This becomes an equivalence when it is restricted to the full subcategory of connected {\esabbr}.

\begin{theorem}[coreflection of $\es$ and $\WDom$]
 \label{th:es-dom-equivalence}
 The functors $\zdom : \es \to \WDom$ and $\zev: \WDom \to \es$ form
 a coreflection $\zev \dashv \zdom$.  It restricts to an equivalence between $\WDom$ and
 $\ces$.
\end{theorem}

\begin{proof}
  Let ${E}$ be an {\esabbr}. Recall that the corresponding domain of
  configurations is
  $\dom{{E}} = \langle \conf{{E}}, \subseteq \rangle$.  Then,
  $\ev{\dom{{E}}} = \langle E', \#', \vdash'\rangle$, where the set of
  events $E'$ is defined as
  \begin{center}
    $E' = \eqclassir{\ir{\dom{{E}}}} = \{ \eqclassir{\esir{C}{e}} \mid C
    \vdash_0 e \}$
  \end{center}
  
  By
  Lemma~\ref{le:es-to-fusion-domain}(\ref{le:es-to-fusion-domain:4}),
  the equivalence class of an irreducible $\esir{C}{e}$ consists of
  all minimal enablings of event $e$ which are connected. Therefore we
  can define a morphism, which is the counit of the adjunction, as follows:
  \begin{center}
    $\begin{array}{lccc}
       \theta_{{E}} : & \ev{\dom{{E}}} &  \to & {E}\\
                             & \eqclassir{\esir{C}{e}}
                             & \mapsto & e
     \end{array}
     $
  \end{center}
  Observe that $\theta_{{E}}$ is surjective. In fact ${E}$ is live and
  thus any event $e \in E$ has at least a minimal enabling
  $C \vdash_0 e$. If we let $I = \esir{C}{e}$, then
  $\eqclassir{I} \in \ev{\dom{{E}}}$ and $\theta_{{E}}(\eqclassir{I}) = e$.
  The mapping $\theta_E$ is clearly
  a  morphism of event structures. In fact, observe that
  \begin{itemize}
  
  \item For 
  $I_1, I_2 \in \ir{\dom{{E}}}$, if
    $\theta_{{E}}(\eqclassir{I_1}) \#
    \theta_{{E}}(\eqclassir{I_2})$
    then $\eqclassir{I_1} \#' \eqclassir{I_2}$.\\
    Let $I_1 = \esir{C_1}{e_1}$ and $I_2 = \esir{C_2}{e_2}$. If
    $\theta_{{E}}(\eqclassir{I_1}) = e_1 \# e_2 =
    \theta_{{E}}(\eqclassir{I_2})$, then there cannot be any
    configuration $C \in \conf{{E}}$ such that $I_1, I_2 \subseteq
    C$. Hence, by definition of conflict in $\ev{\dom{{E}}}$, we have
    $\eqclassir{I_1} \#' \eqclassir{I_2}$.

  \smallskip
  
  \item For  $I_1, I_2 \in \ir{\dom{{E}}}$, with
    $\eqclassir{I_1} \neq \eqclassir{I_2}$, we have that
    $\theta_{{E}}(\eqclassir{I_1}) =
    \theta_{{E}}(\eqclassir{I_2})$
    implies $\eqclassir{I_1} \#' \eqclassir{I_2}$.
    
    In fact, by
    Lemma~\ref{le:es-to-fusion-domain}(\ref{le:es-to-fusion-domain:2}),
    the irreducibles $I_1$ and $I_2$ are of the kind
    $I_1 = \esir{C_1}{e_1}$ and $I_2 = \esir{C_2}{e_2}$. We show that
    if $\eqclassir{I_1}$ and $\eqclassir{I_2}$ are consistent and
    $\theta_{{E}}(\eqclassir{I_1}) = \theta_{{E}}(\eqclassir{I_2})$
    then $\eqclassir{I_1} = \eqclassir{I_2}$.

    Assume
    $\theta_{{E}}(\eqclassir{I_1}) =
    \theta_{{E}}(\eqclassir{I_2})$,
    hence $e_1=e_2$. Since $\eqclassir{I_1}$ and
    $\eqclassir{I_2}$ are consistent, there exists
    $k \in \compact{\dom{{E}}}$ such that
    $\eqclassir{I_1}, \eqclassir{I_2} \in \eqclassir{\ir{k}}$.  
    Compacts in $\dom{{E}}$ are finite configurations, hence the condition
    amounts to the existence of $C \in \conff{{E}}$ such that
    $\eqclassir{I_1}, \eqclassir{I_2} \in \eqclassir{\ir{C}}$, i.e.,
    there  are $I_1', I_2'$ with $I_i \leftrightarrow^* I_i'$ for
    $i \in \{ 1,2\}$, such that $I_1', I_2' \subseteq C$. Since the
    choice of the representatives is irrelevant, we can assume
    that $I_1 = I_1'$ and $I_2 = I_2'$. Summing up, $I_1$ and $I_2$
    are consistent minimal enablings of the same event, hence by
    Lemma~\ref{le:es-to-fusion-domain}(\ref{le:es-to-fusion-domain:4}), $I_1 \leftrightarrow I_2$,
    i.e., $\eqclassir{I_1} = \eqclassir{I_2}$, as desired.

    \smallskip

  \item
    For the enabling relation, we have to show that if
    $X \vdash' \eqclassir{\esir{C}{e}}$ then
    $\theta_E(X) \vdash \theta(\eqclassir{\esir{C}{e}}) = e$.
    Assume $X \vdash' \eqclassir{\esir{C}{e}}$. According to the
    definition of the functor ${\ensuremath{\mathcal{E}}}$, this means
    that there exists
    $i \in \eqclassir{\esir{C}{e}}$ such that
    ${\eqclassir{\ir{i} \setminus \{i \}} \subseteq X}$.
    Let such $i \in \eqclassir{\esir{C}{e}}$ be $i = \esir{C'}{e}$
    with $C' \vdash_0 {e}$. We have
    \begin{center}
      $\ir{\esir{C'}{e}} \setminus \{ \esir{C'}{e} \} = \ir{C'} = \{
      \eqclassir{\esir{ C''}{ e''}} \mid \esir{ C''}{ e''} \subseteq C'\}$.
    \end{center}    
    Therefore from
    $\eqclassir{\ir{\esir{C'}{e'}} \setminus \{ \esir{C'}{e'} \}} \subseteq X$
    we deduce
    \begin{center}
      $\theta_E(\eqclassir{\ir{\esir{C'}{e'}} \setminus \{ \esir{C'}{e'} \}}) = C' \subseteq \theta_E(X)$.
    \end{center}
    Since $C' \vdash_0 e$, by monotonicity of enabling, we conclude
    $\theta_E(X) \vdash e$, as desired.
  \end{itemize}

  We prove the naturality of $\theta$ by showing that the
  diagram below commutes.
  \begin{center}
    \begin{tikzcd}[column sep=large]
      \ev{\dom{{E}_1}} 
      \arrow[d, below, "\ev{\dom{f}}" left]
      \arrow[r, "\theta_{{E}_1}"]
      & {{E}_1} \arrow[d, "f"]\\
      \ev{\dom{{E}_2}} 
      \arrow[r, "\theta_{{E}_2}" below]
      &
      {{E}_2}
    \end{tikzcd}
  \end{center}
  Consider $\eqclassir{\esir{C_1}{e_1}} \in \ev{\dom{{E}_1}}$.
  Recall that $\ev{\dom{f}}(\eqclassir{\esir{C_1}{e_1}})$ is computed
  by considering the image of the irreducible
  $\esir{C_1}{e_1}$ and of its predecessor, namely
  \begin{center}
    $\dom{f}(C_1) = f(C_1)$ and $\dom{f}(\esir{C_1}{e_1}) = f(C_1 \cup\{e_1\})$
  \end{center}
  If $f(e_1)$ is defined, then $f(C_1) \prec f(C_1 \cup\{e_1\})$ and
  $\ev{\dom{f}}(\eqclassir{\esir{C_1}{e_1}}) = f(e_1)$, otherwise
  $\ev{\dom{f}}(\eqclassir{\esir{C_1}{e_1}})$ is undefined. This means
  that in all cases, as desired
  \begin{center}
    $\ev{\dom{f}}(\eqclassir{\esir{C_1}{e_1}}) = f(e_1) =
    f(\theta_{{E}_1}(\eqclassir{\esir{C_1}{e_1}}))$.
  \end{center}
  
  \bigskip

  Vice versa, let $D$ be a weak prime domain. Recall from Definition~\ref{de:esfusdom} that 
  $\ev{D} = \langle E, \#, \vdash \rangle$ is defined as:
  \begin{itemize}
  \item $E =  \eqclassir{\ir{D}}$
  \item $e \# e'$ if there is no $d \in \compact{D}$ such that
    $e, e' \in \eqclassir{\ir{d}}$;
  \item $X \vdash e$ if there exists $i \in e$ such that
    $\eqclassir{\ir{i} \setminus \{ i \}} \subseteq X$.
  \end{itemize}
  and consider $\dom{\ev{D}}$.  Elements of $\compact{\dom{\ev{D}}}$
  are configurations of $C \in \conff{\ev{D}}$.  
  We can define the {unit} of the adjunction as
  \begin{center}
    $\begin{array}{lccc}
       \eta_{D} : & \compact{D} &  \to & \compact{\dom{\ev{D}}}\\
                  & d & \mapsto & \eqclassir{\ir{d}}
     \end{array}
     $
  \end{center}
  Observe that it is well defined, since by Lemma~\ref{le:comp-conf},
  $\eqclassir{\ir{d}}$ is a finite configuration of $\ev{D}$ and thus
  a compact element in $\compact{\dom{\ev{D}}}$.  The function is
  clearly monotone and bijective with inverse
  $\eta_D^{-1} : \compact{\dom{\ev{D}}} \to \compact{D}$ defined, for
  $C \in \compact{\dom{\ev{D}}} = \conff{\ev{D}}$ by letting
  $\eta_D^{-1}(C) = d$, where $d$ is the unique element, given by
  Lemma~\ref{le:comp-conf}, such that $C = \eqclassir{\ir{d}}$.
  By
  algebraicity of the domains, this function thus uniquely extends to
  an isomorphism $\eta_D : D \to \dom{\ev{D}}$.

  Finally, we prove the naturality of $\eta_D$. It is convenient to
  prove the naturality of the inverse, i.e., to show that the diagram
  below commutes.
  \begin{center}
    \begin{tikzcd}[column sep=large]
      \dom{\ev{D_1}}
      \arrow[d, "\dom{\ev{f}}" left]
      \arrow[r, "\eta_{D_1}^{-1}"]
      &
      D_1
      \arrow[d, "f"]\\
      \dom{\ev{D_2}}
      \arrow[r, "\eta_{D_2}^{-1}"]
      &
      D_2
    \end{tikzcd}
  \end{center}
  
  Let $C_1 \in \compact{\dom{\ev{D_1}}}$, namely
  $C_1 \in \conff{\ev{D_1}}$, and let $\eta_{D_1}^{-1}(C_1) = d_1$ be the
  element such that $C_1 = \eqclassir{\ir{d_1}}$.

  The construction offered by Lemma~\ref{le:chains} provides a chain
  \begin{center}
    ${d^0_1} = \bot \prec d_1^1 \prec d_1^2 \prec \ldots \prec d_1^n = d_1$
  \end{center}
  and, by the same lemma, if we take an irreducible $i_1^h \in
  \diff{d_1^h}{d_1^{h-1}}$ for $1 \leq h \leq n$ we have that
  $C_1 = \eqclassir{\ir{d_1}} = \eqclassir{\{i_1^1, \ldots, i_1^n \}}$.
  Therefore the image
  \begin{center}
    $\dom{\ev{f}}(C_1) = \{ \ev{f}(\eqclassir{j_1}) \mid
    \eqclassir{j_1} \in C_1 \} = \{ \ev{f}(\eqclassir{i_1^h}) \mid h
    \in \interval{n}\}$
  \end{center}
  is the set of equivalence classes of irreducibles
  $i_2^1, \ldots, i_2^k$ corresponding to 
  \begin{center}
    $f(d_1^0) = \bot \prec f(d_1^1) \prec f(d_1^2) \prec \ldots \prec
    f(d_1^n) = f(d_1)$
  \end{center}
  namely $i_2^j \in \diff{f(d_1^j)}{f(d_1^{j-1})}$, and,
  again, by Lemma~\ref{le:chains},
  $\eqclassir{\{i_2^1, \ldots, i_2^k\}} = \eqclassir{\ir{f(d_1)}}$.
  Summing up
  \begin{center}
    $\eta_{D_2}^{-1}(\dom{\ev{f}}(C_1)) = \eta_{D_2}^{-1}(\{ \eqclassir{i_2^h}
    \mid 1 \leq h \leq k \} \}) = f(d_1) = f(\eta_{D_1}^{-1}(C_1))$
  \end{center}
  as desired.
  We finally show that the above coreflection restricts to an
  equivalence between $\WDom$ and $\ces$.
  For this, just observe that, in the proof above, when ${E}$ is a connected
  {\esabbr}, then the morphism $\theta_{E}$ defined as
  \begin{center}
    $\begin{array}{lccc}
       \theta_{{E}} : & \ev{\dom{{E}}} &  \to & {E}\\
                             & \eqclassir{\esir{C}{e}} & \mapsto & e
     \end{array}
     $
  \end{center}
  is an isomorphism. We already know that
  it is surjective. We next show that it is also injective. In fact, if
  $\theta_{{E}}(\eqclassir{I}) =
  \theta_{{E}}(\eqclassir{I'})$
  then $I$ and $I'$ are minimal enablings of the same event, i.e.,
  $I = \eqclassir{\esir{C}{e}}$ and $I' = \eqclassir{\esir{C'}{e}}$. Since
  ${E}$ is a weak prime domain, $C \conn{e}^* C'$ and thus, by
  Lemma~\ref{le:es-to-fusion-domain}(\ref{le:es-to-fusion-domain:4}), $I \leftrightarrow^* I'$,
  i.e., $\eqclassir{I} =\eqclassir{I'}$.
  Proving that also the inverse is an {\esabbr} morphism is immediate, by
  exploiting the fact that the {\esabbr} is live.
\end{proof}

The above result indirectly provides a way of turning a general
{\esabbr} into a connected {\esabbr}.

\begin{corollary}[from general to connected {\esabbr}]
  The functors $\zconnes : \es \to \ces$ defined by
  $\zconnes = \zev \circ \zdom$ and the inclusion
  $\zinces : \ces \to \es$ form a coreflection.
\end{corollary}

\begin{proof}
  Immediate consequence of Theorem~\ref{th:es-dom-equivalence}.
\end{proof}

Explicitly, for any event structure $E$ the corresponding connected {\esabbr} 
$\zconnes(E) = \langle E', \vdash', \#' \rangle$ is defined as follows. 
The set of events is $E' = \{ \eqclass[\sim]{\esir{C}{e}} \mid C \vdash_0 e \}$, 
where $\sim$ is the least equivalence such that $\esir{C}{e} \sim \esir{C'}{e}$ 
if $\esir{C}{e}$ and $\esir{C'}{e}$ are consistent. 
Moreover $\eqclass[\sim]{\esir{C}{e}} \#' \eqclass[\sim]{\esir{C'}{e'}}$ if for all 
$\esir{C_1}{e} \sim \esir{C}{e}$ and $\esir{C_1'}{e'} \sim \esir{C'}{e'}$ the 
minimal enablings $\esir{C_1}{e}$ and $\esir{C_1'}{e_1'}$ are not consistent. 
Finally, for $X \subseteq E'$, $X \vdash' \eqclass[\sim]{\esir{C}{e}}$ if there 
exists $\esir{C'}{e} \sim \esir{C}{e}$ such that 
$C' \subseteq \{ e'' \mid \eqclass[\sim]{\esir{C''}{e''}} \in X\}$.

\begin{figure}
\subcaptionbox{
  \label{fi:summary-diag}
}{
  \begin{tikzcd}[row sep=small,ampersand replacement=\&]      
  \&
  {}
  \&
  \PDom \arrow[dd, hook]
  \arrow[dl, shift right=1.2ex]
  \arrow[dll, bend right=20, shift right=2ex]
  \\
  \ses \arrow[dd, hook]
  \arrow[urr, bend left=20, sloped, "\bot"]
  \&
  \pes \arrow[dd, hook]
  \arrow[l, hook]
  \arrow[ur, shift right=1.2ex, sloped, near end, "\sim"]
  \&
  {}\\
  {}
  \&
  {}
  \&
  \WDom
  \arrow[dl, shift right=1.2ex]
  \arrow[dll, bend right=20, shift right=2ex]\\
  \es
  \arrow[urr, bend left=20, sloped, "\bot"]
  \&
  \ces
  \arrow[l, hook]
  \arrow[ur, shift right=1.2ex, sloped, near end, "\sim"]
  \&
  {}
\end{tikzcd}
}
\subcaptionbox{
  \label{fi:summary-pict}
}{
\begin{tikzpicture}[x=8mm,y=10mm]
  \def\esel{(0,1)  ellipse [x radius=2, y radius=4]}
  \def\cesel{[black, pattern=north west lines, pattern color=gray] (0,0) ellipse [x radius=1.5, y radius=2.2]}
  \def\sesel{[black, pattern=north east lines, pattern color=gray](0,2) ellipse [x radius=1.5, y radius=2.2]}

 \def\wdomel {(6,0) ellipse [x radius=1.5, y radius=2.5]}
 \def\pdomel {[fill=gray!10] (6,2) ellipse [x radius=1.5, y radius=2.5]}

 \begin{scope}
   \clip (6,0) ellipse [x radius=1.5, y radius=2.5];
   \draw (6,2) ellipse [x radius=1.5, y radius=2.5];
 \end{scope}
 \begin{scope}
   \clip (6,2) ellipse [x radius=1.5, y radius=2.5];
   \draw [fill=gray!10] (6,0) ellipse [x radius=1.5, y radius=2.5];      
 \end{scope}

  \draw \esel;
  \draw \cesel;
  \draw \sesel;
  \draw \wdomel;

  \node [rectangle, rounded corners, fill=white] (ses) at (0,2.7)   {$\ses$};
  \node [rectangle, rounded corners, fill=white] (pes) at (0, 1)   {$\pes$};
  \node [rectangle, rounded corners, fill=white] (ces) at (0,-0.7)  {$\ces$};
  \node [rectangle, rounded corners, fill=white] (es)  at (0,-2.5)  {$\es$};

  \node (pesh) at (0,1.8)  {\phantom{$xx$}};
  \node (cesl) at (0,-1.5)  {\phantom{$xx$}};

  \node (seq)   at (3,2)    {$\top$};
  \node (pdomh) at (6,2)    {\phantom{$x$}};
  \node (pdoml) at (6,1.4)  {\phantom{$xxxx$}};

  \node [rectangle, rounded corners, fill=white] (pdom) at  (6,0.8)   {$\PDom$};
  \node (sadj)  at (3,1.2)   {$\sim$};
  \node (wdom)  at (6,-0.9)  {$\WDom$};
  \node (weq)   at (3,-1.2)  {$\sim$};
  \node (wdomh)  at (6,-1.5)  {\phantom{$xxxx$}};
  \node (wadj)  at (3,-1.9)  {$\bot$};
  \node (wdoml) at (6,-2.1)  {\phantom{$x$}};

  \draw [->] (ses) --  (pdomh);
  \draw [->] (pes) --  (pdom);
  \draw [->] (pdoml) -- (pesh);  

  \draw [->] (ces) --  (wdom);
  \draw [->] (es)  --  (wdoml);
  \draw [->] (wdomh) -- (cesl);  
\end{tikzpicture}
}

\caption{A summary of the relations among classes of {\esabbr} and domains.}
\label{fi:summary}
\end{figure}
 
An overall picture of the results discussed up to now can be found in Fig.~\ref{fi:summary}. The arrows from classes of event structures to domains are restrictions of the functor $\dom{\cdot}$, while the converse arrows are restrictions of the functor  $\ev{\cdot}$. The Venn diagram stresses the fact that prime {\esabbr} are exactly the {\esabbr} which are stable and connected (see Lemma~\ref{le:constapri}) showing how the notion of connectedness naturally emerges in the framework.

\section{Related Characterisations}
\label{se:characterisations}

In this section we present a characterisation of our proposal in terms
of a formalism reminiscent of the prime event structures with
equivalence of~\cite{win2017,VismeW19}. Moreover, we discuss and
formalise the relation of our work with alternative characterisations
of the domains of (prime) event structures proposed in the literature,
based on intervals and on asynchronous graphs.

\subsection{Prime Event Structures with Equivalence}
\label{ss:pes-equiv}

The previous sections showed that the domains of configurations of
unstable {\esabbr} are weak prime domains, i.e., they satisfy the same
conditions as those of prime domains but only up to the equivalence
induced by interchangeability.
Symmetrically, this suggests the possibility of viewing unstable
{\esabbr} as stable ones up to some equivalence on events.
In this section we consider a formalisation for such a view,
leading to a set up that is closely related to the framework devised
in~\cite{win2017,VismeW19},
which we also call \emph{prime event structures with equivalence}
for the space of this article, since no confusion can arise.

In Section~\ref{ss:es} we mentioned that in prime {\esabbr} a global
notion of causality can be used in place of the enabling. We next
recall the formal definition. We also introduce a notation for direct
(i.e., non-inherited) conflict that will play a role later.

\begin{definition}[causality/direct conflict in prime {event
    structures}]
  Let $P = \langle E, \vdash, \# \rangle$ be a prime {\esabbr}.  Given
  an event $e \in E$, the unique $C \in \conf{P}$ such that
  $C \vdash_0 e$ is called the set of \emph{strict causes} of $e$ and
  denoted by $\scauses{e}$, while the set of causes is
  $\causes{e} = \scauses{e} \cup \{e\}$. The \emph{strict causality
    relation} $<$ is defined by $e' < e$ if $e' \in\, \scauses{e}$,
  and, as usual, we denote by $\leq$ the reflexive closure of $<$. We
  say that $e, e' \in E$ are in \emph{direct conflict}, written
  $e \#_d e'$, when $e \# e'$ and $\scauses{e} \cup \{e'\}$,
  $\scauses{e'} \cup \{e\}$ are consistent.
\end{definition}

We next introduce our notion of prime {\esabbr} with
equivalence. Given a prime {\esabbr} $P$ with an equivalence over the
set of events ${\sim}\subseteq E \times E$, we say that a subset
$X \subseteq E$ is \emph{$\sim$-saturated} if for all $e \in X$ and
$e' \in E$, if $e \sim e'$ and $\scauses{e'} \subseteq X$  then
$e' \in X$. Since the intersection of saturated sets is saturated, given a set $X$ we can always consider the smallest saturated superset of $X$, called the  \emph{saturation} of $X$ and denoted $\sat{X}$.

\begin{definition}[prime {event structures} with equivalence]
  \label{de:pes-eq}
  A \emph{prime {\esabbr} with equivalence} ({\eseqabbr} for short) is
  a pair $\langle P, \sim \rangle$ where
  $P = \langle E, \vdash, \# \rangle$ is a prime {\esabbr} and $\sim$
  is an equivalence on $E$ such that for all $e, e', e_1, e_1' \in E$
  \begin{enumerate}
    
  \item
    \label{de:pes-eq:1}
    if
    $\eqclass[\sim]{\causes{e}} \subseteq \eqclass[\sim]{\causes{e'}}$
    then $e \leq e'$; if in addition $e \sim e'$ then $e=e'$;

  \item
    \label{de:pes-eq:2}
    if $e \sim e'$ and $\scauses{e} \cup \scauses{e'}$ consistent then
    $\neg (e \# e')$.
  
  \item
    \label{de:pes-eq:3}
    if $e \sim e'$, $e_1 \sim e_1'$, and $e \#_d e_1$ then $e' \# e_1'$.
  \end{enumerate}
  We say that $\langle P, \sim \rangle$ is {\emph{connected}}
  if ${\sim} = (\sim \setminus \#)^*$. A morphism of {\eseqabbr}
  $f : \langle P_1, \sim_1 \rangle \to \langle P_{{2}}, \sim_2 \rangle$ is
  an {\esabbr} morphism $f : P_1 \to P_2$ such that for all
  $e_1, e_1' \in P_1$, $e_1 \sim_1 e_1'$ iff
  $f(e_1) \sim_2 f(e_1')$. We denote by $\epes$ the corresponding
  category.
\end{definition}

An {\esabbr} with equivalence is thus just an {\esabbr} equipped with
an equivalence on events. Condition (\ref{de:pes-eq:1}) essentially
says that an event is determined by the equivalence classes of events
in its causal history. In particular, as a consequence, if
$\scauses{e} \subseteq\, \scauses{e'}$ and $e \sim e'$ then $e=e'$,
which intuitively means that distinct equivalent events must
correspond to different enablings of the same event. Moreover, it
implies that the set $\causes{e}$ is $\sim$-saturated and thus it is a
configuration (see Definition~\ref{de:epes-conf} and
Lemma~\ref{le:hist-conf}).
Conditions (\ref{de:pes-eq:2}) and (\ref{de:pes-eq:3}) essentially say that equivalent events can have different conflicts only for the fact that their minimal enablings have different conflicts.
Connectedness 
amounts to the fact that equivalent events must be connected by a 
chain of equivalences going through consistent events. 
We next introduce a notion of configuration.

\begin{definition}[configurations]
  \label{de:epes-conf}
  Let $\langle P, \sim \rangle$ be an {\eseqabbr}. Then
  $\conf{\langle P, \sim \rangle} = \{ C \mid C \in \conf{P}\ \land\
  C\mbox{ $\sim$-saturated}\}$.
\end{definition}

In words, a configuration of a prime {\esabbr} with equivalence is a
configuration $C$ of the underlying event structure, where all events
enabled in $C$ that are equivalent to some event already in $C$ are also in
$C$. Thus equivalent events may have different minimal enablings,
but whenever a configuration contains the causes of two equivalent
events, their executions cannot be taken apart.

\begin{lemma}[histories are configurations]
  \label{le:hist-conf}
  Let $\langle P, \sim \rangle$ be a {\eseqabbr}. For all $e \in E$,
  $\causes{e}$ is a configuration.
\end{lemma}

\begin{proof}
  Let $e \in E$ be any event. We have to show that
  $\causes{e}$ is saturated. If there are
  $e' \in \causes{e}$ and $e'' \sim e'$ such that
  $\scauses{e''} \subseteq \causes{e}$ then
  $\eqclass[\sim]{\causes{e''}} \subseteq
  \eqclass[\sim]{\causes{e}}$ and hence, by condition (\ref{de:pes-eq:1}) in
  Definition~\ref{de:pes-eq}, $e'' \leq e$ which means
  $e'' \in \causes{e}$.
\end{proof}
   
As an example, the connected {\esabbr} of our running example (see
Fig.~\ref{fi:running}), corresponds to the prime {\esabbr} with
equivalence in {Fig.~\ref{fi:running-pes-equivalence:a}}, where we have two
distinct copies of event $c$, namely $c_a \sim c_b$, corresponding to
the possibile minimal enablings. Graphically, causality is represented
by a straight directed line. The corresponding domain of configurations is depicted in {Fig.~\ref{fi:running-pes-equivalence:b}}. Note that $C=\{ a, b, c_a\}$ is not a configuration despite the fact that it is {downward} closed, since it is not $\sim$-saturated: event $c_b$ is missing, but its causes {$\{b\}$} are in $C$.%

\begin{figure}[h!]
  \begin{center}
  \subcaptionbox{
    \label{fi:running-pes-equivalence:a}
  }{
    \begin{tikzpicture}[node distance=5mm, >=stealth',x=10mm,y=6mm]
      \node at (1,0) (Gs)  {};
      \node at (0,1) (Ga)  {$a$};
      \node at (2,1) (Gb)  {$b$};
      \node at (1,2.1) (Gab) {};
      \node at (0,2.7) (Gac)  {$c_a$};
      \node at (1,2.7) (sim)  {$\sim$};
      \node at (2,2.7) (Gbc)  {$c_b$};
      \node at (1,4) (Gc)  {};
      \draw [->] (Ga) -- (Gac);
      \draw [->] (Gb) -- (Gbc);
    \end{tikzpicture}
  }
  \hspace{15mm}
  \subcaptionbox{
    \label{fi:running-pes-equivalence:b}
  }
  {  
    \begin{tikzpicture}[node distance=5mm, >=stealth',x=23mm,y=6mm]
      \node at (1,0) (Gs)  {$\emptyset$};
      \node at (0,1) (Ga)  {$\{a\}$};
      \node at (2,1) (Gb)  {$\{b\}$};
      \node at (1,2.1) (Gab) {$\{a,b\}$};
      \node at (0,2.7) (Gac)  {$\{a,c_a\}$};
      \node at (2,2.7) (Gbc)  {$\{b,c_b\}$};
      \node at (1,4) (Gc)  {$\{a,b,c_a,c_b\}$};
      \draw [->] (Gs) -- (Ga);
      \draw [->] (Gs) -- (Gb);
      \draw [->] (Ga) -- (Gab);
      \draw [->] (Gb) -- (Gab);
      \draw [->] (Ga) -- (Gac);
      \draw [->] (Gb) -- (Gbc);
      \draw [->] (Gac) -- (Gc);
      \draw [->] (Gbc) -- (Gc);
      \draw [->] (Gab) -- (Gc);
    \end{tikzpicture}
    }
  \end{center}
  \caption{A prime {\esabbr} with equivalence and its domain of configurations.}
  \label{fi:running-pes-equivalence}
\end{figure}

Our definition of {\eseqabbr} is similar to that
in~\cite{win2017,VismeW19}.
Concerning {configurations}, while~\cite{win2017,VismeW19} identifies unambiguous
configurations where there is a unique representative for each
equivalence class, here instead we saturate including \emph{all}
equivalent events that are not in conflict.

We finally observe that the constructions above can be ``translated'' 
into constructions that relate directly {\eseqabbr} and weak prime domains.

\begin{proposition}[weak prime domain for {\eseqabbr}]
  \label{pr:wpd-eseq}
  Let $\langle P, \sim \rangle$ be a {\eseqabbr}. Then
  $\domeq{\langle P, \sim \rangle} = \langle  \conf{\langle P, \sim \rangle}, \subseteq \rangle$ 
  is a weak prime domain.
  Conversely, if $D$ is a weak prime domain then
  $\eveq{D} = \langle \langle \ir{D}, \#, \vdash \rangle,
  \leftrightarrow^* \rangle$ with conflict and enabling defined by
  \begin{itemize}
  \item $i_1 \# i_2$ if $\{ i_1, i_2\}$ not consistent;
  \item $X \vdash i$ if
    $X \supseteq \ir{i}\setminus\{i\}$.
  \end{itemize}
  is an {\eseqabbr}.
\end{proposition}

\begin{proof}
  Let $\langle P, \sim \rangle$ be a {\eseqabbr}.
  Then it is easy to see that the irreducibles of
  $\domeq{\langle P, \sim \rangle}$ are the minimal enablings
  $\causes{e}$ for $e \in E$.
  Moreover, given a set of pairwise consistent configurations
  $X \subseteq \conf{\langle P, \sim \rangle}$, the join $\bigsqcup X$
  is the saturation of their union.
  Interchangeability is given by
  $\causes{e} \leftrightarrow \causes{e'}$
  if $e \sim e'$ and $\neg (e \# e')$.
  Using these fact it is almost immediate to conclude that
    $\domeq{\langle P, \sim \rangle}$ is a weak prime domain. Let us
    first observe that $\domeq{\langle P, \sim \rangle}$ is
    {\wi} (Definition~\ref{de:well-interchange}):
    \begin{itemize}
    \item Condition (\ref{de:well-interchange:1}) requires that for
      all $e, e' \in P$ if $\causes{e} \leftrightarrow^* \causes{e'}$
      and $\scauses{e} \cup \scauses{e'}$ consistent then
      $\causes{e} \leftrightarrow \causes{e'}$. Observe that
      $\causes{e} \leftrightarrow^* \causes{e'}$ implies $e \sim
      e'$. Moreover, by condition (\ref{de:pes-eq:2}) in
      Definition~\ref{de:pes-eq}, $\scauses{e} \cup \scauses{e'}$
      consistent implies $\neg (e \# e')$. Hence we conclude
      $\causes{e} \leftrightarrow \causes{e'}$.

    \item Condition~(\ref{de:well-interchange:2}) is an easy
      consequence of condition~(\ref{de:pes-eq:3}) of
      Definition~\ref{de:pes-eq}. In fact, let
      $e, e', e_1, e_1' \in E$ such that
      $\causes{e} \leftrightarrow^* \causes{e'}$,
      $\causes{e_1} \leftrightarrow^* \causes{e_1'}$, i.e.,
      $e \sim e'$ and $e_1 \sim e_1'$. Assume moreover that the sets
      $\{ \causes{e'}, \causes{e_1'}\}$,
      $\{ \causes{e}, \scauses{e_1}\}$,
      $\{ \scauses{e}, \causes{e_1}\}$ are consistent, meaning that $\causes{e'} \cup \causes{e_1'}$, $\causes{e} \cup \scauses{e_1}$, 
      $\scauses{e} \cup \causes{e_1}$ are so. From the consistency of $\causes{e'} \cup \causes{e_1'}$ we have $\neg (e' \# e_1')$. Moreover, the consistency of $\causes{e} \cup \scauses{e_1}$, 
      $\scauses{e} \cup \causes{e_1}$ implies that if $e \# e_1$ then the conflict would be direct and  this would violate condition~(\ref{de:pes-eq:3}) of
      Definition~\ref{de:pes-eq}. Hence we must have $\neg (e \# e_1)$, i.e., $\{  \causes{e}, \causes{e_1}\}$ consistent, as desired.
    \end{itemize}

    Finally, we show that all irreducibles are weak prime. Let $e \in P$, consider the irreducible $\causes{e}$ and a consistent set of configurations $X \subseteq \conf{\langle P, \sim \rangle}$. Assume that $\causes{e} \subseteq \bigsqcup X$. This means that $e$ is in the saturation of $\bigcup X$, which in turn means that there is $C \in X$ and $e' \in C$, whence $\causes{e'} \subseteq C$, such that $e' \sim e$. Since $e, e' \in  \bigsqcup X$, they are consistent, hence $\causes{e} \leftrightarrow \causes{e'}$. Summing up $\causes{e'} \subseteq C$ and $e \sim e'$, as desired.

  Conversely, let $D$ be a weak  prime domain. Observe that the causal
  order  in $\eveq{D}$  is  the  restriction of  the  domain order  to
  irreducibles. Condition (\ref{de:pes-eq:1}) in Definition~\ref{de:pes-eq} is an immediate consequence of Proposition~\ref{pr:unique-dec}.

    Condition (\ref{de:pes-eq:2}) is immediately implied by condition
    (\ref{de:well-interchange:1}) in the definition of {\wi} domain
    (Definition~\ref{de:well-interchange}).
    
    Concerning condition (\ref{de:pes-eq:3}), observe that it becomes:
    for $i, i', i_1, i_1' \in \ir{D}$, if $i \leftrightarrow^* i'$,
    $i_1 \leftrightarrow^* i_1'$, $i, i_1$ not consistent and
    $i', i_1'$ consistent then either $\ir{\pred{i}} \cup \{i_1\}$ or
    $\ir{\pred{i_1}} \cup \{i\}$ not consistent. In turn this is
    easily seen to be equivalent to condition
    (\ref{de:well-interchange:2}) in the definition of {\wi} domain
    (Definition~\ref{de:well-interchange}).
\end{proof}

The correspondence above can be translated to an analogous
correspondence between {\eseqabbr} and unstable {\esabbr}. It is
however impossible to make such correspondence functorial essentially for
the same reason why~\cite{win2017,VismeW19} resorts to a pseudo-adjunction. 
We try to enucleate the problem by showing a correspondence between
(unstable) event structures and {\eseqabbr}.

\begin{definition}[from {\esabbr} to {\eseqabbr} and back]
  Let $\langle P, \sim \rangle$ be an {\eseqabbr}, where
  $P = \langle E, \vdash, \# \rangle$. The corresponding {\esabbr} is
  $\fuse{\langle P, \sim\rangle} = \langle \quotient{E}{\sim},
  \quotient{\vdash}{\sim}, \quotient{\#}{\sim} \rangle$, with
  $\quotient{\vdash}{\sim}$ and $\quotient{\#}{\sim}$ defined by
  \begin{itemize}
  \item $\eqclass[\sim]{X} \quotient{\vdash}{\sim} \eqclass[\sim]{e}$ when $X \vdash e$;
  \item $\eqclass[\sim]{e} \quotient{\#}{\sim} \eqclass[\sim]{e'}$ 
           when $e_1 \# e_1'$ for all $e_1 \in \eqclass[\sim]{e}$ and $e_1' \in \eqclass[\sim]{e'}$.
  \end{itemize}
  Conversely, given an {\esabbr} $P = \langle E, \vdash, \# \rangle$ the
  corresponding {\eseqabbr} is $\unf{P} = \langle Q, \sim \rangle$, with
  $Q = \langle E', \vdash', \#' \rangle$ defined by
  \begin{itemize}
  \item
    $E' = \{ \esir{C}{e} \mid C \in \conf{E}\ \land\ e\in E\ \land\ C
    \vdash_0 e \}$;
  \item $X \vdash' \esir{C}{e}$ if
    $C \subseteq \bigcup \{ C' \cup \{e'\} \mid \langle C', e' \rangle
    \in X\} $;
  \item $\esir{C}{e} \#' \esir{C'}{e'}$ if
    $C \cup C' \cup \{e, e'\}$ is not consistent.
  \end{itemize}
  and the equivalence is defined by $\esir{C}{e} \sim \esir{C'}{e}$ for all $C, C'$ such that $C \vdash_0 e$ and $C' \vdash_0 e$.
\end{definition}

We can easily show, exploiting Proposition~\ref{pr:wpd-eseq},
that the constructions above produce well-defined structures and map
connected structures to connected structures.
Moreover, the two constructions are inverse of each other.

\begin{proposition}
  \label{pr:epes-inverse}
  Let $\langle P, \sim \rangle$ be an {\eseqabbr}.
  Then $\langle P, \sim \rangle$ and
  $\unf{\fuse{\langle P, \sim\rangle}}$ are isomorphic. Dually, 
  let $P = \langle E, \vdash, \# \rangle$ be an
  {\esabbr}. Then
  $\fuse{\unf{P}}$ and $P$ are isomorphic.
\end{proposition}

\begin{proof}
  Let $\langle P, \sim \rangle$ be an {\eseqabbr}.
  Recall that events in
  $\unf{\fuse{\langle P, \sim\rangle}}$ are minimal
  enablings in $\fuse{\langle P, \sim\rangle}$.
  By definitions of $\fuse{\langle P, \sim\rangle}$, for all $e \in P$
  we have
  $\eqclass[\sim]{X} \vdash_{\fuse{\langle P, \sim\rangle}}
  \eqclass[\sim]{e}$ when $X \vdash e$. Therefore
  $\eqclass[\sim]{\scauses{e}} \vdash_{\fuse{\langle P, \sim\rangle}}
  \eqclass[\sim]{e}$, and this enabling is minimal since, by
  Definition~\ref{de:pes-eq}(\ref{de:pes-eq:1}), whenever $e' \sim e$
  and $\eqclassir{\scauses{e'}} \subseteq \eqclassir{\scauses{e}}$ we
  have $e=e'$. And, again relying on the definition of enabling, one
  sees that all minimal enablings are of this shape.
  Therefore we can define
  $c : \langle P, \sim \rangle \to \unf{\fuse{\langle P,
      \sim\rangle}}$ by
  $c(e) =
  \esir{\eqclass[\sim]{\scauses{e}}}{\eqclass[\sim]{e}}$.
  By the previous arguments it is a bijection and it can be shown $c$
  to be an isomorphism of {\eseqabbr}.

  \smallskip

  Conversely, let $\langle E, \vdash, \# \rangle$ be an an
  {\esabbr}.
  According to the definition, events in $\unf{E}$ are minimal
  enablings $\esir{C}{e}$ in $E$, and they are equivalent when they are minimal
  enablings of the same event.
  Then events in $\fuse{\unf{E}}$ are just equivalence classes of events in $\unf{E}$. Therefore we can define $u : E \to \fuse{\unf{E}}$ by
  $u(e) = \{ \esir{C}{e} \mid C \in \conf{E}\ \land\ C \vdash_0 e \}$. It is
  immediate to see that it is a bijection and 
  an isomorphism of {\esabbr}.
\end{proof}

Observe that the construction from {\eseqabbr} to {\esabbr} can be
easily turned into a functor $\zfuse : \epes \to \es$. In fact, given a
morphism
$f : \langle P_1, \sim_1 \rangle \to \langle P_2, \sim_2\rangle$ we
can let $\fuse{f}(\eqclass[\sim_1]{e_1}) = \eqclass[\sim_2]{f(e_1)}$.

Instead, making the converse construction from {\esabbr} to
{\eseqabbr} functorial is problematic. In fact, consider the {\esabbr} of the running example
$E = \{ a, b, c \}$, with $\emptyset \vdash_0 a$,
$\emptyset \vdash_0 b$ and $\{ a, b \} \vdash_0 c$ and the
{\esabbr} with events $E' = \{ a', b', c' \}$ with
$\emptyset \vdash_0 a'$, $\emptyset \vdash_0 b'$ and
$\{ a' \} \vdash_0 c'$ and $\{ b' \} \vdash_0 c'$ and the morphism
$f : E \to E'$ with $f(x) = x'$ for $x \in \{ a, b, c\}$. Then
$\unf{E} = \{ \esir{\emptyset}{a}, \esir{\emptyset}{b},
\esir{\{a, b\}}{c} \}$ and
$\unf{E'} = \{ \esir{\emptyset}{a'}, \esir{\emptyset}{b'},
\esir{\{a'\}}{c'}, \esir{\{b'\}}{c'} \}$. Observe that, while clearly
$\unf{f}(\esir{\emptyset}{a}) = \esir{\emptyset}{a'}$ and
$\unf{f}(\esir{\emptyset}{b}) = \esir{\emptyset}{b'}$, when we come
to $\unf{f}(\esir{\{a,b\}}{c})$ we can define it as one of the
two equivalent events $\esir{\{a'\}}{c'}$ and $\esir{\{b'\}}{c'}$.

The solution offered by~\cite{win2017,VismeW19} is to move towards 
pseudo-functors,
i.e., considering two {\eseqabbr} morphisms $g, g' : P_1 \to P_2$ equivalent
if $g(e_1) \sim_2 g'(e_1)$ for all $e_1\in P_1$ and requiring that functors 
are defined only up-to morphism equivalence. 
Indeed, it is easy to see that the two possible choices for $f$ above lead to
equivalent morphisms.

\subsection{Relation with Interval Based Characterisations}
\label{ss:intervals}

The correspondence between event structures and domains has been often
studied in the literature by relying on the notion of
interval~\cite{Winskel:phd,NPW:PNES,Win:ES,Dro:ESD}.

\begin{definition}[interval]
  \label{de:interval}
  Let $D$ be a domain. An \emph{interval} is a pair $\dint{d}{d'}$ of
  elements of $D$ such that $d \prec d'$. The set of intervals of $D$
  is denoted by $\IntSet{D}$.
  Given two intervals $\dint{c}{c'}, \dint{d}{d'} \in \IntSet{D}$ we define
  \begin{center}
    $\dint{c}{c'} \leq \dint{d}{d'}$ \quad if $(c = c' \sqcap d)\ \wedge \ (c'
    \sqcup d = d')$, 
  \end{center}
  and we let $\sim$ be the equivalence obtained as the
  symmetric and transitive closure of $\leq$.
\end{definition}
It can be shown that $\leq$ is a partial order on intervals and thus $\sim$
is indeed an equivalence. An interval represents a pair of elements
differing only for a ``quantum'' of information, intuitively the
execution of an event. The equivalence $\sim$ is intended to identify
intervals corresponding to the execution of the same event in diffent states.
Indeed, in~\cite{NPW:PNES} it is shown that for prime domains there is a
bijective correspondence between $\sim$-classes of intervals and
complete primes.
In weak prime domains we can establish a similar correspondence, with $\leftrightarrow^*$-classes of irreducibles 
playing the role of the primes.

\begin{lemma}[intervals vs. irreducibles]
  \label{le:int-vs-irr}
  Let $D$ be a weak prime domain. Define $\inir : \quotient{\IntSet{D}}{\sim} \to \quotient{\ir{D}}{\leftrightarrow^*}$ by
  \begin{center}
    $\inir(\dint{d}{d'}_\sim) = \eqclassir{i}$,
  \end{center}
  where $i$ is any element in $\diff{d'}{d}$. Then $\inir$ is a bijection, whose inverse is  $\irin :  \quotient{\ir{D}}{\leftrightarrow^*} \to \quotient{\IntSet{D}}{\sim}$
  defined by
  \begin{center}
    $\irin(\eqclassir{i}) = \dint{\pred{i}}{i}_\sim$.
  \end{center}
\end{lemma}

\begin{proof}
  We first observe that $\inir$ is well-defined, i.e., if
  $\dint{c}{c'} \sim \dint{d}{d'}$ are equivalent intervals then for
  all $i \in \diff{c'}{c}$, $i' \in \diff{d'}{d}$ it holds
  $i \leftrightarrow i'$. This follows by noting that if
  $\dint{c}{c'} \leq \dint{d}{d'}$, $i \in \diff{c'}{c}$ and
  $i' \in \diff{d'}{d}$ then $i \leftrightarrow i'$.
  In order to prove the last assertion, observe that since
  $i \in \ir{c'}$ we have $i \sqsubseteq c' \sqsubseteq d'$, thus
  $i \in \ir{d'}$.  Moreover, $i \not\in \ir{d}$, otherwise, by
  $i \sqsubseteq d$, $i \sqsubseteq c'$ and $c = d \sqcap c'$, we
  would get $i \sqsubseteq c$, contradicting the assumption that
  $i \in \diff{c'}{c}$. Hence $i \in \diff{d'}{d}$ and by
  Lemma~\ref{le:prec-irr-b}(\ref{le:prec-irr-b:3}) we conclude.

  Also the converse map $\irin$ is well-defined. This follows from the
  observation that for all irreducibles $i, i' \in \ir{D}$ if
  $i \leftrightarrow i'$ then
  $\dint{\pred{i}}{i}, \dint{\pred{i'}}{i'} \leq \dint{\pred{i} \sqcup
    \pred{i'}}{i \sqcup i'}$ and thus
  $\dint{\pred{i}}{i} \sim \dint{\pred{i'}}{i'}$. Let us prove, for
  instance, that
  \begin{center}
    $\dint{\pred{i}}{i} \leq \dint{\pred{i} \sqcup \pred{i'}}{i \sqcup
      i'}$.
  \end{center}
  Since $i \leftrightarrow i'$, surely
  $\pred{i} \sqsubseteq \pred{i} \sqcup \pred{i'}$ and
  $\pred{i} \prec i$, hence by Lemma~\ref{le:inf}, we deduce
  $i \sqsubseteq \pred{i} \sqcup \pred{i'}$ or
  $\pred{i} = i \sqcap (\pred{i} \sqcup \pred{i'})$. The first
  possibility, $i \sqsubseteq \pred{i} \sqcup \pred{i'}$, by the fact
  that $i$ is irreducible leads to $i \sqsubseteq \pred{i'}$ (since
  $i \sqsubseteq \pred{i}$ is clearly false). Thus
  $i \sqcup \pred{i'} = \pred{i'} \prec i' \sqsubseteq \pred{i} \sqcup
  i'$, that, by Lemma~\ref{le:eq-char}(\ref{le:eq-char:4}),
  contradicts $i \leftrightarrow i'$.
  Hence the second possibility must hold, i.e.,
  $\pred{i} = i \sqcap (\pred{i} \sqcup \pred{i'})$. Moreover, again
  by Lemma~\ref{le:eq-char}(\ref{le:eq-char:4}), we have
  $i \sqcup (\pred{i} \sqcup \pred{i'}) = i \sqcup i'$. Hence
  $\dint{\pred{i}}{i} \leq \dint{\pred{i} \sqcup \pred{i'}}{i \sqcup
    i'}$ as desired.

  \bigskip
  
  The two maps are inverse each other.
  \begin{itemize}
  \item If $\dint{d}{d'} \in \IntSet{D}$ and $i \in \diff{d'}{d}$ then
    $\dint{d}{d'} \sim \dint{\pred{i}}{i}$.\\
    Observe that $d \sqcup i = d'$ by
    Lemma~\ref{le:prec-irr-b}(\ref{le:prec-irr-b:1}). Moreover, in
    order to show that $d \sqcap i = \pred{i}$, note that, since
    $i \in \diff{d'}{d}$ and, by
    Lemma~\ref{le:prec-irr-b}(\ref{le:prec-irr-b:1bis}), the set
    $\diff{d'}{d}$ is flat, we have that $\pred{i} \sqsubseteq d$.
    Moreover $\pred{i} \prec i$, therefore by Lemma~\ref{le:inf},
    $\pred{i} = d \sqcap i$, as desired.

  \item If $i \in \ir{D}$ and $i' \in \dint{\pred{i}}{i}$ then
    $i \leftrightarrow i'$.\\
    Just observe that $i \in \dint{\pred{i}}{i}$ and then use
    Lemma~\ref{le:prec-irr-b}(\ref{le:prec-irr-b:3}).
  \end{itemize}
  
\end{proof}

In ~\cite{Winskel:phd,Dro:ESD} the domain of
configurations of general event structures with binary conflict is characterised in terms of intervals. It is shown (see, e.g.,~\cite[Theorem 3.3.3]{Winskel:phd}), that given a an event structure with binary conflict, the domain of configuration is an 
algebraic complete partial order where the
following axioms hold

\begin{description}
\item[(F)] for all $d \in \compact{D}$ the set $\principal{d}$ is finite;

\item[(C)] for all $x, y, z \in \compact{D}$, if $x \prec y$,
  $x \prec z$, $\{y, z\}$ consistent, and $y \neq z$ then there exists
  $y \sqcup z$ and $y \prec y \sqcup z$ and $z \prec y \sqcup z$;

\item[(R)] for all intervals $\dint{x}{y}$, $\dint{x}{z}$ if
  $\dint{x}{y} \sim \dint{x}{z}$ then $y=z$;

\item[(V)] for all $x,x',y,y',x'',y'' \in \compact{D}$ if
  $\dint{x}{x'} \sim \dint{y}{y'}$, $\dint{x}{x''} \sim \dint{y}{y''}$, and $\{x',x''\}$
  consistent then $y',y''$ consistent.
\end{description}

Conversely, in~\cite{Dro:ESD} an explicit construction of the {\esabbr} corresponding to
a domain is provided. Given
$d \in \compact{D}$, let
$s(d) = \{ \dint{c}{c'}_{\sim} \mid c' \sqsubseteq d \}$.

\begin{definition}[{event structure} from a domain~\cite{Dro:ESD}]
  \label{de:evwd}
  Given a domain $D$ satisfying the axioms (F), (C), (R), (V), the
  corresponding {\esabbr} with binary conflict is defined as
  $\evwd{D} = (E, \#, \vdash)$ where
  \begin{itemize}
    
  \item
    $E = \quotient{\IntSet{D}}{\sim}$;
  \item $\dint{c}{c'}_\sim \# \dint{d}{d'}_\sim$ if for all
    $\dint{c_1}{c_1'}$, $\dint{d_1}{d_1'}$ such that
    $\dint{c_1}{c_1'} \sim \dint{c}{c'}$ and
    $\dint{d_1}{d_1'} \sim \dint{d}{d'}$ the set $\{c_1',d_1'\}$ is
    not consistent;
    
  \item 
    for $X \subseteq E$, $X \vdash \dint{c}{c'}_\sim$ if
    $s(c_1) \subseteq X$ for some interval
    $\dint{c_1}{c_1'} \sim \dint{c}{c'}$.
  \end{itemize}
\end{definition}

The above construction produces an event structure with binary conflict that is mapped back to the original domain (see, e.g.,~\cite[Corollary 2.10]{Dro:ESD}).

\begin{theorem}
  \label{th:evwd}
  Let $D$ be a domain satisfying axioms (F), (C), (R), (V). Then
  $\dom{\evwd{D}}$ is isomorphic to $D$.
\end{theorem}

We can build on the above results to show that the domains satisfying
axioms (F), (C), (R) and (V) are exactly the weak prime domains.

\begin{proposition}[weak prime domains and intervals]
  Let $D$ be a domain. Then $D$ is a weak prime domain iff $D$
  satisfies axioms (F), (C), (R) and (V).
\end{proposition}

\begin{proof}
  Let $D$ be a domain satisfying axioms (F), (C), (R) and (V).  By
  Theorem~\ref{th:evwd}, $\dom{\evwd{D}} \simeq D$. Since, by
  Proposition~\ref{pr:es-to-dom}, the set of configurations of any
  event structure forms a weak prime domain, we conclude that $D$ is
  weak prime.

  For the converse, let $D$ be a weak prime domain. By
  Theorem~\ref{th:es-dom-equivalence}, we have that
  $\dom{\ev{D}} \simeq D$ and thus, since by~\cite{Winskel:phd,Dro:ESD},
  the domain of configuration of an event structure with binary
  conflict satisfies axioms (F), (C), (R) and (V), we conclude.
\end{proof}

Moreover, relying on Lemma~\ref{le:int-vs-irr}, we can
show that the event structures associated with a domain
in~\cite{Dro:ESD} (Definition~\ref{de:evwd}) and in our work (Definition~\ref{de:esfusdom})
coincide.

\begin{proposition}
  \label{pr:es-int}
  Let $D$ be a weak prime domain. Then $\ev{D}$ and $\evwd{D}$ are isomorphic.
\end{proposition}

\begin{proof}
  By Lemma~\ref{le:int-vs-irr}, the function
  $\inir : \quotient{\IntSet{D}}{\sim} \to \quotient{\ir{D}}{\leftrightarrow^*}$ is a
  bijection. Note that $\quotient{\IntSet{D}}{\sim}$ and
  $\quotient{\ir{D}}{\leftrightarrow^*}$ are the sets of events respectively of
  $\ev{D}$ and $\evwd{D}$. We next show that $\inir$ is an isomorphism
  of event structures.

  \medskip

  Let $e_1, e_2$ be events in $\evwd{D}$. We show that $e_1 \# e_2$
  iff $\inir(e_1) \# \inir(e_2)$.
  
  If $\neg (e_1 \# e_2)$, from Definition~\ref{de:evwd}, we get that
  there exist $\dint{c_1}{c_1'} \in e_1$ and
  $\dint{c_2}{c_2'} \in e_2$ such that $\{ c_1', c_2' \}$ is
  consistent. Let $d \in D$ be an upper bound, i.e.,
  $c_1', c_2' \sqsubseteq d$. Now, $\inir(e_j) = \eqclassir{i_j}$ for
  $i_j \in \diff{c_j}{c_j'}$, for $j \in \{1,2\}$. Clearly,
  $i_1, i_2 \in \ir{d}$ whence
  $\eqclassir{i_1}, \eqclassir{i_2} \subseteq \eqclassir{\ir{d}}$ and
  thus, according to Definition~\ref{de:esfusdom}, we have
  $\neg (\eqclassir{i_1} \# \eqclassir{i_2})$, as desired.
  The argument can be reversed to prove that if
  $\neg( \inir(e_1) \# \inir(e_2) )$ then $\neg (e_1 \# e_2)$.
  
  \medskip
  
  Concerning the enabling relation, we show that $X \vdash e$ in
  $\evwd{D}$ iff $\inir(X) \vdash \inir(e)$ in $\ev{D}$.
  Assume that $X \vdash e$ in $\evwd{D}$. This means that there exists
  $\dint{c}{c'} \in e$ such that
  $s(c) = \{ \dint{d}{d'}_{\sim} \mid d' \sqsubseteq c \} \subseteq
  X$.
  Now, recall that $\inir(e) = \eqclassir{i}$ with
  $i \in \diff{c'}{c}$.
  In order to show that $\inir(X) \vdash \inir(e)$, according to
  Definition~\ref{de:esfusdom}, we prove that
  $\eqclassir{\ir{i} \setminus \{ i\}} \subseteq \inir(X)$.
  Let $j \in \ir{i} \setminus \{ i \}$. Clearly
  $j \in \diff{j}{\pred{j}}$ and thus
  $\eqclassir{j} = \inir(\dint{\pred{j}}{j}_{\sim})$.
  Moreover, by
  Lemma~\ref{le:prec-irr-b}(\ref{le:prec-irr-b:1bis}) the set
  $\diff{c'}{c}$ is flat and thus, since $j \sqsubset i$
  necessarily $j \not\in \diff{c'}{c}$. Since
  $j \in \ir{c'}$ we conclude that $j \in \ir{c}$, namely
  $j \sqsubseteq c$.
  This implies that $\dint{\pred{j}}{j}_{\sim} \in s(c)$ and thus
  \begin{center}
    $
    \begin{array}{lll}
      \eqclassir{j} & = \inir(\dint{\pred{j}}{j}_{\sim}) & \\
      & \subseteq \inir(s(c)) & \\
      & \subseteq \inir(X) & \mbox{[since $s(c) \subseteq X$]}
    \end{array}
    $
  \end{center}
  We thus conclude that
  $\eqclassir{\ir{i} \setminus \{ i\}} \subseteq \inir(X)$ as desired.

  Also in this case, the argument can be easily reversed to prove the
  converse implication.
\end{proof}

The paper by Droste~\cite{Dro:ESD} considers also the case of event
structures with a general consistency relation (rather than a binary
conflict). The correspondence with our approach can be extended to this setting, as further detailed in Appendix~\ref{app:consistency}.

\subsection{Relation with Asynchronous Graphs}
\label{ss:async-graphs}
 
A characterisation  of prime {\esabbr} in terms of their transition
graph has been given in~\cite{PU:RMC}. A slightly different, yet
equivalent formalisation has been rediscovered in~\cite{Mel:hab}, in the
context of the work on the abstract theory of rewriting and concurrent
games.
Here we show that an analogous characterisation can be obtained for (connected)
event structures. For our development we refer to the formalisation in~\cite{Mel:hab}.
Given a graph $G = \langle N, U, s, t \rangle$, a sequence of edges
$w = u_1; \ldots; u_n \in U^*$ is a path whenever each edge has a
target that coincide with the source of the subsequent edge, i.e., for
all $i \in \interval{n-1}$, $t(u_i) = s(u_{i+1})$.
Let us denote by $P_2(G)$ the set of paths of length $2$, i.e.,
$P_2(G) = \{ u_1; u_2 \mid u_1, u_2 \in E \}$. Note that two paths of
length 2 with the same source and target can be seen as a ``square''
in the graph. An asynchronous graph is then a transition system where
some squares are declared to commute.

\begin{definition}[asynchronous graph]
  \label{de:async-graph}
  An \emph{asynchronous graph} 
  is a tuple
  $A = \langle G, n_0, \simeq \rangle$ where
  $G = \langle N, U, s, t \rangle$ is a directed graph,
  $n_0 \in N$ is the origin and
  ${\simeq} \subseteq P_2(G) \times P_2(G)$ is an equivalence
  relation
  on coinitial and cofinal paths of length $2$ (i.e., if
  $u_1;u_2 \simeq v_1;v_2$ then $s(u_1)=s(v_1)$ and $t(u_2)=t(v_2)$)
  such that the following axioms hold (in pictures, all squares
  depicted are assumed to commute)
  \begin{enumerate}

  \item if $u_1;u_2 \simeq v_1;v_2$ and $u_2 \neq v_2$ then $u_1 \neq v_1$;
    \label{de:async-graph:1}
    \begin{center}
      \begin{tikzpicture}[node distance=6mm, >=stealth',baseline=(current bounding box.center), x=12mm, y=6mm]
        \node at (1,2) (T) {$\bullet$};    
        \node at (0,1) (L) {$\bullet$};
        \node at (2,1) (R) {$\bullet$};
        \node at (1,0) (B) {$\bullet$};    
        \path (B) edge[->] node[trans, below] {$u_1$} (L);
        \path (B) edge[->] node[trans, below] {$v_1$} (R);
        \path (L) edge[->] node[trans, above] {$u_2$} (T);
        \path (R) edge[->] node[trans, above] {$v_2$} (T);
        \pgfBox
      \end{tikzpicture}
    \end{center}
    
  \item if $u;u_1 \simeq v_1;v_2$ and $u;u_1' \simeq v_1';v_2'$ then
    ($u_1=u_1'$ iff $v_1=v_1'$);
    \label{de:async-graph:2}
    \begin{center}
      \begin{tikzpicture}[node distance=6mm, >=stealth',baseline=(current bounding box.center), x=12mm, y=6mm]
        \node at (1,2) (T) {$\bullet$};    
        \node at (0,1) (L) {$\bullet$};
        \node at (2,1) (R) {$\bullet$};
        \node at (1,0) (B) {$\bullet$};    
        \path (B) edge[->] node[trans, below] {$u$} (L);
        \path (B) edge[->] node[trans, below] {$v_1$} (R);
        \path (L) edge[->] node[trans, above] {$u_1$} (T);
        \path (R) edge[->] node[trans, above] {$v_2$} (T);
        \pgfBox
      \end{tikzpicture}
      and 
      \begin{tikzpicture}[node distance=6mm, >=stealth',baseline=(current bounding box.center), x=12mm, y=6mm]
        \node at (1,2) (T) {$\bullet$};
        \node at (0,1) (L) {$\bullet$};
        \node at (2,1) (R) {$\bullet$};
        \node at (1,0) (B) {$\bullet$};    
        \path (B) edge[->] node[trans, below] {$u$} (L);
        \path (B) edge[->] node[trans, below] {$v_1'$} (R);
        \path (L) edge[->] node[trans, above] {$u_1'$} (T);
        \path (R) edge[->] node[trans, above] {$v_2'$} (T);
        \pgfBox
      \end{tikzpicture}
    \end{center}

  \item Cube 
    \label{de:async-graph:3}
    \begin{center}
      \begin{tikzpicture}[node distance=6mm, >=stealth',baseline=(current bounding box.center), x=12mm, y=7mm]
        \node at (1,3) (T) {$\bullet$};
        \node at (0,2) (TL) {$\bullet$};
        \node at (2,2) (TR) {$\bullet$};
        \node at (1,1) (C) {$\bullet$};    
        \node at (0,1) (L) {$\bullet$};
        \node at (2,1) (R) {$\bullet$};
        \node at (1,0) (B) {$\bullet$};    
        \path (B) edge[->] node[trans, below] {$u_1$} (L);
        \path (B) edge[->] node[trans, below] {$v_1$} (R);
        \path (C) edge[->] node[trans, above] {} (TL);
        \path (C) edge[->] node[trans, above] {} (TR);
        \path (L) edge[->] node[trans, left] {$u_2$} (TL);
        \path (R) edge[->] node[trans, right] {$v_2$} (TR);
        \path (TL) edge[->] node[trans, above] {$u_3$} (T);
        \path (TR) edge[->] node[trans, above] {$v_3$} (T);
        \path (B) edge[->] node[trans] {} (C);

        \pgfBox
      \end{tikzpicture}
      $
      \begin{array}{c}
        \stackrel{(a)}{\Rightarrow}\\
        \stackrel{(b)}{\Leftarrow}
      \end{array}
      $        
      \begin{tikzpicture}[node distance=6mm, >=stealth',baseline=(current bounding box.center), x=12mm, y=7mm]
        \node at (1,3) (T) {$\bullet$};
        \node at (0,2) (TL) {$\bullet$};
        \node at (2,2) (TR) {$\bullet$};
        \node at (1,2) (C) {$\bullet$};    
        \node at (0,1) (L) {$\bullet$};
        \node at (2,1) (R) {$\bullet$};
        \node at (1,0) (B) {$\bullet$};    
        \path (B) edge[->] node[trans, below] {$u_1$} (L);
        \path (B) edge[->] node[trans, below] {$v_1$} (R);
        \path (L) edge[->] node[trans, above] {} (C);
        \path (R) edge[->] node[trans, above] {} (C);
        \path (L) edge[->] node[trans, left] {$u_2$} (TL);
        \path (R) edge[->] node[trans, right] {$v_2$} (TR);
        \path (TL) edge[->] node[trans, above] {$u_3$} (T);
        \path (TR) edge[->] node[trans, above] {$v_3$} (T);
        \path (C) edge[->] node[trans] {} (T);

        \pgfBox
      \end{tikzpicture}
    \end{center}

  \item Coherence axiom
    \label{de:async-graph:4}
      \begin{center}
      \begin{tikzpicture}[node distance=6mm, >=stealth',baseline=(current bounding box.center), x=12mm, y=7mm]
        \node at (1,3) (NIL) {$ $};
        \node at (1,2) (CT) {$\bullet$};
        \node at (0,2) (TL) {$\bullet$};
        \node at (2,2) (TR) {$\bullet$};
        \node at (1,1) (C) {$\bullet$};    
        \node at (0,1) (L) {$\bullet$};
        \node at (2,1) (R) {$\bullet$};
        \node at (1,0) (B) {$\bullet$};    
        \path (B) edge[->] node[trans, below] {$u_1$} (L);
        \path (B) edge[->] node[trans, below] {$v_1$} (R);
        \path (C) edge[->] node[trans, above] {} (TL);
        \path (C) edge[->] node[trans, above] {} (TR);
        \path (L) edge[->] node[trans, left] {$u_2$} (TL);
        \path (R) edge[->] node[trans, right] {$v_2$} (TR);
        \path (L) edge[->] node[trans, above] {} (CT);
        \path (R) edge[->] node[trans, above] {} (CT);
        \path (B) edge[->] node[trans] {} (C);

        \pgfBox
      \end{tikzpicture}
      $\Rightarrow$
      \begin{tikzpicture}[node distance=6mm, >=stealth',baseline=(current bounding box.center), x=12mm, y=7mm]
        \node at (1,3) (T) {$\bullet$};
        \node at (0,2) (TL) {$\bullet$};
        \node at (2,2) (TR) {$\bullet$};
        \node at (1,2) (C) {$\bullet$};    
        \node at (0,1) (L) {$\bullet$};
        \node at (2,1) (R) {$\bullet$};
        \node at (1,0) (B) {$\bullet$};    
        \path (B) edge[->] node[trans, below] {$u_1$} (L);
        \path (B) edge[->] node[trans, below] {$v_1$} (R);
        \path (L) edge[->] node[trans, above] {} (C);
        \path (R) edge[->] node[trans, above] {} (C);
        \path (L) edge[->] node[trans, left] {$u_2$} (TL);
        \path (R) edge[->] node[trans, right] {$v_2$} (TR);
        \path (TL) edge[->] node[trans, above] {} (T);
        \path (TR) edge[->] node[trans, above] {} (T);
        \path (C) edge[->] node[trans] {} (T);

        \pgfBox
      \end{tikzpicture}
    \end{center}

  \end{enumerate}
\end{definition}

Given an asynchronous graph, we denote by the same symbol $\simeq$ the
extension of the equivalence to all paths by contextual closure, i.e.,
$w_1;w;w_2 \simeq w_1;w';w_2$ for all $w_1, w_2, w, w' \in U^*$ with $w \simeq w'$.
The equivalence
classes of paths from the origin can be ordered by prefix, leading to
a partial order $P(A)$.
Then it can be shown that the partial orders of finite configurations
of prime {\esabbr} exactly correspond to asynchronous graphs such
that all cofinal paths from the origin are %
equivalent.

\begin{definition}[prime asynchronous graph]
  An asynchronous graph 
  $A = \langle G, n_0, \simeq \rangle$ is called  \emph{prime} if 
  all cofinal paths from the origin $n_0$ are equivalent.
\end{definition}

It can be seen that the requirement of having all cofinal paths equivalent amounts to that of having all {coinitial and cofinal} paths of length $2$ (squares) equivalent. This is indeed how the condition is formalised in~\cite{PU:RMC}.

\begin{theorem}[asynchronous graphs/prime {\esabbr}~\cite{Mel:hab}]
  \label{th:async-pes}
  Let $A$ be a prime asynchronous graph. The ideal completion
  $\ideal{P(A)}$ is a prime domain. Conversely, each
  prime domain is isomorphic to $\ideal{P(A)}$ for some prime
  asynchronous graph $A$.
\end{theorem}

With respect to~\cite{Mel:hab},
we added the coherence
axiom~(\ref{de:async-graph:4}) in the definition of
asynchronous graph, which
is going to be pivotal in our later characterisation of weak prime domains (Proposition~\ref{pr:was-ces}). This is actually necessary already for having a correspondence with
prime domains and {\esabbr}.\footnote{In a personal communication, Paul Andr\'ee Melli\`es agreed that condition (4) is necessary for the correctness of Theorem 3 of Section 2.6 of~\cite{Mel:hab},
rephrased here as Theorem~\ref{th:async-pes}.
}

The correspondence established by Theorem~\ref{th:async-pes}
generalises to connected {\esabbr} and what we call weak
asynchronous graphs, i.e., asynchronous graphs where only the forward
part of the cube axiom holds, while the converse implication (indeed
sometimes referred to as stability axiom) may fail.

\begin{definition}[weak asynchronous graphs]
  A \emph{weak asynchronous graph} is defined as in
  Definition~\ref{de:async-graph}, but omitting the stability
  axiom~(\ref{de:async-graph:3}a). It is called \emph{weak prime} if
  additionally all cofinal paths from the origin are equivalent.
\end{definition}

Then we can prove that weak prime domains are exactly the partial
orders generated by weak prime asynchronous graphs (which in turn
correspond to connected {\esabbr}).

\begin{proposition}[weak asynchronous graphs and  domains]
  \label{pr:was-ces}
  Let $A$ be a weak prime asynchronous graph. The ideal completion
  $\ideal{P(A)}$ is a weak prime domain. Conversely, each weak
  prime domain is isomorphic to $\ideal{P(A)}$ for some weak prime
  asynchronous graph $A$.
\end{proposition}

\begin{proof}
  First observe that in a weak asynchronous graph
  $A = \langle G, n_0, \simeq \rangle$ with
  $G = \langle N, U, s, t \rangle$ such that all the {cofinal} paths from the
  origin are equivalent we have that all the squares are commuting. Thus
  axioms~(\ref{de:async-graph:1}) and ~(\ref{de:async-graph:2}) imply that the graph
  is simple, that there are at most two different paths of length $2$ with
  the same source and target, and that there is at most one way of
  closing a square.

  Now, let $D$ be a weak prime domain and consider the subset of compact
  elements $\compact{D}$. It can be seen as an (acyclic) graph by taking compact
  elements as nodes and intervals as edges, with source and target
  functions being the obvious ones $s(\dint{c}{c'}{)} = c$ and
  $t(\dint{c}{c'}) = c'$.
  Then taking $\emptyset$ as origin and letting all the squares commute,
  we get a weak asynchronous graph where all the paths are equivalent.
  In detail, as observed above, axiom~(\ref{de:async-graph:1}) follows
  from the fact that the graph is simple. 
  Axiom (\ref{de:async-graph:2}) says that there are at most two
  paths of length $2$ between the same source and target. Assume that this
  is not the case, i.e., $\compact{D}$ contains a substructure as below, with $y_1, y_2, y_3$ pairwise distinct.
  \begin{center}
    \begin{tikzpicture}[node distance=6mm, >=stealth',baseline=(current bounding box.center), x=12mm, y=8mm]
      \node at (1,2) (z) {$z$};
      \node at (0,1) (y1) {$y_1$};    
      \node at (1,1) (y2) {$y_2$};
      \node at (2,1) (y3) {$y_3$};
      \node at (1,0) (x) {$x$};    
      \path (x) edge[->] (y1);
      \path (x) edge[->] (y2);
      \path (x) edge[->] (y3);
      \path (y1) edge[->] (z);
      \path (y2) edge[->] (z);
      \path (y3) edge[->] (z);
    \end{tikzpicture}
  \end{center}
  Then we would have that $y_1$ is an irreducible which is not a weak
  prime. In fact $y_1 \sqsubseteq y_2 \sqcup y_3$, but it is not the
  case that either $y_1 \leftrightarrow y_2$ or $y_1 \leftrightarrow y_3$.

  Axiom (\ref{de:async-graph:3}a) follows from bounded completeness
  and the fact that if $x \prec y_1$ and $x \prec y_2$, with
  $y_1 \neq y_2$ then $y_1 \prec y_1 \sqcup y_2$ and
  $y_2 \prec y_1 \sqcup y_2$.

  Axiom (\ref{de:async-graph:4}) is an immediate
  consequence of coherence.

  Finally, we have to prove that all the paths from $\emptyset$ to the
  same target are equivalent. We prove more generally that all
  coinitial and cofinal paths are equivalent. First 
  notice that given
  two paths $w = y_1 \ldots y_n$ and $w' = y_1' \ldots y_m'$ with
  $y_1=y_1'$ and $y_n=y_m'$ then  $n=m=|\eqclassir{\ir{y_n}} \setminus \eqclassir{\ir{y_1}}|$, by Lemma~\ref{le:chains}. We prove
  by induction on $n=m$ that the two paths are equivalent. The base
  cases $n=1$ and $n=2$ are obvious.
  In the inductive case, consider
  $z = y_2 \sqcup y_2'$.
  \begin{center}
    \begin{tikzpicture}[node distance=6mm, >=stealth',baseline=(current bounding box.center), x=18mm, y=7mm]
      \node at (1,4) (T) {$y_n=y_n'$};
      \node at (0,2.5) (TL) {$y_3$};
      \node at (2,2.5) (TR) {$y_3'$};
      \node at (1,2) (C) {$z$};    
      \node at (0,1) (L) {$y_2$};
      \node at (2,1) (R) {$y_2'$};
      \node at (1,0) (B) {$y_1=y_1'$};    
      \path (B) edge[->]  (L);
      \path (B) edge[->]  (R);
      \path (L) edge[->]  (C);
      \path (R) edge[->]  (C);
      \path (L) edge[->]  (TL);
      \path (R) edge[->]  (TR);
      \path (TL) edge[->, dotted, bend left]  (T);
      \path (TR) edge[->, dotted, bend right]  (T);
      \path (C) edge[->, dotted] (T);
     \end{tikzpicture}
  \end{center}
  
  Then, as already observed, $y_2 \prec z$ and
  $y_2' \prec z$. Then
  \begin{equation}
    \label{eq:equiv}
    y_1 y_2 z \simeq y_1' y_2' z
  \end{equation}
  Moreover, since $z \sqsubseteq y_n=y_n'$ there is a path
  $y_2 z \ldots y_n$ of length $n-1$ in a way that we can apply the
  inductive hypothesis to prove that
  $y_2 y_3 \ldots y_n \simeq y_2 z \ldots y_n$.  Similarly, on the
  left side, we get $y_2' y_3' \ldots y_n \simeq y_2' z \ldots
  y_n'$. Therefore, together with (\ref{eq:equiv}), we conclude that
  $w= y_1 y_2 y_3 \ldots y_n \simeq y_1 y_2 z \ldots y_n \simeq y_1'
  y_2' z \ldots y_n' \simeq y_1' y_2' y_3' \ldots y_n' = w'$.

  \bigskip
  
  Conversely, let $A = \langle G, n_0, \simeq \rangle$ where
  $G = \langle N, U, s, t \rangle$ is a weak asynchronous graph such that
  all the paths from the origin are equivalent. Then, in particular, all
  the squares are commuting and, by axiom~(\ref{de:async-graph:1}), the
  graph is simple, i.e., we can think of edges as a relation on
  nodes.
  This allows us to view $A$ as a concurrent automata
  $(Q,\Sigma, T, (\parallel_q)_{q \in Q})$ in the sense
  of~\cite{Dro:CAD} as follows. Define an equivalence on edges by
  $u \equiv u'$ if there are $v, v' \in U$ such that $u v \sim v'u'$
  (namely, $u, u'$ are the opposite edges of a square). Then take
  nodes as states $Q = N$, equivalence classes of edges as labels
  $\Sigma = \quotient{U}{\equiv}$, transition relation
  $T = \{ (s(u), u, t(u)) \mid u \in U \}$ and local concurrency given
  by $\eqclass[\equiv]{u} \parallel_n \eqclass[\equiv]{v}$ when $u, v$ are such that
  $s(u)=s(v)=n$ and there are $u', v' \in E$ such that $uu' \sim v v'$.
  The fact that $\parallel_n$ is well-defined uses in an essential way
  axioms~(\ref{de:async-graph:3}a) and~(\ref{de:async-graph:4}).  Then an immediate adaptation
  of~\cite[Theorem 10]{Dro:CAD} to asynchronous graphs shows that
  $P(A)$ is a domain that satisfies axioms (F), (C), and (R) in
  Subsection~\ref{ss:intervals}. Finally, observe that axiom (V) is a
  ``global'' version of the axiom (\ref{de:async-graph:1}). The fact
  that the latter implies the former can be proved by exploiting the
  fact that each bounded subset of $P(A)$ is a semimodular lattice~\cite[Theorem 3.1]{DK:ACRS}. 
  Hence $D$ is a weak prime domain.
\end{proof}

\section{Domain and event structure semantics for graph rewriting}
\label{se:graphs}

In this section we consider graph rewriting systems where rules are
left-linear but possibly not right-linear and thus, as an effect of a
rewriting step, some items can be merged. We argue that weak prime
domains and connected {\esabbr} are the right tool for providing a
concurrent semantics to this class of rewriting systems. More
precisely, in Subsection~\ref{ss:back-graph} we review the basics of
graph rewriting and we generalise the notion of independence between
rule applications to graph rewriting with left-linear rules.  Then in
Subsections~\ref{ss:graph-dom} and \ref{ss:es-graph} we show that the
domain associated with a graph rewriting system by a generalisation of
a classical construction is a weak prime domain and vice versa that
each connected {\esabbr} and thus each weak prime domain arise as the
semantics of some graph rewriting system. Finally, in
Subsection~\ref{ss:prime-gg} we show how a prime event structure
semantics for a graph rewriting system can be recovered by imposing
suitable restriction on rule application.

\subsection{Graph rewriting and concatenable traces}
\label{ss:back-graph}

We first review the basic definitions about graph rewriting in the double-pushout approach~\cite{Ehr:TIAA}.
We recall graph grammars and then introduce a notion of trace, which provides a representation of a sequence of 
rewriting steps that abstracts from the order of independent rewrites. This requires an original generalisation of the notion of independence between rewriting steps to the case of left-linear rules.
Traces are then turned into a category 
$\tr{\mathcal{G}}$ of concatenable derivation traces~\cite{CELMR:ESSG}.

\begin{definition}
A \emph{(directed) graph} is a tuple
$G = \langle N, E, s, t \rangle$, where $N$ and $E$ are 
sets of
nodes and edges, and $s, t: E \rightarrow N$ are the source and target
functions.  The components of a graph $G$ are often denoted by
$N_G$, $E_G$, $s_G$, $t_G$.  A \emph{graph morphism}
$f: G \rightarrow H$ is a pair of functions
$\langle f_N: N_G \rightarrow N_H, f_E: E_G \rightarrow E_H\rangle$ such
that $f_N \circ s = s' \circ f_E$ and $f_N \circ t = t' \circ f_E$.
We denote by $\graph$ the category of graphs and graph morphisms

\end{definition}

An \emph{abstract graph} $[G]$ is an isomorphism class of
graphs.
We work with typed graphs, i.e., graphs which are
``labelled'' over some fixed graph. Formally, given a graph $T$, the category of
\emph{graphs typed over $T$}, as introduced in~\cite{CMR:GP}, is the
slice category $\slice{\graph}{T}$, also denoted $\tgraph{T}$.

\begin{definition} [graph grammar]
  A \emph{($T$-typed graph) rule} is a span
  $(L \stackrel{l}{\leftarrow} K \stackrel{r}{\rightarrow} R)$ in
  $\tgraph{T}$ where $l$ is mono and not epi.
  The typed graphs $L$, $K$, and $R$ are called the \emph{left-hand side},
  the \emph{interface}, and the \emph{right-hand side} of the rule, respectively. 
  A \emph{($T$-typed) graph grammar} is a tuple
  $\mathcal{G} = \langle T, G_s, P, \pi \rangle$, where $G_s$ is the
  \emph{start (typed) graph}, $P$ is a 
  set of \emph{rule names}, and
  $\pi$ maps each rule name in $P$ into a rule.
\end{definition}
Sometimes we write
$p:(L \stackrel{l}{\leftarrow} K \stackrel{r}{\rightarrow} R)$ for
denoting the rule $\pi(p)$.  When clear from the context we omit the
word ``typed'' and the typing morphisms.  Note that we consider only
\emph{consuming} grammars, i.e., grammars where for every rule
$\pi(p)$
the morphism $l$ is not epi.  
Also note that rules are, by default, \emph{left-linear}, i.e., the morphism
$l$ is mono. When also the morphism $r$ is mono, the rule is called
\emph{right-linear}.

An example of graph grammar has been discussed in the introduction
(see Fig.~\ref{fi:running-gg}). The type graph was left implicit: it
can be found in the top part of Fig.~\ref{fi:rules-grammars}. The
typing morphisms for the start graph and the rules are implicitly
represented by the labelling. Also observe that for the rules only the
left-hand side $L$ and the right-hand side $R$ were reported. The same
rules with the interface graph explicitly represented are in
Fig.~\ref{fi:rules-grammars}.

\begin{figure}

  \begin{center}
    \subcaptionbox*{$T$}{
        \begin{tikzpicture}[node distance=8mm, >=stealth']
          \node at (0,0) [node, label=below:${c,v}$] (nc) {} 
          edge [in=200, out=170, loop]  node [lab,above] {$\bar{a}$} ()
          edge [in=150, out=120, loop]  node [lab,above] {$\bar{b}$} ()
          edge [in=95, out=65, loop]  node [lab,above] {$\bar{\nu}$} ()
          edge [in=30, out=0, loop]  node [lab,above] {$\mathit{in}$} ();
          \pgfBox
        \end{tikzpicture}
      }\\
    \hspace{-4mm}
    \subcaptionbox*{$p_y$ ($y \in \{a,b\}$)}{
    \begin{tikzpicture}[node distance=2mm, baseline=(current bounding box.center)]
      \node (l) {
      \begin{tikzpicture}[node distance=8mm, >=stealth']
      \node at (0,0) [node, label=below:$c$] (nc) {} 
      edge [in=105, out=75, loop]  node [lab,above] {$\bar{y}$} ();
      \node at (.5,0) [node, label=below:$\nu$] (nu) {} 
      edge [in=105, out=75, loop]  node [lab,above] {$\bar{\nu}$} ();
      \pgfBox
      \end{tikzpicture} 
    };
    \node [right=of l] (k) {
      \begin{tikzpicture}[node distance=8mm, >=stealth']
      \node at (0,0) [node, label=below:$c$] (nc) {};
      \node at (.5,0) [node, label=below:$\nu$] (nu) {} 
      edge [in=105, out=75, loop]  node [lab,above] {$\bar{\nu}$} ();
      \pgfBox
      \end{tikzpicture} 
    };
    \node  [right=of k] (r) {
      \begin{tikzpicture}[node distance=8mm, >=stealth']
        \node at (0,0) [node, label=below:${c,\nu}$] (nc) {} 
        edge [in=105, out=75, loop]  node [lab,above] {$\bar{\nu}$} ();
        \pgfBox
      \end{tikzpicture}
    };
    \path (k) edge[->] node[trans, above] {} (l);
    \path (k) edge[->] node[trans, above] {} (r);
    \end{tikzpicture}
    }
    \ \hfill\ 
    \subcaptionbox*{$p_c$}{
    \begin{tikzpicture}[node distance=2mm, baseline=(current bounding box.center)]
      \node (l) {
      \begin{tikzpicture}[node distance=8mm, >=stealth']
        \node at (0,0) [node, label=below:${c,\nu}$] (nc) {} 
        edge [in=160, out=130, loop]  node [lab,above] {$\mathit{in}$} ()
        edge [in=105, out=75, loop]  node [lab,above] {$\bar{\nu}$} ();
        \pgfBox
      \end{tikzpicture} 
    };
    \node [right=of l] (k) {
      \begin{tikzpicture}[node distance=8mm, >=stealth']
        \node at (0,0) [node, label=below:${c,\nu}$] (nc) {} 
         edge [in=105, out=75, loop]  node [lab,above] {$\bar{\nu}$} ();
        \pgfBox
      \end{tikzpicture} 
    };
    \node  [right=of k] (r) {
      \begin{tikzpicture}[node distance=8mm, >=stealth']
        \node at (0,0) [node, label=below:${c,\nu}$] (nc) {} 
        edge [in=105, out=75, loop]  node [lab,above] {$\bar{\nu}$} ();
        \pgfBox
      \end{tikzpicture}
    };
    \path (k) edge[->] node[trans, above] {} (l);
    \path (k) edge[->] node[trans, above] {} (r);
  \end{tikzpicture}  
  }
\end{center}
\caption{The type graph  and the rules of the grammar in Fig.~\ref{fi:running-gg}.}
\label{fi:rules-grammars}
\end{figure}

\begin{figure}[t]
\[
\xymatrix@R=4mm{ 
  {L} \ar[d]_{m^L} & {K} \ar[l]_{l} \ar[r]^{r} 
  \ar[d]^{m^K} & {R} \ar[d]^{m^R}\\
  {G} & {D} \ar[l]^{l^*} \ar[r]_{r^*} & {H} }
\]
\caption{A direct derivation.}
\label{fi:deriv}
\end{figure}

\begin{definition}[direct derivation]
  \label{de:direct derivation}
  Let $G$ be a typed graph,
  let $p : (L \stackrel{l}{\leftarrow} K \stackrel{r}{\rightarrow} R)$
  be a rule, and let $m^L$ be a \emph{match}, i.e., a typed graph
  morphism $m^L: L \rightarrow G$.  A \emph{direct derivation}
  $\delta$ from $G$ to $H$ via $p$ (based on $m^L$) is a diagram as in
  Fig.~\ref{fi:deriv}, where both squares are required to be pushouts in
  $\tgraph{T}$.
  We write $\delta : G \Rrel{p/m} H$, where $m = \langle
  m^L, m^K, m^R\rangle$, or simply $\delta : G \Rrel{p} H$.
\end{definition}

Since pushouts are defined only up to isomorphism, given isomorphisms
$\kappa : G' \to G$ and $\nu : H \to H'$, also $G' \Rrel{p/m'} H$
with $m' = \langle \kappa^{-1} \circ m^L, m^K, m^R \rangle$ and
$G \Rrel{p/m''} H'$ with $m'' = \langle  m^L, m^K, \nu \circ m^R \rangle$  are direct derivations, 
denoted by $\kappa \cdot \delta$ and $\delta \cdot \nu$, respectively.
Informally, the rewriting step removes (the image
of) the left-hand side from $G$ and replaces it by (the
image of) the right-hand side $R$. The interface $K$ (the common part 
of $L$ and $R$) specifies what is preserved.
For example, the transitions in Fig.~\ref{fi:running-rewriting} are
all direct derivations.
When rules are not right-linear, some items in the (image of the) interface are merged. This happens, e.g., for  $p_a$ and $p_b$.

\begin{definition}[derivations]
  \label{de:derivations}
  Let $\mathcal{G} = \langle T, G_s, P, \pi \rangle$ be a graph
  grammar. A \emph{derivation}
  $\rho: G_0 \Rrel{}^{*}_{\mathcal{G}} G_n$ over ${\mathcal{G}}$ is a
  (possibly empty) sequence of direct derivations
  $\{G_{i-1} \Rrel{p_{i}} G_i\}_{i \in \interval{n}}$ where
  $p_i \in P$ for $i \in \interval{n}$.  The graphs $G_0$ and $G_n$
  are called the \emph{source} and the \emph{target} of $\rho$, and
  denoted by $\source{\rho}$ and $\target{\rho}$, respectively. The
  \emph{length} of $\rho$ is $|\rho| = n$.
  The derivation is called {\em
    proper} if $|\rho| > 0$.
  Given two derivations $\rho$
  and $\rho'$ such that $\target{\rho} = \source{\rho'}$, their
  sequential composition
  $\rho\,;\,\rho': \source{\rho} \Rrel{}^{*} \target{\rho'}$ is
  defined in the obvious way.
\end{definition}

When irrelevant/clear from the context, the subscript
$\mathcal{G}$ is omitted.  
If
$\rho : G \Rrel{}^{*} H$ is a proper derivation 
and
$\kappa : G' \to G$, $\nu : H \to H'$ are graph isomorphisms, then
$\kappa \cdot \rho : G' \Rrel{}^{*} H$ and
$\rho \cdot \nu : G \Rrel{}^{*} H'$ are defined as expected.

In the double pushout approach to graph rewriting, it is natural to
consider graphs and derivations up to isomorphism.
However some structure must be imposed on the start and end graphs
of an abstract derivation in order to have a meaningful notion of
sequential composition.
In fact, isomorphic graphs are, in general, related by several
isomorphisms, while in order to concatenate derivations keeping track
of the flow of causality one must specify how the items of the source
and target isomorphic graphs should be identified.
We follow~\cite{Handbook}, inspired by the theory of
Petri nets~\cite{DMM:AANCP}: we choose for each class of isomorphic
typed graphs a specific graph, named the \emph{canonical graph}, and we
decorate the source and target graphs of a derivation with
isomorphisms from the corresponding canonical graphs to such
graphs.

Let $\mathsf{C}$ denote the operation that associates with each ($T$-typed)
graph its \emph{canonical graph}, thus satisfying $\can{G} \simeq G$
and if $G \simeq G'$ then $\can{G} = \can{G'}$.

\begin{definition}[decorated derivation]
  A \emph{decorated derivation} $\psi : G_0 \Rrel{}^* G_n$ is a
  triple $\langle \alpha, \rho, \omega \rangle$, where $\rho : G_0 \Rrel{}^*
  G_n$ is a derivation and $\alpha: \can{G_0} \to G_0$, $\omega: \can{G_n} \to
  G_n$ are isomorphisms. 
  We define $\source{\psi} = \can{\source{\rho}}$, $\target{\psi} =
  \can{\target{\rho}}$ and $|\psi| =|\rho|$.
\end{definition}

\begin{definition}[sequential composition]
\label{def:seq_com_decorated}
  Let $\psi = \langle \alpha, \rho, \omega \rangle$, $\psi' = \langle \alpha',
  \rho', \omega' \rangle$ be decorated derivations such that
  $\target{\psi} = \source{\psi'}$. Their \emph{sequential
  composition} $\psi ; \psi'$ is defined, if $\psi$ and
  $\psi'$ are proper, as 
    $\langle \alpha, (\rho \cdot \omega^{-1}); (\alpha' \cdot \rho'), \omega' \rangle$.
  Otherwise, if $|\psi|=0$ then $\psi ; \psi' =
  \langle \alpha' \circ \omega^{-1} \circ \alpha, \rho', \omega' \rangle$, and similarly,
  if $|\psi'| = 0$ then $\psi ; \psi' =
  \langle \alpha, \rho, \omega \circ {\alpha'}^{-1} \circ \omega' \rangle$.
\end{definition}

We next define an abstraction equivalence that identifies derivations that differ
only in representation details.

\begin{definition}[abstraction equivalence]
  \label{de:der-equiv1}
  Let $\psi = \langle \alpha, \rho, \omega \rangle$, $\psi' = \langle \alpha',
  \rho', \omega' \rangle$ be decorated derivations with $\rho: G_0
  \Rrel{}^{*} G_n$ and $\rho': G'_0 \Rrel{}^{*} G'_{n'}$
  (whose $i^{th}$ step is depicted in the lower rows of
  Fig.~\ref{fi:der-iso}).  They are \emph{abstraction equivalent}, 
  written  $\psi \equiv^{a} \psi'$, 
  if  $n = n'$,  $p_{i} = p_{i}'$ for all $i \in \interval{n}$,
  and there exists a family of isomorphisms
  $\{\theta_{X_i}: X_i \rightarrow X'_i \mid X \in \{G, D\},$
  $i \in \interval{n} \} \cup\{\theta_{G_0}\}$
  between corresponding graphs in the two derivations such that
  (1)
    the isomorphisms relating the source and target
    commute with the decorations, i.e., $\theta_{G_0} \circ \alpha = \alpha'$
    and $\theta_{G_n} \circ \omega = \omega'$; and
  (2) 
    the resulting diagram (whose $i^{th}$ step is represented in
    Fig.~\ref{fi:der-iso}) commutes.
\end{definition}

  Equivalence classes of decorated derivations 
  with respect to $\equiv^{a}$ 
  are called \emph{abstract derivations} and 
  denoted by
  $[\psi]_{a}$, where $\psi$ is an element of the class.

\begin{figure}
\vspace{-1.1cm}

\[
\scalebox{.8}{
\xymatrix@R=1.4mm@C=2mm{ 
  & & & {p_i:} & {L_i} \ar[ddddl]|/-6mm/{\bx{m_i^L}} \ar[ddr]|{\bx{m_i'^L}} & & {K_i}
  \ar[ll]_{l_i} \ar[rr]^{r_i} \ar[ddddl]|/-6mm/{\bx{m_i^K}} \ar[ddr]|{\bx{m_i'^K}} & &
  {R_i} \ar[ddddl]|/-6mm/{\bx{m_i^R}} \ar[ddr]|{\bx{m_i'^R}}& & & &\\
  & & & & & & & & & & &\\  
  & & & {G_0'} & & {G_{i-1}'} & &
  {D_i'} \ar[ll]|{\bx{l_i'^*}} \ar[rr]|{\bx{r_i'^*}} & & {G_{i}'} & & {G_n'} &\\
  {\can{G_0}} \ar[urrr]^{\alpha'} \ar[dr]_{\alpha}  & & & & & & & & & & & &
  {\can{G_n}} \ar[ul]_{\omega'} \ar[dlll]^{\omega} \\
  & {G_0} \ar[uurr]|{\theta_{G_0}} &  & {G_{i-1}} \ar[uurr]|{\theta_{G_{i-1}}} & &
  {D_i} \ar[ll]^{l_i^*} \ar[rr]_{r_i^*} \ar[uurr]|{\theta_{D_i}} & &
  G_{i} \ar[uurr]|{\theta_{G_{i}}} & & {G_n} \ar[uurr]|{\theta_{G_n}} & & }
}
\]

\caption{Abstraction equivalence of decorated derivations.}
\label{fi:der-iso}
\end{figure}

From a concurrent perspective, also derivations that only differ for
the order in which two independent direct derivations are applied
should not be distinguished. This is classically formalised by a
notion of sequential independence between rewrites inducing the
so-called shift equivalence on derivations. When working with rules
which are only left-linear, we need to refine the notion of
independence as discussed below.

\begin{definition}[sequential independence]
  \label{de:seq-ind}
  Consider a derivation $G \Rrel{p_1/m_1} H \Rrel{p_2/m_2} M$ as in
  Fig.~\ref{fi:strongseq}. Then, its components are \emph{weakly sequentially
    independent} if there exists an \emph{independence pair} among
  them, i.e., two graph morphisms $i_1: R_1 \rightarrow D_2$ and
  $i_2: L_2 \rightarrow D_1$ such that $l_2^* \circ i_1 = m_{R_1}$,
  $r_1^* \circ i_2 = m_{L_2}$.
  They are \emph{sequentially
    independent} if the independence pair is unique.
\end{definition}

Intuitively, when the independence pair is not unique, we can think
that the first rewrite has performed some fusions that the second
rewrite is using in an essential way. Hence the steps should not
considered independent.
Note that when dealing with linear rules, the independence pair, if it exists,
is always unique. Hence the two notions independence coincide and
reduce to the classical one in~\cite{CMREHL:AAGT}.

\begin{figure}[t]
\[
\def\objectstyle{\scriptstyle}\def\labelstyle{\scriptstyle}
\xymatrix@C=0.7mm@R=7mm{
  {L_1} \ar[d]_{m_{L_1}}
        & & {K_1} \ar[ll]_{l_1} \ar[rr]^{r_1} \ar[d]|{m_{K_1}}
        & & {R_1} \ar[drrr]|{m_{R_1}}  \ar@(r,ul)@{.>}|(.75){i_1}[drrrrrrrr]
        & & & & & &
  {L_2}\ar[dlll]|{m_{L_2}} \ar@(l,ur)@{.>}|(.75){i_2}[dllllllll]
        & & {K_2} \ar[ll]_{l_2} \ar[rr]^{r_2} \ar[d]|{m_{K_2}}
        & & {R_2} \ar[d]^{m_{R_2}} \\
    {G}
     & & {D_1} \ar[ll]^{l^*_1} \ar[rrrrr]_{r^*_1}
     & & & &  & {H} & &
     & & & {D_2} \ar[lllll]^{l^*_2} \ar[rr]_{r^*_2}
     & & {M}  }
\]
\caption{Sequential independence for
  $\rho = G \Rrel{p_1/m_1} H \Rrel{p_2/m_2} M$.}
\label{fi:strongseq}
\end{figure}

\begin{proposition}[interchange operator]
  \label{pr:interchangeEG}
  Let $\rho = {G} \Rrel{p_1/m_1} {H} \Rrel{p_2/m_2} {M}$ be a
  derivation whose components are sequentially independent via
  an independence pair $\xi$.
  Then, a derivation $IC_\xi(\rho) = {G} \Rrel{p_2/m^*_2}
  {H^*} \Rrel{p_1/m^*_1} {M}$ can be uniquely constructed.
  The interchange is called \emph{proper} when it produces a
  derivation that is again sequentially independent.
\end{proposition}

We explicitly observe that the result of the interchange of two
sequentially independent rewrites is still weak sequentially
independent, but, differently from what happens for linear rules, it
could fail to be sequentially independent due to non-uniqueness of the
independence pair. This motivates the notion of proper interchange.

The interchange operator is used to introduce a notion of shift
equivalence~\cite{CMREHL:AAGT}, identifying (as for the analogous
permutation equivalence of $\lambda$-calculus) those derivations which
differ only for the scheduling of independent steps. Due to the fact
that the interchange of a sequential independence derivation is not
necessarily sequential independent some care must be put for making the
relation symmetric.

\begin{definition}[shift equivalence]
  \label{de:shift-equivalence}
  The derivations $\rho$ and $\rho'$ are \emph{shift equivalent},
  written $\rho \equiv^{sh} \rho'$, if $\rho'$ can be obtained from
  $\rho$ by repeated proper interchanges of pairs of
  sequentially independent rewrite steps.
\end{definition}
If we are
interested in the way $\rho'$ is obtained from $\rho$, we write
$\rho \equiv^{sh}_\perm \rho'$, for
$\perm : \interval{n} \to \interval{n}$ the associated permutation.
It is easy to see that, due to the requirement that interchanges are proper, the relation $\equiv^{sh}$ is indeed symmetric.

For instance, in Fig.~\ref{fi:running-rewriting} it is easy to see
that the derivation $G_s \Rrel{p_a} G_b \Rrel{p_b} G_{ab}$
consists of sequentially independent direct derivations. It is shift
equivalent to $G_s \Rrel{p_b} G_a \Rrel{p_a} G_{ab}$,
via the permutation $\perm = \{(1,2),(2,1)\}$.

Two decorated derivations are said to be shift equivalent when the
underlying derivations are, i.e.,
$\langle \alpha, \rho, \omega \rangle \equiv^{sh} \langle \alpha,
\rho', \omega \rangle$ if $\rho \equiv^{sh} \rho'$. Then the
equivalence of interest arises by joining abstraction and shift equivalence.

\begin{definition}[concatenable traces]
  We denote by $\equiv^c$ the equivalence on decorated derivations
  arising as the transitive closure of the union of the relations
  $\equiv^{a}$ and $\equiv^{sh}$.
  Equivalence classes of decorated derivations with respect to
  $\equiv^c$ are denoted as $[\psi]_c$ and are called
  \emph{concatenable (derivation) traces}.
\end{definition}
We will sometimes annotate $\equiv^c$ with the permutation relating
the equivalent permutations. Formally, $\equiv^c_\sigma$ can be
defined inductively as: if $\psi \equiv^{a} \psi'$ then
$\psi \equiv^c_{id} \psi'$ , if $\psi \equiv^{sh}_{\sigma} \psi'$ then
$\psi \equiv^c_{\sigma} \psi'$, and if
$\psi \equiv^{c}_{\sigma} \psi'$ and
$\psi' \equiv^{c}_{\sigma'} \psi''$ then
$\psi \equiv^{c}_{\sigma' \circ \sigma} \psi''$.

{ %

Several proofs concerning concatenable traces exploit a 
property of equivalence $\equiv^c$ presented in \cite[Sec. 3.5]{Handbook}, 
that we summarize and adapt here to our framework.

If $\psi$ and $\psi'$ are decorated derivations, then a \emph{consistent permutation}
between their steps relates two direct derivations if they consume
and produce the same items, up to an isomorphism that is consistent with the decorations.

\begin{definition}[consistent permutation]
\label{de:consistent_permutation}

Given a decorated derivation  $\psi = \langle \alpha, \rho, \omega \rangle:  G_0 \Rrel{}^* G_n$, 
we denote by $\col{\psi}$ the colimit of the corresponding diagram in category $\tgraph{T}$, and by $in^X_{\col{\psi}}$
the injection of $X$ into the colimit, for any graph $X$ in $\rho$.
Given two such decorated derivations $\psi$ and $\psi'$ of equal length $n$, 
a \emph{consistent permutation} $\perm$ from $\psi$ to $\psi'$ is a permutation
$\perm$ on $\interval{n}$ such that
\begin{enumerate}
\item there exists an isomorphism $\xi: \col{\psi} \to \col{\psi'}$;
\item for each $i \in \interval{n}$ the direct derivations $\delta_i$ of $\psi$ and $\delta_{\perm(i)}$ of $\psi'$
use the same rule;

\item for each $i \in \interval{n}$, let  $p : (L \stackrel{l}{\leftarrow} K \stackrel{r}{\rightarrow} R)$ be the
rule used by direct derivations $\delta_i: G_{i-1} \Rrel{p/m} G_{i}$ and
$\delta'_{\perm(i)}: G'_{\perm(i)-1} \Rrel{p/m'} G'_{\perm(i)}$; then 
\begin{itemize}
\item$\xi \circ in^{G_{i-1}}_{\col{\psi}} \circ m^L =
  in^{G_{\perm(i)-1}}_{\col{\psi'}} \circ m'^L$, and 
  \item $\xi \circ in^{G_{i}}_{\col{\psi}} \circ m^R =
  in^{G_{\perm(i)}}_{\col{\psi'}} \circ m'^R$;
\end{itemize}
  
\item {[$\alpha$-consistency]} $\xi \circ in^{G_0}_{\col{\phi}} \circ \alpha = in^{G'_0}_{\col{\phi'}} \circ \alpha'$;

\item {[$\omega$-consistency]}  $\xi \circ in^{G_n}_{\col{\phi}} \circ \omega = in^{G'_n}_{\col{\phi'}} \circ \omega'$;

\end{enumerate}
A permutation $\perm$ from $\psi$ to $\psi'$ is called \emph{left-consistent} if it satisfies 
conditions (1)-(4), but possibly not $\omega$-consistency.
It is easy to show, by induction on the length of derivations, that the isomorphism 
$\xi: \col{\psi} \to \col{\psi'}$ is uniquely determined by conditions (2)-(4), if it exists. 

\end{definition}

In the case of linear rules the existence of a consistent permutation
is a characterisation of equivalence $\equiv^c$. Here, it just
provides a necessary condition.

\begin{lemma}
\label{le:consistent_permutations}
Let $\psi$, $\psi'$ be decorated derivations. 
\begin{enumerate}
\item if $\psi \equiv^c_\sigma \psi'$ then
$|\psi| = |\psi'|$ and $\perm$ is a consistent permutation on 
$\interval{|\psi|}$ between them. We write  $\psi \equiv^{c}_\perm \psi'$ in this case.
\item If $\psi; \psi_1  \equiv^{c}_\perm \psi'; \psi'_1$ and $\psi  \equiv^{c}_{\perm_0} \psi'$, 
then $\perm_0$ is the restriction of $\perm$ to $\interval{|\psi|}$. In this case it also holds
$\psi_1  \equiv^{c}_{\perm_1} \psi'_1$, with $\perm_1(i) = \perm(i + |\psi|) -  |\psi|$.
\item
If $\psi \equiv^{c} \psi'$, then there is a unique consistent permutation $\perm$ such that 
$\psi \equiv^{c}_\perm \psi'$.
\end{enumerate}
\end{lemma}
} %

\begin{proof}{[sketch]}
\begin{enumerate}
\item This holds by  \cite[Thm. 3.5.3]{Handbook}. Just note that the proof of this direction does not use linearity of rules. 

\item Suppose by absurd that  $j$ be the smallest index in $\interval{|\psi|}$ such that 
$\perm(j) \not =  \perm_0(j)$. Let  $p : (L \stackrel{l}{\leftarrow} K \stackrel{r}{\rightarrow} R)$ be the 
rule used in $\delta_j$ and let 
$x \in L \setminus{l(K)}$ be an item consumed by it, which exists because all rules are consuming. 
By Definition~\ref{de:consistent_permutation} we deduce that both direct derivations $\delta'_{\perm(j)}$
and $\delta'_{\perm_0(j)}$ of $\psi';\psi'_1$ use the same rule $p$ (say, with matches $m'$ and $m''$), 
and that 
the items ${m'}^L(x) \in G'_{\perm(j)-1}$ and 
${m''}^L(x) \in G'_{\perm_0(j)-1}$ which are consumed by $\delta'_{\perm(j)}$ and $\delta'_{\perm_0(j)}$, 
respectively, are identified in the colimit $\col{\psi';\psi'_1}$ (actually, from $\psi  \equiv^{c}_{\perm_0} \psi'$
we know that there is a morphism $G'_{\perm_0(j)-1} \to \col{\psi'}$, but we can compose it with 
the obvious -- not necessarily injective -- morphism  $\col{\psi'} \to  \col{\psi';\psi'_1}$).
 But given the shape of the derivation diagram determined
by the left-linearity of rules, and the 
properties of colimits in $\graph$, this is not possible, because there is no undirected path of morphisms 
relating the images of element $x \in L$ in  $G'_{\perm(j)-1}$ and $G'_{\perm_0(j)-1}$ respectively.
Therefore $\perm$ and $\perm_0$ must coincide on $\interval{|\psi|}$.

For the second part, by the fact just proved clearly $\perm_1$ is a
well-defined permutation on $\interval{|\psi_1|}$. Then the fact that
$\psi_1 \equiv^{c}_{\perm_1} \psi'_1$ is almost immediate. Only
commutation of the source decorations is not obvious, but it follows
from commutation of the target for $\psi \equiv^{c}_{\perm_0} \psi'$
and Definition~\ref{def:seq_com_decorated}.

\item Direct consequence of the previous point, considering zero-length decorated derivations $\psi_1$ and $\psi_1'$.
\end{enumerate}
\end{proof}

The sequential composition of decorated derivations lifts to
composition of derivation traces so that we can consider the
corresponding category.

\begin{definition}[category of concatenable traces]
  \label{de:abs-shift-cat}
   Let $\mathcal{G}$ be a graph grammar.
   The \emph{category of concatenable traces} of 
  $\mathcal{G}$, denoted by \tr{\mathcal{G}}, has abstract graphs as
  objects and concatenable traces as arrows.
\end{definition}

\subsection{A weak prime domain for a grammar}
\label{ss:graph-dom}

For a grammar $\mathcal{G}$ we obtain a partially ordered representation of its derivations
starting from the initial graph by considering the concatenable traces ordered 
by prefix. 
Formally, as done in~\cite{Handbook,Bal:PhD} for linear
grammars, we consider the category
$\slice{[G_s]}{\tr{\mathcal{G}}}$, which, by definition of sequential
composition between traces, is easily shown to be a preorder.

\begin{proposition}
  Let $\mathcal{G}$ be a graph grammar. Then the category
  $\slice{[G_s]}{\tr{\mathcal{G}}}$ is a preorder.
\end{proposition}

\begin{proof}
  Let $[\psi] : [G_s] \to [G]$, $[\psi'] : [G_s] \to [G']$ be
  concatenable traces and let
  $[\psi_1], [\psi_2] : [\psi] \to [\psi']$ be arrows in the slice
  category. Spelled out, this means that $\psi_1, \psi_2: G \to G'$
  are such that $\psi; \psi_1 \equiv^c \psi; \psi_2 \equiv^c
  \psi'$. By point (2) of Lemma~\ref{le:consistent_permutations}, 
  using the fact that $\psi \equiv^c \psi$ we can conclude
  that $\psi_1 \equiv^c \psi_2$, as desired. 
 \end{proof}

Explicitly, elements of the preorder are concatenable traces
$[\psi]_c : [G_s] \to [G]$ and, for $[\psi']_c : [G_s] \to [G']$, we
have $[\psi]_c \sqsubseteq [\psi']_c$ if there is $\psi'' : G \to G'$
such that $\psi ; \psi'' \equiv^c \psi'$.
Note that, given two concatenable traces $[\psi]_c : [G_s] \to [G]$
and $[\psi']_c : [G_s] \to [G']$, if
$[\psi]_c \sqsubseteq [\psi']_c \sqsubseteq [\psi]_c$ then $\psi$ can be
obtained from $\psi'$  by composing it with a zero-length
trace. Hence the elements of the partial order induced by
$\slice{[G_s]}{\tr{\mathcal{G}}}$ intuitively consist of classes of
concatenable traces whose decorated derivations  are
related by an isomorphism that has to be consistent with
the decoration of the source.
Once applied to the grammar in Fig.~\ref{fi:running-gg}, this
construction produces a domain isomorphic to that in
Fig.~\ref{fi:running-configurations}.

\begin{lemma}
\label{le:partial_order_gg}
  Let $\mathcal{G}$ be a graph grammar. The partial order induced by
  $\slice{[G_s]}{\tr{\mathcal{G}}}$, denoted $\poset{\mathcal{G}}$,
  has as elements
  $\ltrace{\psi} = \{ [\psi \cdot \nu]_c \mid \nu : \target{\psi}
  \stackrel{\sim}{\to} \target{\psi} \}$ and
  $\ltrace{\psi} \sqsubseteq \ltrace{\psi'}$ if
  $\psi;\psi'' \equiv^c \psi'$ for some decorated derivation $\psi''$.
\end{lemma}

\begin{proof}
  Immediate.
\end{proof}

The domain of interest is then obtained by ideal completion of
$\poset{\mathcal{G}}$, with (the principal ideals generated by) the elements in
$\poset{\mathcal{G}}$ as compact elements.
In order to give a proof for this, we need a preliminary technical lemma that
essentially proves the existence and provides the shape of the least 
upper bounds in the domain of traces.

\begin{lemma}[properties of $\equiv^c$]
  \label{le:prop-c}
  Let $\psi$ and $\psi'$ be decorated derivations. Then the following hold:
  \begin{enumerate}
  \item
    \label{le:prop-c:1}
    Let $\psi_1, \psi_1'$ be such
    that $\psi;\psi_1 \equiv^c_{\perm} \psi';\psi_1'$ and let
    $n = | \{ j \in \interval[|\psi|]{|\psi;\psi_1|-1} \mid \perm(j)
    < |\psi'| \}|$. Then for all $\phi_2, \phi_2'$ such that
    $\psi;\phi_2 \equiv^c \psi';\phi_2'$ it holds $|\phi_2| \geq n$
    and there are $\psi_2, \psi_2',\psi_3$ such that

    \begin{itemize}
    \item $\psi;\psi_1 \equiv^c \psi;\psi_2;\psi_3$
    \item $\psi;\psi_2 \equiv^c \psi';\psi_2'$
      
    \item $|\psi_2| = n$
    
    \end{itemize}

  \item
    \label{le:prop-c:2}
    Let $\psi_1, \psi_1'$, $\psi_2, \psi_2'$ be such that
    $\psi;\psi_1 \equiv^c_{\perm_1} \psi';\psi_1'$ and
    $\psi;\psi_2 \equiv^c_{\perm_2} \psi';\psi_2'$ with $\psi_1, \psi_2$
    of minimal length. Then $\psi_1 \equiv^c_\perm \psi_2 \cdot \nu$,
    where $\nu : \target{\psi_2} \to \target{\psi_2}$ is some graph isomorphism and
    $\perm(j) = \perm_2^{-1}(\perm_1(j + |\psi|))-|\psi|$ for
    $j \in \interval[0]{|\psi_1|-1}$.    
  \end{enumerate}
\end{lemma}  

\begin{proof}
\begin{enumerate}
\item 
  We first observe that if $\psi, \psi'$ are derivation traces and
  $\psi_1, \psi_1'$ are such that
  $\psi;\psi_1 \equiv^c_{\perm} \psi';\psi_1'$, with $|\psi| = k$, $|\psi'|=k'$, $|\psi;\psi_1| = |\psi';\psi_1'|=h$ then there is
  a $\phi_1$ such that
  $\psi;\psi_1 \equiv^c \psi;\phi_1 \equiv^c_{\perm_1} \psi';\psi_1'$
  and
  \begin{center}
    for $i, j \in \interval[|\psi|]{h-1}$, $i \leq j$ implies
     $\perm_1(i) \leq \perm_1(j)$. \hfill{(\dag)}
  \end{center}
  
  In order to prove this, we can proceed by induction on the number of
  inversions
  $x = |\{ (i, j) \in \interval[|\psi|]{h - 1} \mid i \leq
  j\ \land\ \perm(i) > \perm(j) \}|$, i.e., on the number of pairs
  $(i,j)$ in the interval of interest that do not respect the
  monotonicity condition. When $x=0$ the thesis immediately
  holds. Assume that $x > 0$. Then there are certainly indices
  $j \in \interval[|\psi|]{h-2}$ such that
  $\perm(j) > \perm(j+1)$. Among these, take the index $i$ such that
  $\perm(i+1)$ is minimal. Then it can be shown that direct derivations at 
  position $i$ and
  $i+1$ in $\psi_1$ are sequentially independent,  and thus they can be switched, i.e., there
  is $\phi_2$ such that
  $\psi; \phi_2 \equiv^c_{id[i \mapsto i+1, i+1 \mapsto i]} \psi;
  \psi_1$. Therefore
  $\psi; \phi_2 \equiv^c_{\perm \circ id[i \mapsto i+1, i+1 \mapsto
    i]} \psi'; \psi_1'$. This reduces the number of inversions and thus
  the inductive hypothesis allows us to conclude.

  In the same way, we can prove that there is
  a $\phi'_1$ such that
  $\psi;\phi_1 \equiv_{\perm_2}^c \psi';\phi'_1 \equiv^c \psi';\psi_1'$
  and
  \begin{center}
    for $i, j \in \interval[|\psi'|]{h\!-\!1}$, if $i \leq j$ then
     $\perm^{-1}_2(i) \leq \perm^{-1}_2(j)$ {(\ddag)}
  \end{center}
  
  Putting conditions (\dag) and (\ddag) together we derive that
  $\psi;\psi_1 \equiv^c \psi;\phi_1\equiv^c_{\perm'} = \psi';\phi'_1
  \equiv^c \psi';\psi'_1$.  Now let $y \in\interval[|\psi|]{h-1}$ be
  the largest index such that $\perm'(y) < |\psi'|$ (or $y= |\psi|$ if
  it does not exist), let $l_3 = h - y$ and consider decorated
  derivations $\psi_2, \psi_3, \psi_2',\psi_3'$ such that
  $|\psi_3| = |\psi'_3| = l_3$ and
  $ \psi;\psi_2;\psi_3 = \psi;\phi_1 \equiv^c_{\perm'} \psi';\phi'_1 =
  \psi';\psi'_2;\psi'_3$.  By construction we obtain that
  $|\psi_2| = n$ and that $\perm'$ restricts to a permutation
  $\perm'_2$ on $\interval[0]{|\psi;\psi_2|-1}$. Commutation with the
  target decoration can be obtained, if necessary, by changing the
  $\omega$ decoration of $\psi_2$, affecting only the $\alpha$
  decoration of $\psi_3$. Thus
  we conclude that
  $\psi;\psi_2 \equiv^c \psi';\psi_2'$.
  
  Finally, notice that by the definition of $y$ and the properties of $\perm'$, 
  it follows that $\perm'(j) < |\psi'|$ for all $j \in
  \interval[|\psi|]{|\psi;\psi_2|-1}$ and $\perm'(j) \geq |\psi'|$ for all $j \in \interval[|\psi;\psi_2|]{h-1}$. 
  That is, the direct derivations in $\psi_2$ match 
  all direct derivations of $\psi'$ that are not matched in $\psi$. This implies that 
  there cannot exist a derivation $\phi_2$ shorter than $n$ such that 
  $\psi;\phi_2 \equiv^c \psi';\phi_2'$ for some $\phi_2'$.

\item 
  Let $n = |\psi|$ and $m = |\psi_1| = |\psi_2|$, which must have the same length.
  By the last part of the proof of the previous point,
  since 
  both $\psi_1$ and $\psi_2$ are of 
  minimal length, we have that for all $j \in
  \interval[n]{n+m-1}$ it holds  $\perm_1(j) < |\psi'|$ and $\perm_2(j) < |\psi'|$.
Furthermore, $\perm_1(\interval[n]{n+m-1}) = \perm_2(\interval[n]{n+m-1})$, 
because both $\psi_1$ and  $\psi_2$ consist of direct derivation that match 
those of $\psi'$ which are not matched in $\psi$.

Thus 
$\perm(j) = \perm_2^{-1}(\perm_1(j + |\psi|))-|\psi|$ is a
well-defined permutation on $\interval[0]{|\psi_1|-1}$ from $\psi_1$
to $\psi_2$. It is easy to see that the only condition that can be
violated for concluding $\psi_1 \equiv^c_\sigma \psi_2$ is commutation
of the target decorations. This can be reestablished by post-composing
$\psi_2$ with a graph isomorphism.
\end{enumerate}
\end{proof}

Relying on the results above we can easily prove that the ideal
completion of the partial order of traces is a domain.

\begin{proposition}[domain of traces] 
  \label{pr:domain-gg} 
  Let $\mathcal{G}$ be a graph grammar. Then
  $\dom{\mathcal{G}} =  \ideal{\poset{\mathcal{G}}}$
  is a
  domain.
\end{proposition}

\begin{proof}
  By Lemma~\ref{le:generators}
  it is sufficient to prove 
  (1) that
  $\principal{\ltrace{\psi}}$ is finite for every $\ltrace{\psi} \in
  \poset{\mathcal{G}}$, and (2) that if $\{\ltrace{\psi_1},
  \ltrace{\psi_2},
  \ltrace{\psi_3}\}$ is pairwise consistent then $\ltrace{\psi_1}
  \sqcup
  \ltrace{\psi_2}$ exists and is consistent with $\ltrace{\psi_3}$.

  \begin{enumerate}

  \item Let $\ltrace{\psi'} \sqsubseteq \ltrace{\psi}$. By
    Lemma~\ref{le:partial_order_gg}
    we know that
    $\psi';\psi'' \equiv^c_{\perm} \psi$ for some decorated derivation
    $\psi''$ and a
    permutation $\perm$.  Now suppose that
    $\psi'_1$ and $\psi'_2$ are decorated derivations such that
    $\psi'_1;\psi''_1 \equiv^c_{\perm_1} \psi$ and
    $\psi'_2;\psi''_2 \equiv^c_{\perm_2} \psi$ for some $\psi''_1$,
    $\psi''_2$, and that
    $\perm_1(\interval[0]{|\psi'_1|}) =
    \perm_2(\interval[0]{|\psi'_2|}) \subseteq
    \interval[0]{|\psi|}$. Then $\perm_2^{-1}\circ \perm_1$ is a
    permutation on $\interval[0]{|\psi'_1|}$ from $\psi'_1$ to
    $\psi'_2$ witnessing $\psi'_1 \equiv^c_{\perm_2^{-1}\circ \perm_1} \psi'_2; \nu$ for some isomorphism $\nu$. 
    Therefore
    $\ltrace{\psi'_1} = \ltrace{\psi'_2}$.  As a consequence, the
    cardinality of $\principal{\ltrace{\psi}}$ is bound by
    $2^{|\psi|}$.

  \item Given two consistent elements $\ltrace{\psi_1}$ and
    $\ltrace{\psi_2}$ of $\poset{\mathcal{G}}$, there exists
    $\ltrace{\psi} = \ltrace{\psi_1} \sqcup \ltrace{\psi_2}$, where
    $\psi$ is the minimal common extension of $\psi_1$ and $\psi_2$,
    provided by Lemma~\ref{le:prop-c}(\ref{le:prop-c:1}). Uniqueness
    of $\ltrace{\psi}$ follows by
    Lemma~\ref{le:prop-c}(\ref{le:prop-c:2})
    because minimal common are essentially unique (up to $\equiv^c$
    and right-composition with isomorphisms).  Suppose further that
    $\ltrace{\psi_3}$ is compatible with both $\ltrace{\psi_1}$ and
    $\ltrace{\psi_2}$: we have to show that it is compatible with
    $\ltrace{\psi}$.  Let
    $\ltrace{\psi'} = \ltrace{\psi_2} \sqcup \ltrace{\psi_3}$. Then
    there exist $\phi_1$, $\phi$ and $\phi'$ such that
    $\psi_1 ; \phi_1 \equiv^c_{\perm_1} \psi_2 ; \phi \equiv^c_{\perm}
    \psi$ and $\psi_2 ; \phi' \equiv^c_{\perm'} \psi'$ for suitable
    permutations $\perm_1$, $\perm$ and $\perm'$.
    
    We conclude by showing that either $\ltrace{\psi}$ and $\ltrace{\psi'}$ 
    are compatible, or $\ltrace{\psi_1} \sqcup  \ltrace{\psi_3}$ and $\ltrace{\psi'}$
    are compatible, both of which are equivalent and imply the thesis. We proceed by 
    induction on $k = |\psi_1|+|\psi_3|$.  If
    $|\psi_1|=0$, i.e.~$\psi_1$ is a zero-length decorated derivation, hence,
    by Lemma~\ref{le:prop-c}, also $\phi$ is so and thus
    $\ltrace{\psi} = \ltrace{\psi_2}$, and the latter is compatible
    with $\ltrace{\psi'}$. If $|\psi_3|=0$ we conclude analogously.  
    If $k>0$, let $\delta$ be the last
     derivation step in $\psi_1$, i.e.,
    $\psi_1 = \psi_1'; \delta$. If $\sigma_1(|\psi_1|-1) < |\psi_2|$,
    namely if step $\delta$ is already in $\psi_2$, then by
    Lemma~\ref{le:prop-c} we get that
    $\ltrace{\psi} = \ltrace{\psi_1'} \sqcup \ltrace{\psi_2}$. Since
    $|\psi_1'| < k$ we conclude by inductive hypothesis
    that $\psi$ and $\psi'$ are compatible. If instead,
    $\sigma_1(|\psi_1|-1) \geq |\psi_2|$ then, again by Lemma~\ref{le:prop-c},
    we can write $\psi$ as
    $\psi \equiv^c_{\sigma''} \psi_2; \phi''; \delta'$, where
    $\ltrace{\psi_2; \phi''} = \ltrace{\psi_1'} \sqcup
    \ltrace{\psi_2}$ and $\sigma''(|\psi_1|-1) = |\psi|-1$, i.e., $\delta$
    is mapped to $\delta'$. Hence, by inductive hypothesis
    $\psi_2; \phi''$ and $\psi'$ are compatible.

    Now, since $\ltrace{\psi_1}$ and  $\ltrace{\psi_3}$ are compatible (thus
    $\psi_1;\phi'_1 \equiv^c_{\perm_3} \psi_3;\phi'_3$ for suitable derivations 
    $\phi'_1, \phi'_3$ and permutation $\perm_3$), either step $\delta$ is 
    already in $\psi_3$ (thus $\perm_3(|\psi_1|-1) < |\psi_3|)$, or it isn't, and 
    $\perm_3(|\psi_1|-1) \geq |\psi_3|$. In the first case $\delta$ is related to 
    a step in $\psi'$, and it follows that $\ltrace{\psi'} \sqcup  \ltrace{\psi_2; \phi''} = 
    \ltrace{\psi'} \sqcup  \ltrace{\psi_2; \phi'';\delta'}$ and we conclude. 
    If instead $\delta$ is not a step in $\psi_3$,  we can write $\psi_3;\phi'_3$ as
     $\psi_3;\phi''_3;\delta''$, where step $\delta''$ matches step $\delta$ of $\psi_1$.
     By inductive hypothesis we have that $\psi_3;\phi''_3$ and $\psi'$ are compatible, 
     and we get $\ltrace{\psi_3;\phi''_3} \sqcup \ltrace{\psi'} = 
     \ltrace{\psi_2;\phi''} \sqcup \ltrace{\psi'}$. Since both steps $\delta'$ and $\delta''$ are 
     related by suitable
     permutations to step $\delta$ of $\psi_1$, we 
     can extend uniformly the two derivations preserving consistency,  obtaining
      $\ltrace{\psi_3;\phi''_3;\delta''} \sqcup \ltrace{\psi'} = 
     \ltrace{\psi_2;\phi'';\delta'} \sqcup \ltrace{\psi'} =  \ltrace{\psi} \sqcup \ltrace{\psi'}$, as desired.

  \end{enumerate}
\end{proof}

We can show that $\dom{\mathcal{G}}$ is a weak prime domain. The proof relies on the fact that irreducibles
are (the principal ideals of) elements of the form $\ltrace{\epsilon}$, where
$\epsilon = \psi; \delta$ is a decorated derivation such that its last
direct derivation $\delta$ cannot be shifted back, i.e., minimal
traces enabling some direct derivation. These are called
\emph{pre-events} in~\cite{Handbook,Bal:PhD}, where graph
grammars are linear and thus, consistently with Lemma~\ref{pr:irr-prime-alg}, such
elements provide the primes of the domain. Two irreducibles
$\ltrace{\epsilon}$ and $\ltrace{\epsilon'}$ are interchangeable when
they are different minimal traces for the same direct derivation. 

\begin{theorem}[weak prime domains from graph grammars]
  \label{th:fusion-domain-for-gg}
  Let $\mathcal{G}$ be a graph grammar. Then
  $\dom{\mathcal{G}}$ is a weak prime domain.
\end{theorem}

\begin{proof}
  We know by Proposition~\ref{pr:domain-gg} that $\dom{\mathcal{G}}$
  is a domain. Hence, recalling Definition~\ref{de:fusion-domain-new},
  we have to show that $\dom{\mathcal{G}}$ is weak prime
  algebraic. 

  We will exploit the characterisation in
  Lemma~\ref{le:generators}. First provide a characterisation of
  irreducibles and of the interchangeability relation among them. As
  usual, we confuse compact elements of $\dom{\mathcal{G}}$ with the
  corresponding generators in $\poset{\mathcal{G}}$.

  \medskip

  As mentioned above, irreducibles in $\dom{\mathcal{G}}$ are, in the
  terminology of~\cite{Handbook,Bal:PhD}, \emph{pre-events}, namely
  elements of the form $\ltrace{\epsilon}$, where
  $\epsilon = \psi; \delta$ is a decorated derivation such that its
  last direct derivation $\delta$ cannot be switched back. Formally,
  $\ltrace{\epsilon}$ is a pre-event if letting $n = |\epsilon|$ then
  for all $\epsilon = \psi; \delta \equiv^c_\perm \psi'$ it holds
  $\perm(n) = n$.

  In fact, assume that
  $\ltrace{\epsilon} = \ltrace{\psi_1} \sqcup \ltrace{\psi_2}$, and
  let
  $\epsilon  \equiv^c_\perm \psi_1 ; \psi_1' \equiv^c_{\perm'} \psi_2;
  \psi_2'$ for suitable $\psi_1', \psi_2'$ of minimal length. 
  Since $\epsilon$ is a pre-event, we have that if
  $n = |\psi; \delta| = |\psi_1 ; \psi_1'| = |\psi_2 ; \psi_2'|$, then
  $\perm'(n) = n$. This implies that $|\psi_1'| = 0$ (and thus
  $\ltrace{\epsilon} = \ltrace{\psi_1}$) or $|\psi_2'| = 0$ (and thus
  $\ltrace{\epsilon} = \ltrace{\psi_2}$),
  as desired.

   \medskip

   Two irreducibles $\ltrace{\epsilon}$ and $\ltrace{\epsilon'}$ are
   interchangeable iff the corresponding traces are compatible and
   whenever $\epsilon; \psi_1 \equiv^c_{\perm} \epsilon' ; \psi_1'$
   with $\psi_1, \psi_1'$ of minimal length (thus
   $\ltrace{\epsilon; \psi_1} = \ltrace{\epsilon'; \psi'_1} =
   \ltrace{\epsilon} \sqcup \ltrace{\epsilon'}$), then
   $\perm(|\epsilon|) = |\epsilon'|$.

   In fact, assume that $\ltrace{\epsilon} = \ltrace{\psi; \delta}$
   and $\ltrace{\epsilon'} = \ltrace{\psi';\delta'}$ are
   interchangeable, and
   $\epsilon; \psi_1 \equiv^c_{\perm} \epsilon' ; \psi_1'$ with
   $\psi_1, \psi_1'$ of minimal length. By the proof of
   Lemma~\ref{le:prop-c}(\ref{le:prop-c:1})
   we have that $\perm$ maps steps in $\psi_1$ to
   $\epsilon'$ and, analogously, $\perm^{-1}$ maps steps in
   $\psi_1'$ to $\epsilon$ (formally, $\perm(j) < |\epsilon'|$ for $j
   \geq |\epsilon|$ and, dually, if $\perm(j) \geq
   |\epsilon'|$ then $j <
   |\epsilon|$). By Lemma~\ref{le:eq-char}(\ref{le:eq-char:4}) we have
   that $\ltrace{\epsilon} \sqcup \ltrace{\epsilon'} = \ltrace{\psi}
   \sqcup \ltrace{\epsilon'} = \ltrace{\epsilon} \sqcup
   \ltrace{\psi'}$. Hence we can view the previous equivalence of
   decorated derivations as $\psi; (\delta ; \psi_1) \equiv^c_{\perm}
   (\psi'; \delta') ;\psi_1'$, with $\delta ; \psi_1$ and
   $\psi_1'$ of minimal length. This means that
   $\perm$ maps steps in $\delta;\psi_1$ to
   $\epsilon'$ and, with a dual argument, steps in
   $\delta';\psi_1'$ to
   $\epsilon$. Putting all this together we get that necessarily
   $\perm(|\epsilon|) = |\epsilon'|$, as desired.

   For the converse, assume that
   $\ltrace{\epsilon}$,
   $\ltrace{\epsilon'}$ are compatible, that $\ltrace{\psi} =
   \ltrace{\epsilon} \sqcup \ltrace{\epsilon'}$, and that $\psi
   \equiv^c \epsilon; \psi_1 \equiv^c_{\perm} \epsilon' ;
   \psi_1'$ where $\perm(|\epsilon|) =
   |\epsilon'|$. Then, reverting the reasoning above, we get that
   $\ltrace{\psi} \sqcup \ltrace{\epsilon'} = \ltrace{\epsilon} \sqcup
   \ltrace{\psi'}$, and thus we conclude that $\ltrace{\epsilon},
   \ltrace{\epsilon'}$ are interchangeable by
   Lemma~\ref{le:eq-char}(\ref{le:eq-char:4}).

  \medskip

  We conclude that $\dom{\mathcal{G}}$ is a weak prime domain, relying on
  Lemma~\ref{le:generators}. Let $\ltrace{\epsilon}$ with
  $\epsilon = \psi; \delta$ be an irreducible, and
  $\ltrace{\epsilon} \sqsubseteq \ltrace{\psi_1} \sqcup
  \ltrace{\psi_2}$.
  Let $\psi_1'$ and $\psi_2'$ be decorated derivations of minimal
  length such that
  $\epsilon;\psi \equiv^c_\perm \psi_1;\psi_1' \equiv^c_{\perm_1}
  \psi_2;\psi_2'$ for some $\psi$. If
  $\perm(|\epsilon|) \in \interval[0]{|\psi_1|-1}$ then consider
  $\phi_1$ such that
  $\psi_1;\psi_1' \equiv^c_{\perm'} \phi_1; \psi_1'$ and
  $\perm'(\perm(|\epsilon|))$ is minimal. Then $\ltrace{\phi_1}$ is an
  irreducible, $\ltrace{\phi_1}$ and $\ltrace{\epsilon}$ are
  interchangeable, and clearly
  $\ltrace{\phi_1} \sqsubseteq \ltrace{\psi_1}$. If instead
  $\perm(|\epsilon|) \geq |\psi_1|$ we have that
  $\perm_1(\perm(|\epsilon|)) < |\psi_2|$, and we can conclude, in the
  same way, the existence of
  $\ltrace{\phi_2} \sqsubseteq {\ltrace{\psi_2}}$ irreducible and
  interchangeable with $\ltrace{\epsilon}$.
\end{proof}

Note that when the rules are right-linear the domain and {\esabbr} semantics
specialises to the usual prime event structure semantics
(see~\cite{Handbook,Bal:PhD,Sch:RRSG}), since the construction of the
domain in the present paper is formally the same as
in~\cite{Handbook}.

\subsection{Any connected {\esabbr} is generated by some grammar}
\label{ss:es-graph}

By Theorem~\ref{th:fusion-domain-for-gg}, given a graph grammar
$\mathcal{G}$ the domain $\dom{\mathcal{G}}$ is weak prime.
We next show that also the converse holds, i.e., any connected
{\esabbr} (and thus any weak prime domain) is generated by a suitable graph
grammar.
This shows that weak prime domains and connected {\esabbr} are
precisely what is needed to capture the concurrent semantics of
non-linear graph grammars, and thus strengthen our claim that they
represent the right structure for modelling formalisms with fusions.

\smallskip

\noindent
\textbf{Construction (graph grammar for a connected {\esabbr})} 
Let $\langle E, \#, \vdash \rangle$ be a 
connected {\esabbr}. 
The grammar $\mathcal{G}_{E} = \langle T, P, \pi, G_s \rangle$ is
defined as follows. 

First, for every element $e \in E$, we define the following graphs, which are
then used as basic building blocks
\begin{itemize}
\item $I_e$ and $S_e$ as shown in Fig.~\ref{fi:any-fes}(\subref{fi:any-fes-ie}) and Fig.~\ref{fi:any-fes}(\subref{fi:any-fes-Se});

\item let $\pmin{e}$ denote the set-theoretical product of the minimal enablings of
  $e$, i.e., $\pmin{e} = \Pi \{ X \subseteq E \mid X \vdash_0 e
  \}$; for every tuple $u \in \pmin{e}$ we define the graph $L_{u,e}$ as in
  Fig.~\ref{fi:any-fes}(\subref{fi:any-fes-Lue}).
\end{itemize}
Moreover, for every pair of events $e, e' \in E$ such that $e \# e'$,
we define a graph $C_{e,e'}$ as in
Fig.~\ref{fi:any-fes}(\subref{fi:any-fes-C}).

The set of productions is $P = E$, i.e., we add a rule for every event $e \in E$, and we define such rule in a way that
\begin{itemize}

\item it deletes $I_e$ and  $C_{e,e'}$ for each $e' \in E$ such that $e \# e'$.
\item it preserves the graph $S_e \cup \bigcup_{u \in U_e} L_{u,e}$ 

\item for all $e' \in E$, for all graphs $L_{u,e'}$ such that
  $e$ occurs in $u$, it merges the corresponding nodes and that of $S_{e'}$ into one.

\end{itemize}
The graph $S_e \cup \bigcup_{u \in U_e} L_{u,e}$ arises from $S_e$ and
$L_{u,e}$, $u \in \pmin{e}$ by merging all the nodes (we use $\bigcup$
and $\biguplus$ to denote union and disjoint union, respectively, with
a meaning illustrated in Fig.~\ref{fi:any-fes}(\subref{fi:any-fes-U})
and Fig.~\ref{fi:any-fes}(\subref{fi:any-fes-Uplus}).) Hence, there is a match for the
rule $e$ only if $S_e$ and all $L_{u,e}$ for $u \in \pmin{e}$ have been merged
and this happens if and only if at least one minimal enabling of $e$
has been entirely executed.  The deletion of the graphs $C_{e,e'}$ establishes the
needed conflicts. The rule is consuming since it deletes the node of graph
$I_e$.
\full{
Formally, the rule for $e$ has as left-hand side the graph
\[
  I_e \cup 
  (\bigcup_{\substack{e' \in E\\ e\#e'}} C_{e,e'}) \cup
  (\bigcup_{e' \in E} (S_{e'} \uplus \biguplus_{\substack{\ u'\in\pmin{e'}\\e \in u'}} L_{u',e'})) 
  \cup (S_e \cup \bigcup_{u \in \pmin{e}} L_{u,e}) 
\]
while the right-hand side is
\[
  (S_e \cup \bigcup_{u \in \pmin{e}} L_{u,e}) 
  \cup
  (\bigcup_{e' \in E} (S_{e'} \cup \bigcup_{\substack{\ u'\in\pmin{e'}\\e \in u'}} L_{u',e'}))
\]
}
The rule is schematised in
Fig.~\ref{fi:any-fes}(\subref{fi:any-fes-rule}), where it is intended
that $e$ occurs in $u_j^1, \ldots, u_j^{n_j}$ for $u_j^i \in U_{e_j}$,
$i \in \interval{n_j}$, $j \in \interval{k}$. Moreover
$e_1', \ldots, e_h'$ are the events in conflict with $e$ and, finally,
$U_e = \{ u_1, \ldots, u_n \}$.

The start graph is just the disjoint union of all the basic graphs
introduced above
\[
  G_s = (\bigcup_{e\#e'} C_{e,e'}) \cup
  \bigcup_{e \in E}  (I_e \cup S_e \uplus \biguplus_{u \in U_e} L_{u,e})
\]

\full{
Then the type graph is
\[
  T = (\bigcup_{e\#e'} C_{e,e'}) \cup
  \bigcup_{e \in E}  (I_e \cup S_e \cup \bigcup_{u \in U_e} L_{u,e})
\]
}
\smallskip

Note that the interfaces of the rules are not given explicitly. They
can be deduced from the left and right-hand side, and the
labelling. The same applies to the type graph.

It is not difficult to show that the grammar
$\mathcal{G}_{E}$ generates exactly the {\esabbr} ${E}$. 

\begin{theorem}[connected ES from graph grammars]
  Let $\langle E, \#, \vdash \rangle$ be a 
  connected {\esabbr}.  Then,
  ${E}$ and $\ev{\dom{\mathcal{G}_{E}}}$ are isomorphic.
\end{theorem}

\begin{proof}
  First observe that any rule in $\mathcal{G}_{E}$ is executed at
  most once in a derivation since it consumes an item (the node of graph
  $I_e$) that is generated by no other rule. If we consider
  $\dom{\mathcal{G}_{E}}$, then the irreducibles are minimal
  $\ltrace{\epsilon}$ with $\epsilon = \psi; \delta$.
  By the shape of rule $e$, the derivation $\psi$ must contain
  the occurrences of a minimal set of rules such that the graphs
  $S_e$ and $L_{u,e}$ for $u \in U_e $ are merged along the common node.
  By construction, in order to merge all such graphs, if we denote by
  $X_\psi$ the set of rules applied in $\psi$, it must be
  $X_\psi \supseteq C$ for some $C \in \conff{E}$ such that
  $C \vdash_0 e$. Therefore by minimality we conclude that
  $X_\psi \vdash_0 e$. Relying on this observation, a routine
  induction on the $|C|$ shows that minimal enablings $C \vdash_0 e$
  in $E$ are in one to one correspondence with irreducibles
  $\ltrace{\epsilon}$ in $\dom{\mathcal{G}_{E}}$. Recalling, that, in
  turn, irreducibles in $\dom{E}$ are again minimal enablings, i.e.,
  $\esir{C}{e}$ with $C \in\conff{E}$ such that $C \vdash_0 e$ we
  obtain a bijection between irreducibles in $\dom{\mathcal{G}_{E}}$
  and $\dom{E}$.
  
  The fact that the correspondence preserves and reflects the order
  is, again, almost immediate by construction. In fact, consider two
  irreducibles $\ltrace{\epsilon}$ and $\ltrace{\epsilon'}$ in
  $\dom{\mathcal{G}_{E}}$ and the corresponding irreducibles
  $\esir{C}{e}$ and $\esir{C'}{e'}$ in $\dom{E}$. If
  $\esir{C}{e} \subseteq \esir{C'}{e'}$, take
  $X = \esir{C'}{e'} \setminus \esir{C}{e}$. Then $\epsilon$ can be
  extended with the rules corresponding to the events in $X$, thus
  showing the existence of a derivation $\psi$ such that
  $\epsilon; \psi \equiv^c \epsilon'$. In fact, if this were not
  possible, there would be an event $e'' \in X$ such that the
  corresponding rule compete for deleting some item of the start graph
  with a rule $e_1$ in $\epsilon$, hence $e_1 \in \esir{C}{e}$. By
  construction, the only possibility is that the common item is
  $C_{e'',e_1}$. But this would mean that $e'' \# e_1$. This
  contradicts the fact that $\{ e_1, e''\} \subseteq \esir{C'}{e'}$.
  The converse, i.e., the fact that if
  $\ltrace{\epsilon} \sqsubseteq \ltrace{\epsilon'}$ then
  $\esir{C}{e} \subseteq \esir{C'}{e'}$ is immediate.

  Recalling that domains are irreducible algebraic
  (Proposition~\ref{pr:domains-irr-alg}), we conclude that
  $\dom{\mathcal{G}_{E}}$ and $\dom{E}$ are isomorphic. 
  Since $E$ is connected {\esabbr}, by
  Theorem~\ref{th:es-dom-equivalence}, $E \simeq \ev{\dom{E}}$ and thus
  $\ev{\dom{\mathcal{G}_{E}}}$ and $E$ are isomorphic, as desired.
\end{proof}

\begin{figure}
  \hfill
  \mbox{}
  \subcaptionbox{$I_{e}$\label{fi:any-fes-ie}}{
    \begin{tikzpicture}[node distance=10mm, >=stealth',x=8mm]
      \node at (0,0) [node, label=below:${i_e}$] (ie) {}
      node[above] {};
      \pgfBox
    \end{tikzpicture}
  }
   \hfill
  \subcaptionbox{$S_e$\label{fi:any-fes-Se}}[10mm]
  {
    \begin{tikzpicture}[node distance=10mm, >=stealth',x=8mm]
      \node at (0,0) [node, label=below:$e$] (ne0) {} edge [in=105, out=75,loop]  node [above] {$e$} ();
      \pgfBox
    \end{tikzpicture}
  }
 \hfill
  \subcaptionbox{$L_{u,e}$\label{fi:any-fes-Lue}}[12mm]
  {
    \begin{tikzpicture}[node distance=10mm, >=stealth',x=8mm]
      \node at (1,0) [node, label=below:$e$] (ne1) {} edge [in=105, out=75,loop]  node [above] {$u$} ();
      \pgfBox
    \end{tikzpicture}
  }
  \hfill
  \subcaptionbox{$C_{e, e'}$\label{fi:any-fes-C}}{
    \begin{tikzpicture}[node distance=10mm, >=stealth',x=8mm]
      \node at (0,0) [node, label=below:${e \# e'}$] (ne0) {}
      node[above] {};
      \pgfBox
    \end{tikzpicture}
  }
  \mbox{}
  \hfill\\[4mm]
  \subcaptionbox{rule $e$\label{fi:any-fes-rule}}{
    \begin{tikzpicture}
      \node (l) {
        \begin{tikzpicture}[node distance=10mm, >=stealth',x=8mm]
          \node at (0,2) [node, lab, label=below:$i_e$] (ie) {}; 
          \node at (1,2) [node, lab, label=below:\scriptsize $e \# e_1'$] (c1) {}; 
          \node at (2,2) [node, lab, label=below:\scriptsize $e \# e_h'$] (cn) {}; 
          \draw[dotted] (c1) -- (cn); 
          \node at (3.5,2) [node, label=below:$e$] (ne0) {}
          edge [in=195, out=165, loop] node [below] (lu1)  {$e$} ()        
          edge [in=150, out=120, loop] node [above] (lu1)  {$u_1$} ()
          edge [in=60, out=30, loop]  node [above]  (lun)  {$u_n$} ();
          \draw[dotted] (lu1) -- (lun);           
          \node at (0,0) [node, label=below:{$e_1$}] (n1e0) {}  edge [in=105, out=75,loop]  node [above] {$e_1$} ();
          \node at (0.7,0) [node, label=below:{$e_1$}] (n1e1) {}  edge [in=105, out=75,loop]  node [above] {$u_1^1$} ();
          \node at (1.5,0) [node, label=below:{$e_1$}] (n1e2) {} edge [in=105, out=75,loop]  node [above] {$u_1^{n_1}$} ();
          \draw[dotted] (n1e1) -- (n1e2); 
          \node at (2.5,0) [node, label=below:{$e_k$}] (nke0) {}  edge [in=105, out=75,loop]  node [above] {$e_k$} ();
          \node at (3.2,0) [node, label=below:{$e_k$}] (nke1) {}  edge [in=105, out=75,loop]  node [above] {$u_k^1$} ();
          \node at (4.0,0) [node, label=below:{$e_k$}] (nke2) {} edge [in=105, out=75,loop]  node [above] {$u_k^{n_k}$} ();
          \draw[dotted] (nke1) -- (nke2); 
          \pgfBox
        \end{tikzpicture}
      };
      \node  [right=of l] (r) {
        \begin{tikzpicture}[node distance=10mm, >=stealth',x=8mm]
          \node at (3,2) [node, label=below:$e$] (ne0) {}
          edge [in=195, out=165, loop] node [below] (lu1)  {$e$} ()        
          edge [in=150, out=120, loop] node [above] (lu1)  {$u_1$} ()
          edge [in=60, out=30, loop]  node [above]  (lun)  {$u_n$} ();
          \draw[dotted] (lu1) -- (lun);           
          \node at (1,0) [node, label=below:$e_1$] (ne0) {} 
          edge [in=195, out=165, loop] node [below] (lu1)  {$e_1$} ()
          edge [in=150, out=120, loop] node [above] (lu1)  {$u_1^1$} ()
          edge [in=60, out=30, loop]  node [above]  (lun)  {$u_1^{n_1}$} ();
          \draw[dotted] (lu1) -- (lun); 
          \node at (3,0) [node, label=below:$e_k$] (ne0) {}
          edge [in=195, out=165, loop] node [below] (lu1)  {$e_k$} ()
          edge [in=150, out=120, loop] node [above] (lu1)  {$u_k^1$} ()
          edge [in=60, out=30, loop]  node [above]  (lun)  {$u_k^{n_k}$} ();
          \pgfBox
        \end{tikzpicture}
      };
      \path (l) edge[->] node[trans, above] {$e$} (r);
    \end{tikzpicture}
  }\\[4mm]
  \mbox{} \hfill
  \subcaptionbox{$S_e \uplus L_{u_1,e} \uplus L_{u_2,e}$\label{fi:any-fes-U}}
  [35mm]
  {
    \begin{tikzpicture}[node distance=10mm, >=stealth',x=8mm]
      \node at (0,0) [node, label=below:$e$] (ne0) {} 
      edge [in=105, out=75,loop]  node [above] {$e$} ();;
      \node at (1,0) [node, label=below:$e$] (ne1) {} 
      edge [in=105, out=75,loop]  node [above] {$u_1$} ();
      \node at (2,0) [node, label=below:$e$] (ne2) {} edge [in=105, out=75,loop]  node [above] {$u_2$} ();
      \pgfBox
    \end{tikzpicture}
  }
  \hfill
  \subcaptionbox{$S_e \cup  L_{u_1,e} \cup L_{u_2,e}$\label{fi:any-fes-Uplus}}
  [30mm]
  {
    \begin{tikzpicture}[node distance=10mm, >=stealth',x=8mm]
      \node at (0,0) [node, label=below:$e$] (ne0) {}
      edge [in=195, out=165, loop] node [below] (lu1)  {$e$} ()    
      edge [in=130, out=100, loop]  node [above] {$u_1$} ()
      edge [in=80, out=50, loop]  node [above] {$u_2$} ();
      \pgfBox
    \end{tikzpicture}
  }
  \hfill
  \mbox{}
  \caption{Some graphs illustrating the construction of $\mathcal{G}_E$.}
  \label{fi:any-fes}
\end{figure}

\begin{example}
  \label{ex:final}
  Consider the running example {\esabbr}, from
  Example~\ref{ex:event-structure}, with set of events
  $\{ a, b, c \}$, empty conflict relation and the minimal enablings by
  $\{a\} \vdash_0 c$ and $\{b\} \vdash_0 c$. The associated grammar is
  depicted in Fig.~\ref{fi:any-fes-a}.
  
  As a further example, consider an {\esabbr} $E_1$ with events
  $\{ a, b, c, d, e \}$. The conflict relation $\#$ is given by
  $e \# d$ and minimal enablings $\emptyset \vdash_0 a$,
  $\emptyset \vdash_0 b$, $\emptyset \vdash_0 c$, $\emptyset \vdash_0 e$,
  $\{a, b\} \vdash_0 d$ and $\{c\} \vdash_0 d$. The grammar is in
  Fig.~\ref{fi:any-fes-b}.
\end{example}

\begin{figure}
  \subcaptionbox*{$T$}{
    \begin{tikzpicture}[node distance=10mm, >=stealth',x=6.2mm]
      \node at (0,1) [node, label=below:${a}$] (na) {} edge [in=105, out=75,loop]  node [above] {$a$} ();
      \node at (1,1) [node, label=below:${b}$] (nb) {} edge [in=105, out=75,loop]  node [above] {$b$} ();
      \node at (2.2,1) [node, label=below:${c}$] (nc) {} edge [in=160, out=130, loop]  node [above] {$c$} ()
      edge [in=50, out=20, loop]  node [above] {$(a,b)$} ();
      \node at (0,0) [node, label=below:$i_{a}$] (ia) {};
      \node at (1,0) [node, label=below:$i_{b}$] (ib) {};
      \node at (2,0) [node, label=below:$i_{c}$] (ic) {};
      \pgfBox
    \end{tikzpicture}
  }
  \hspace{10mm}
  \subcaptionbox*{$G_s$}{
    \begin{tikzpicture}[node distance=10mm, >=stealth',x=6.2mm]
      \node at (0,1) [node, label=below:$a$] (ne0) {} edge [in=105, out=75,loop]  node [lab,above] {$a$} ();
      \node at (1,1) [node, label=below:$b$] (ne1) {} edge [in=105, out=75,loop]  node [above] {$b$} ();
      \node at (2,1) [node, label=below:$c$] (ne2) {} edge [in=105, out=75,loop]  node [above] {$c$} ();
      \node at (3,1) [node, label=below:$c$] (ne2) {} edge [in=105, out=75,loop]  node [above] {$(a,b)$} ();
      \node at (0,0) [node, label=below:$i_{a}$] (ia) {};
      \node at (1,0) [node, label=below:$i_{b}$] (ib) {};
      \node at (2,0) [node, label=below:$i_{c}$] (ic) {};
      \pgfBox
      \pgfBox
    \end{tikzpicture}  }
  \\[2mm]
  \subcaptionbox*{}{    
    \begin{tikzpicture}[node distance=6.2mm]
      \node (l) {
      \begin{tikzpicture}[node distance=10mm, >=stealth',x=6.2mm]
        \node at (0,0) [node, label=below:$i_{a}$] (ia) {};
        \node at (1,0) [node, label=below:$a$] (ne0) {} edge [in=105, out=75,loop]  node [above] {$a$} ();
        \node at (2,0) [node, label=below:$c$] (ne2) {} edge [in=105, out=75,loop]  node [above] {$c$} ();
        \node at (3,0) [node, label=below:$c$] (ne2) {} edge [in=105, out=75,loop]  node [above] {$(a,b)$} ();
        \pgfBox
      \end{tikzpicture} 
    };
    \node  [right=of l] (r) {
      \begin{tikzpicture}[node distance=10mm, >=stealth',x=6.2mm]
        \node at (0,0) [node, label=below:$a$] (ne0) {} edge [in=105, out=75,loop]  node [above] {$a$} ();
        \node at (1.5,0) [node, label=below:${c}$] (nc) {} edge [in=160, out=130, loop]  node [above] {$c$} ()
        edge [in=50, out=20, loop]  node [above] {$(a,b)$} ();
        \pgfBox
      \end{tikzpicture}
    };
    \path (l) edge[->] node[trans, above] {$a$} (r);
    \end{tikzpicture}
  }
  \\
  \subcaptionbox*{}{    
    \begin{tikzpicture}[node distance=6.2mm]
      \node (l) {
      \begin{tikzpicture}[node distance=10mm, >=stealth',x=6.2mm]
        \node at (0,0) [node, label=below:$i_{b}$] (ib) {};
        \node at (1,0) [node, label=below:$b$] (ne0) {} edge [in=105, out=75,loop]  node [above] {$b$} ();
        \node at (2,0) [node, label=below:$c$] (ne2) {} edge [in=105, out=75,loop]  node [above] {$c$} ();
        \node at (3,0) [node, label=below:$c$] (ne2) {} edge [in=105, out=75,loop]  node [above] {$(a,b)$} ();
        \pgfBox
      \end{tikzpicture} 
    };
    \node  [right=of l] (r) {
      \begin{tikzpicture}[node distance=10mm, >=stealth',x=6.2mm]
        \node at (0,0) [node, label=below:$b$] (ne0) {} edge [in=105, out=75,loop]  node [above] {$b$} ();
        \node at (1.5,0) [node, label=below:${c}$] (nc) {} edge [in=160, out=130, loop]  node [above] {$c$} ()
        edge [in=50, out=20, loop]  node [above] {$(a,b)$} ();
        \pgfBox
      \end{tikzpicture}
    };
    \path (l) edge[->] node[trans, above] {$b$} (r);
    \end{tikzpicture}
  }
  \hfill
  \subcaptionbox*{}{    
    \begin{tikzpicture}[node distance=6.2mm]
      \node (l) {
      \begin{tikzpicture}[node distance=10mm, >=stealth',x=6.2mm]
        \node at (0,0) [node, label=below:$i_{c}$] (ic) {};
        \node at (1,0) [node, label=below:${c}$] (nc) {} 
        edge [in=160, out=130, loop]  node [above] {$c$} ()
        edge [in=50, out=20, loop]  node [above] {$(a,b)$} ();
        \pgfBox
      \end{tikzpicture} 
    };
    \node  [right=of l] (r) {
      \begin{tikzpicture}[node distance=10mm, >=stealth',x=6.2mm]
        \node at (0,0) [node, label=below:${c}$] (nc) {} 
        edge [in=160, out=130, loop]  node [above] {$c$} ()
        edge [in=50, out=20, loop]  node [above] {$(a,b)$} ();
        \pgfBox
      \end{tikzpicture}
    };
   \path (l) edge[->] node[trans, above] {$c$} (r);
    \end{tikzpicture}
  }
  \caption{The grammar associated to our running example.}
  \label{fi:any-fes-a}
\end{figure}

\begin{figure}
  \subcaptionbox*{$T$}{
    \begin{tikzpicture}[node distance=10mm, >=stealth']
      \node at (0,0) [node, label=below:${a}$] (na) {} edge [in=105, out=75,loop]  node [above] {$a$} ();
      \node at (1,0) [node, label=below:${b}$] (nb) {} edge [in=105, out=75,loop]  node [above] {$b$} ();
      \node at (2,0) [node, label=below:${c}$] (nc) {} edge [in=105, out=75,loop]  node [above] {$c$} ();
      \node at (3,0) [node, label=below:${d}$] (nc) {} edge [in=180, out=150, loop]  node [above] {$d$} ()
      edge [in=70, out=110, loop]  node [above] {$(a,c)$} ()
      edge [in=30, out=0, loop]  node [above right=0mm and -3.5mm] {$(b,c)$} ();
      \node at (4.2,0) [node, label=below:${e}$] (ne) {} edge [in=105, out=75,loop]  node [above] {$e$} ();
      \node at (5,0) [node, label=below:$d\#e$] (de) {};
      \pgfBox
    \end{tikzpicture}
  }
  \\[2mm]
  \subcaptionbox*{$G_s$}{
    \begin{tikzpicture}[node distance=10mm, >=stealth',x=6.2mm]
      \node at (0,1) [node, label=below:${a}$] (na) {} edge [in=105, out=75,loop]  node [above] {$a$} ();
      \node at (1,1) [node, label=below:${b}$] (nb) {} edge [in=105, out=75,loop]  node [above] {$b$} ();
      \node at (2,1) [node, label=below:${c}$] (nc) {} edge [in=105, out=75,loop]  node [above] {$c$} ();
      \node at (3,1) [node, label=below:${d}$] (nd) {} edge [in=105, out=75,loop]  node [above] {$d$} ();
      \node at (3.9,1) [node, label=below:${d}$] (nd1) {} edge [in=105, out=75,loop]  node [above] {$(a,c)$} ();
      \node at (5.1,1) [node, label=below:${d}$] (nd2) {} edge [in=105, out=75,loop]  node [above] {$(b,c)$} ();
      \node at (6,1) [node, label=below:${e}$] (ne) {} edge [in=105, out=75,loop]  node [above] {$e$} ();
      \node at (7,1) [node, label=below:$d\#e$] (de) {};
      \node at (0,0) [node, label=below:$i_{a}$] (ia) {};
      \node at (1,0) [node, label=below:$i_{b}$] (ib) {};
      \node at (2,0) [node, label=below:$i_{c}$] (ic) {};
      \node at (3,0) [node, label=below:$i_{d}$] (ic) {};
      \node at (4,0) [node, label=below:$i_{e}$] (ic) {};
      \pgfBox
    \end{tikzpicture}
  }
  \\[2mm]
  \subcaptionbox*{}{    
    \begin{tikzpicture}[node distance=6.2mm]
      \node (l) {
      \begin{tikzpicture}[node distance=10mm, >=stealth',x=6.2mm]
        \node at (0,0) [node, label=below:$i_{a}$] (ia) {};
        \node at (1,0) [node, label=below:$a$] (ne0) {} edge [in=105, out=75,loop]  node [above] {$a$} ();
        \node at (2,0) [node, label=below:$d$] (ne2) {} edge [in=105, out=75,loop]  node [above] {$d$} ();
        \node at (3,0) [node, label=below:$d$] (ne2) {} edge [in=105, out=75,loop]  node [above] {$(a,c)$} ();
        \pgfBox
      \end{tikzpicture} 
    };
    \node  [right=of l] (r) {
      \begin{tikzpicture}[node distance=10mm, >=stealth',x=6.2mm]
        \node at (0,0) [node, label=below:$a$] (ne0) {} 
        edge [in=105, out=75,loop]  node [above] {$a$} ();
        \node at (1.5,0) [node, label=below:${d}$] (nc) {} 
        edge [in=160, out=130, loop]  node [above] {$d$} ()
        edge [in=50, out=20, loop]  node [above] {$(a,c)$} ();
        \pgfBox
      \end{tikzpicture}
    };
   \path (l) edge[->] node[trans, above] {$a$} (r);
    \end{tikzpicture}
  }
  \hfill
  \subcaptionbox*{}{    
    \begin{tikzpicture}[node distance=6.2mm]
      \node (l) {
      \begin{tikzpicture}[node distance=10mm, >=stealth',x=6.2mm]
        \node at (0,0) [node, label=below:$i_{b}$] (ib) {};
        \node at (1,0) [node, label=below:$b$] (ne0) {} edge [in=105, out=75,loop]  node [above] {$b$} ();
        \node at (2,0) [node, label=below:$d$] (ne2) {} edge [in=105, out=75,loop]  node [above] {$d$} ();
        \node at (3,0) [node, label=below:$d$] (ne2) {} edge [in=105, out=75,loop]  node [above] {$(b,c)$} ();
        \pgfBox
      \end{tikzpicture} 
    };
    \node  [right=of l] (r) {
      \begin{tikzpicture}[node distance=10mm, >=stealth',x=6.2mm]
        \node at (0,0) [node, label=below:$b$] (ne0) {} 
        edge [in=105, out=75,loop]  node [above] {$b$} ();
        \node at (1.5,0) [node, label=below:${d}$] (nc) {}
        edge [in=160, out=130, loop]  node [above] {$d$} ()
        edge [in=50, out=20, loop]  node [above] {$(b,c)$} ();
        \pgfBox
      \end{tikzpicture}
    };
   \path (l) edge[->] node[trans, above] {$b$} (r);
    \end{tikzpicture}
  }
  \hfill
  \subcaptionbox*{}{    
    \begin{tikzpicture}[node distance=6.2mm]
      \node (l) {
      \begin{tikzpicture}[node distance=10mm, >=stealth', x=6mm]
        \node at (0,0) [node, label=below:$i_{c}$] (ic) {};
        \node at (1,0) [node, label=below:$c$] (ne0) {} edge [in=105, out=75,loop]  node [above] {$c$} ();
        \node at (2,0) [node, label=below:$d$] (ne2) {} edge [in=105, out=75,loop]  node [above] {$d$} ();
        \node at (3,0) [node, label=below:$d$] (ne2) {} edge [in=105, out=75,loop]  node [above] {$(a,c)$} ();       
        \node at (4.5,0) [node, label=below:$d$] (ne2) {} edge [in=105, out=75,loop]  node [above] {$(b,c)$} ();
        \pgfBox
      \end{tikzpicture} 
    };
    \node  [right=of l] (r) {
      \begin{tikzpicture}[node distance=10mm, >=stealth']
        \node at (0,0) [node, label=below:$c$] (ne0) {} edge [in=105, out=75,loop]  node [above] {$c$} ();
        \node at (1,0) [node, label=below:${d}$] (nc) {} edge [in=180, out=150, loop]  node [above] {$d$} ()
        edge [in=70, out=110, loop]  node [above] {$(a,c)$} ()
        edge [in=30, out=0, loop]  node [above right=0mm and -3.5mm] {$(b,c)$} ();
        \pgfBox
      \end{tikzpicture}
    };
   \path (l) edge[->] node[trans, above] {$c$} (r);
    \end{tikzpicture}
  }
  \hfill
  \subcaptionbox*{}{    
    \begin{tikzpicture}[node distance=6.2mm]
      \node (l) {
      \begin{tikzpicture}[node distance=10mm, >=stealth']
        \node at (0,0) [node, label=below:$i_{d}$] (id) {};
        \node at (1,0) [node, label=below:${d}$] (nc) {} edge [in=180, out=150, loop]  node [above] {$d$} ()
        edge [in=70, out=110, loop]  node [above] {$(a,c)$} ()
        edge [in=30, out=0, loop]  node [above right=0mm and -3.5mm] {$(b,c)$} ();
        \node at (2.1,0) [node, label=below:$d\#e$] (de) {};
        \pgfBox
      \end{tikzpicture}       
      }; 
      \node  [right=of l] (r) {
      \begin{tikzpicture}[node distance=10mm, >=stealth']
        \node at (0,0) [node, label=below:${d}$] (nc) {} edge [in=180, out=150, loop]  node [above] {$d$} ()
        edge [in=70, out=110, loop]  node [above] {$(a,c)$} ()
        edge [in=30, out=0, loop]  node [above right=0mm and -3.5mm] {$(b,c)$} ();
        \pgfBox
      \end{tikzpicture}
    };
   \path (l) edge[->] node[trans, above] {$d$} (r);
    \end{tikzpicture}
  }
  \subcaptionbox*{}{    
    \begin{tikzpicture}[node distance=6.2mm]
      \node (l) {
      \begin{tikzpicture}[node distance=10mm, >=stealth', x=6mm]
        \node at (0,0) [node, label=below:$i_{e}$] (ie) {};
        \node at (1,0) [node, label=below:${e}$] (e) {} 
        edge [in=75, out=105, loop] node [above] {$e$} ();
        \node at (2.1,0) [node, label=below:$d\#e$] (de) {};
        \pgfBox
      \end{tikzpicture}       
      }; 
      \node  [right=of l] (r) {
      \begin{tikzpicture}[node distance=10mm, >=stealth']
        \node at (0,0) [node, label=below:${e}$] (e) {} 
        edge [in=75, out=105, loop] node [above] {$e$} ();
        \pgfBox
      \end{tikzpicture}
    };
   \path (l) edge[->] node[trans, above] {$e$} (r);
    \end{tikzpicture}
  }
  \caption{The grammar for the {\esabbr} in example~\ref{ex:final}.}
  \label{fi:any-fes-b}
\end{figure}

\subsection{A prime {\esabbr} semantics for grammars with fusions}
\label{ss:prime-gg}

A possibility for recovering a notion of causality based on prime
{\esabbr} also for graph grammars with fusions is to
introduce suitable restrictions on the concurrent applicability of
rules.
Indeed, the lack of stability arises essentially from
considering as concurrent those fusions which act on common items.
Preventing fusions to act on already merged items, one may lose some
concurrency, yet gaining a definite notion of causality.
Technically, a prime {\esabbr} can be obtained for left-linear
rewriting systems by restricting the applicability condition: the
match must be such that the pair $\langle l; m^L, r \rangle$ of
Fig.~\ref{fi:deriv} is jointly mono. This essentially means that items
which have been already fused, should not be fused again.

Formally, this means changing the applicability condition, restricting
to fusion safe derivations.
 
\begin{figure}[t]
\[
\xymatrix@R=6mm{ 
  {L} \ar[d]_{m^L} & {K} \ar[l]_{l} \ar[r]^{r} 
  \ar[d]^{m^K} & {R} \ar[d]^{m^R}\\
  {G} & {D} \ar[l]^{l^*} \ar[r]_{r^*} & {H} }
\]
\caption{A direct derivation.}
\label{fi:deriv-jointly}
\end{figure}

\begin{definition}[fusion safe (direct) derivation]
  A \emph{fusion safe} direct derivation is a direct derivation as in
  Fig.~\ref{fi:deriv-jointly} where $\langle l; m^L, r \rangle$ is
  jointly mono. A derivation is fusion safe if it consists of a
  sequence of fusion safe direct derivations.
\end{definition}

Consider our running example in Fig.~\ref{fi:running}.
Clearly, the derivations labelled $p_a$ and $p_b$ starting from 
$G_s$ are now in conflict, since e.g. the application of $p_a$ 
forbids the application of $p_b$ to $G_a$, since the 
derivation would not be anymore jointly mono.
We thus end up in the situation presented by the configurations
in Fig.~\ref{fi:non-fusion}, 
hence the applications of $p_c$ to $G_a$ and $G_b$ respectively 
must be considered as different events.

The notion of sequential independence remains unchanged. Note that the interchange operator (see Proposition~\ref{pr:interchangeEG}) 
applied to sequential independent derivations that are fusion safe produces a new pair of fusion safe 
derivations.
Then we can consider concatenable fusion safe traces, that form a subcategory of the category of traces.

\begin{definition}[fusion safe traces]
  \label{de:abs-shift-cat-2}
  Let $\mathcal{G}$ be a graph grammar.  The \emph{category of
    concatenable fusion safe traces} of $\mathcal{G}$, denoted by
  \trs{\mathcal{G}}, has abstract graphs as objects and concatenable
  fusion safe traces as arrows.
\end{definition}

The construction of Theorem~\ref{th:fusion-domain-for-gg} recasted on
fusion safe traces now produces a prime domain (hence a
prime {\esabbr}).

\begin{theorem}[prime domain structure for graph grammars]
  Let $\mathcal{G}$ be a graph grammar. Then
  $\ideal{\slice{[G_s]}{\trs{\mathcal{G}}}}$ is a prime domain.
\end{theorem}

\begin{proof}
  The proof is the same as for
  Theorem~\ref{th:fusion-domain-for-gg}. We already know that the
  domain is weak prime, hence, by Proposition~\ref{pr:fusion-domains},
  all irreducibles are weak primes.  Additionally, interchangeability,
  as characterised in the proof of the mentioned theorem, is the
  identity.

  In fact, given two irreducibles $\ltrace{\epsilon_i}$ with
  $\epsilon_i = \psi_i; \delta_i$ for $i \in \{1,2\}$ such that
  $\ltrace{\epsilon_1} \leftrightarrow \ltrace{\epsilon_2}$, by
  interchangeability
  $\ltrace{\psi_1} \sqcup \ltrace{\epsilon_2} =
  \ltrace{\epsilon_1} \sqcup \ltrace{\psi_2}$. Let such join be
  $\ltrace{\psi_1; \delta_1; \psi_1'} = \ltrace{\psi_2; \delta_2;
    \psi_2'}$ for suitable $\psi_1', \psi_2'$. This means that
  $\psi_1; \delta_1; \psi_1' \equiv^c_\sigma \psi_2; \delta_2;
  \psi_2'$ for a suitable permutation $\sigma$, with
  $\sigma(|\epsilon_1|) = \sigma(|\epsilon_2|)$. There are two
  possibilities. If $|\psi_1| = |\psi_2|$ and $\sigma$ restricts to a
  permutation of $\interval{|\psi_1|}$, then $\psi_1 \equiv^c \psi_2$
  and we conclude. Otherwise a step in $\psi_2$ is not mapped
  to $\psi_1$ or viceversa. Assume, without loss of generality, that
  there is $i \in \interval{|\psi_1|}$ such that
  $\sigma(i) > |\psi_2|$. This means that the $i$-th step in $\psi_1$
  is performed in $\psi_2'$. Since such step is performed after
  $\delta_2$, it cannot generate items consumed by $\delta_1$. Hence
  it must merge items that are merged by a different step in
  $\psi_2$. But this contradicts its fusion safety.

  \smallskip
  
  Hence all weak primes are primes and we
  conclude.
\end{proof}

\section{Conclusions and Related Work}
\label{se:conc}

In the paper we provided a characterisation of a class of domains,
referred to as weak prime algebraic domains, which is appropriate for
describing the concurrent semantics of those formalisms where a
computational step can merge parts of the state. We show a
categorical equivalence between weak prime algebraic domains and a
suitably defined class of connected event structures. We also prove
that the category of general event structures coreflects into
the category of weak prime algebraic domains.

The appropriateness of
the class of weak prime domains is witnessed by the results 
that show that weak prime algebraic domains
are precisely those arising from left-linear graph rewriting systems,
i.e., those systems where rules besides generating and deleting can
also merge graph items.
Furthermore, we show how to recover 
a prime event structures semantics also for rule-based
formalisms with fusions by introducing suitable restrictions on the
concurrent applicability of rules.

We have shown that the  
characterisations of prime domains and event structures in terms of intervals 
and asynchronous graphs naturally extend to weak prime domains. 
The characterisation of weak prime domains in terms of the interchangeability 
equivalence on irreducibles naturally suggest a presentation in terms of prime 
event structures endowed with an equivalence relation, allowing us to establish 
a link with the work in~\cite{win2017,VismeW19}.

Technically, the starting point for our proposal is the relaxation of the stability
condition for event structures. As already noted by Winskel
in~\cite{Win:ESSCCS} ``[t]he stability axiom would go if one wished to
model processes which had an event which could be caused in several
compatible ways [\ldots]; then I expect complete irreducibles would
play a similar role to complete primes here''.  Indeed, the
correspondence between irreducibles and weak primes, which
exploits the
notion of interchangeability, is the ingenious step that allows us to
obtain a smooth extension of the classical duality between
prime event structures and prime algebraic domains.

The coreflection between the category of general unstable event structures 
(with binary conflict)
and the one of weak prime algebraic domains says that the latter are
exactly the partial orders of configurations of the former. 
Such class of domains has been
studied originally in~\cite{Winskel:phd} where, generalising the work
on concrete domains and sequentiality~\cite{KP:CD}, a characterisation
is given in terms of a set of axioms expressing properties of prime
intervals.
In our paper we also provide an in depth comparison with these  results,
based on the observation that, roughly speaking, weak primes correspond to executions 
of events with their minimal enablings, while intervals can be seen as executions of 
events in a generic configuration.
A comparison is also drawn with the more recent notions of asynchronous graph~\cite{Mel:hab}, 
an alternative representation of prime algebraic domains based on the notion of path equivalence,
which we generalise in order to account for weak prime ones.

The need of resorting to unstable {\esabbr} for modelling the
concurrent computations of name passing process calculi has been
observed by several authors. In particular, in~\cite{CVY:ESSPE} an
{\esabbr} semantics for the $\pi$-calculus is defined by relying on
 {\esabbr} \emph{with names}, %
namely labelled {\esabbr} tailored for modelling
parallel extrusions. An event can have various minimal enablings but
with the constraint that distinct minimal enablings can differ only for one
event (intuitively, the extruder).
{\esabbr} with names are not connected {\esabbr} since they can have 
non-connected minimal enablings (roughly, because identical events in 
disconnected minimal enablings can be identified via the labelling).
Nevertheless, a connected {\esabbr} semantics could be obtained by
transforming {\esabbr} with names through the coreflection in the
paper: More details are reported 
in Appendix~\ref{app:pi}.

We believe that our results cover a long road in establishing weak prime domains and 
connected event structures as a foundational concept in the event-based semantics for 
concurrent computational systems.
Our next step will be to look at possible more general formalisms. 
In particular, the paper~\cite{GP:CSESPN} studies a characterisation 
of the partial order of configurations for a variety of classes of event
structures in terms of axiomatisability of the associated
propositional theories. Even if the focus in the present paper is 
on event structures that generalise Winskel's ones, we believe that
our work can provide interesting suggestions for
further development.

\smallskip

\paragraph*{Acknowledgements}
We are grateful to the anonymous referees of the conference version of
the paper for their insightful comments and suggestions. We are also indebted to Paul-Andr\`e Melli\'es for insightful discussions on the relation between event structures and asynchronous graphs.

\bibliography{Unstable}

\bibliographystyle{IEEEtran}

\appendix

\subsection{Event Structures with Non-Binary Conflict}
\label{app:consistency}

In the literature also {\esabbr} with non-binary conflict have been considered, 
where the binary conflict relation is replaced by a consistency predicate~\cite{Dro:ESD}. 
It is noteworthy that the duality results of Section~\ref{se:fes} easily adapt to this case.

\begin{definition}[{\esabbr} with non-binary conflict]
  An {\esabbr} \emph{with non-binary conflict} ({\esnabbr} for
  short) is a tuple $\langle E, \vdash, Con \rangle$ such
  that
  \begin{itemize}
  \item $E$ is a set of events
  \item $Con \subseteq \mathbf{2}^E_{fin}$ is the consistency
    predicate, satisfying $X \in Con$ and $Y \subseteq X$ implies $Y
    \in Con$
  \item $\vdash\ \subseteq Con \times E$ is the \emph{enabling}
    relation, satisfying $X \vdash e$ and $X \subseteq Y \in Con$
    implies $Y \vdash e$.
  \end{itemize}
  The {\esnabbr} ${E}$ is \emph{stable} if $X\vdash e$,
  $Y \vdash e$, and $X \cup Y \cup \{e\} \in Con$ imply
  $X \cap Y \vdash e$.
\end{definition}

A configuration $C \subseteq E$ is just a set such that
it is secured and all its finite subsets are consistent.
The notion of live {\esnabbr} is easily adapted to take into account
non-binary conflicts and also in this case we will implicitly assume
all {\esnabbr} to be live.

\begin{definition}[live {\esnabbr}]
  An {\esnabbr} ${E}$ is \emph{live} if for all
  $X \in Con$ there is $C \in \conf{E}$ such that
  $X \subseteq C$ and moreover
  for all $e \in E$ we have $\{ e \} \in  Con$.
\end{definition}

The notion of the category of {\esnabbr} is adapted accordingly.

\begin{definition}[category of event structures]
  \label{a-de:es-morphism}
  A morphism of {\esnabbr} $f : {E}_1 \to {E}_2$ is a partial
  function $f : E_1 \to E_2$ such that
  for all $X_1 \subseteq E_1$ and $e_1, e_1' \in E_1$ with $f(e_1)$, $f(e_1')$ defined
  \begin{itemize}
  \item if $X_1 \in Con_1$ then $f(X_1) \in Con_2$;
  \item if $\{ e_1, e_1' \} \in Con_1$ and $f(e_1) = f(e_1')$ then $e_1 = e_1'$;
  \item if $X_1 \vdash_1 e_1$ then 
  $f(X_1) \vdash_2 f(e_1)$.
  \end{itemize}
  We denote by $\esn$ the category of {\esnabbr} and {\esnabbr} morphisms,
  and by $\cesn$ its full subcategory having connected {\esnabbr} as objects (the definition of conectedness remains unchanged).
\end{definition}

In the definition of domains (Definition~\ref{c-de:domain}), the
existence of joins is now required only for consistent subsets,
instead of being required for pairwise consistent.

\begin{definition}[b-domains]
  \label{de:domain}
  A \emph{bounded complete domain (b-domain)} is an algebraic
  finitary partial order where all consistent subsets $X \subseteq D$
  admit a least upper bound $\bigsqcup X$.
  B-domain morphisms are as in Definition~\ref{de:domain-category}. 
  We denote by $\Domb$ the corresponding category.
\end{definition}

The definition of weak prime algebraicity remains formally the same, 
but the underlying partial order is required to be a b-domain.

\begin{definition}[weak prime algebraic b-domain]
  \label{a-de:fusion-domain-new}
  A \emph{weak prime algebraic b-domain} (or simply \emph{weak prime
    b-domain}) is a {\wi} b-domain $D$ such that for all $d \in D$ it
  holds $d = \bigsqcup (\principal{d} \cap \wpr{D})$. We denote by
  $\WDomb$ the corresponding category.
\end{definition}

The proof of the fact that, given an {\esnabbr} ${E}$, the partial order of
configurations $\dom{{E}} = \langle \conf{{E}}, \subseteq \rangle$ is
a weak prime b-domain, is unchanged. The same holds for the fact that
if $f : E_1 \to E_2$ is an {\esnabbr} morphism then
$\dom{f} : \dom{{E}_1} \to \dom{{E}_2}$ is a weak prime b-domain
morphism.

Vice versa the {\esnabbr} associated with a
weak prime b-domain is defined as follows.

\begin{definition}[{\esnabbr} for a weak prime b-domain]
  \label{a-de:es-for-dom}
  Let $D$ be a weak prime b-domain. The {\esnabbr} 
  $\ev{D} = \langle E, Con, \vdash \rangle$ is defined as
  follows
  \begin{itemize}

  \item $E = \eqclassir{\ir{D}}$;

  \item
    $Con = \{ X \mid \exists d \in \compact{D}.\ X \subseteq
    \eqclassir{\ir{d}} \}$;

  \item $X \vdash e$ if there exists $i \in e$ such that
    $\eqclassir{\ir{i} \setminus \{ i \}} \subseteq X$.

  \end{itemize}

Given a morphism $f : D_1 \to D_2$, its image
$\ev{f} : \ev{D_1} \to \ev{D_2}$ is defined for
$\eqclassir{i_1} \in E$ as
$\ev{f}(\eqclassir{i_1}) = \eqclassir{i_2}$, where
 $i_2 \in \diff{f(i_1)}{f(\pred{i_1})}$,
and $\ev{f}(\eqclassir{i_1})$ is undefined if
$f(\pred{i_1}) = f(i_1)$.
\end{definition}

We then get a result corresponding to Theorem~\ref{th:duality} for {\esabbr}
with non-binary conflict and  weak prime b-domains.

\begin{theorem}[corecflection of $\esn$ and $\WDomb$]
 \label{th:es-dom-equivalence-non-binary}
 The functors $\zdom : \esn \to \WDomb$ and $\zev : \WDomb \to \esn$
 form a coreflection.  It restricts to an equivalence between $\WDomb$
 and $\cesn$.
\end{theorem}

Concerning the interval-based characterisation in
Section~\ref{ss:intervals}, we recall that the paper by
Droste~\cite{Dro:ESD} considers also the case of event structures with
a general consistency relation (rather than a binary conflict) and
shows that the corresponding domains can be characterised as
algebraic complete partial orders where
axioms (F), (C) of Section~\ref{ss:intervals} and, additionally, (I) below hold.

\begin{description}
\item[(I)] for all $x,x',y,y' \in \compact{D}$ if $\dint{x}{x'} \sim \dint{y}{y'}$
  and $x \sqsubseteq x'$ then $y \sqsubseteq y'$.
\end{description}

The definition of the {\esabbr} $\evwd{D}$ associated with a domain $D$
(Definition~\ref{de:evwd}) can be easily adapted to the non-binary
case. The only thing that changes is the definition of
consistency: a set $X \subseteq E$ is consistent if for all $e \in X$
there exists a representative $\dint{c_e}{c_e'} \in e$ such that
$\{ c_e \mid e \in X \}$ is bounded in $D$.
Then the correspondence with our approach established in
Section~\ref{ss:intervals} easily extends to this setting: algebraic
complete partial orders where axioms (F), (C) and (I) hold are exactly
the weak prime b-domains and the obvious rephrasing of
Proposition~\ref{pr:es-int} continue to hold.

Also the connection with asynchronous graphs in
Section~\ref{ss:async-graphs} can be adapted easily. Unsurprisingly,
Proposition~\ref{pr:was-ces} holds for connected {\esabbr} with
non-binary conflict if we modify the definition of asynchronous graph
(Definition~\ref{de:async-graph}) by omitting the coherence
axiom~(\ref{de:async-graph:4}).

\subsection{An Event Structure Semantics for the $\pi$-calculus}
\label{app:pi}

The need of resorting to unstable {\esabbr} for modelling the
concurrent computations of name passing process calculi has been
observed by several authors. In particular, in~\cite{CVY:ESSPE} an
{\esabbr} semantics for the pi-calculus is defined by relying on
so-called {\esabbr} \emph{with names} ({\esnmabbr} for short), namely
{\esabbr} that are tailored for parallel extrusions: labelled unstable
{\esabbr} with the constraint that two minimal enablings can differ
only for one event (intuitively, the extruder).

Given a global set of names $\mathcal{N}$, {\esabbr} \emph{with names}
({\esnmabbr} for short) are triples $(E, X, \lambda)$ where $E$ is a
prime {\esabbr}, $X \subseteq \mathcal{N}$ is a set of names
(intuitively, the names that are restricted), and
$\lambda : E \to \{ x(y), \bar{x}(y) \}$ is a function mapping each event
to either an input or an output prefix.

A configuration $C$ is a configuration of the underlying prime
{\esabbr} such that there exists a maximal $e \in C$
satisfying
\begin{itemize}
\item $C \setminus \{ e \}$ is a configuration;
\item if $\lambda(e) = x(y)$ or $\lambda(e) = \bar{x}(y)$ with
  $x \in X$ then there exists $e' \in C \setminus \{e\}$ such that
  $\lambda(e') = \bar{z}(x)$.
\end{itemize}

The latter requirement above boils down to ensuring that if the name were 
restricted, it has been extruded before.
Thus, {\esnmabbr} are unstable {\esabbr} with the additional constraint 
that two minimal enablings can differ only for one event (the extruder!):
namely, if $X_1 \vdash_0 e$ and $X_2 \vdash_0 e$ then
$X_1 \cap X_2 = X_1 \setminus \{ e_1 \} = X_2 \setminus \{ e_2 \}$
for suitable $e_1, e_2$.

Note that, {\esnmabbr} are not connected {\esabbr} since they can have 
non-connected minimal enablings (roughly, because identical events in 
disconnected minimal enablings are identified via the labelling).
Consider e.g. $E = \{ \bar{a}(x), \bar{b}(x), x(y) \}$, with
$\bar{a}(x) \# \bar{b}(x)$, and $X=\{x\}$. Then the
configurations are $\emptyset$, $\{\bar{a}(x)\}$, $\{\bar{b}(x)\}$,
$\{\bar{a}(x), x(y)\}$, and $\{\bar{b}(x), x(y)\}$, hence $x(y)$ has two
non-connected minimal enablings. 

\end{document}